\documentclass[]{article}

\RequirePackage{amsthm,amsmath,amsfonts,amssymb}
\RequirePackage[numbers]{natbib}
\usepackage{authblk}

\usepackage{graphicx,psfrag,epsf}
\usepackage{enumerate}
\usepackage{algorithm,algorithmic}
\usepackage{multirow}
\usepackage{subcaption}
\usepackage[margin = 1in]{geometry}
\usepackage{url}

\theoremstyle{plain}
\newtheorem{theorem}{Theorem}

\newtheorem{lemma}{Lemma}

\theoremstyle{definition}

\newtheorem{property}{Property}

\newcommand{\bY}{\boldsymbol{Y}}

\newcommand{\bX}{\boldsymbol{X}}

\newcommand{\bZ}{\boldsymbol{Z}}

\newcommand{\bv}{\boldsymbol{v}}
\newcommand{\beps}{\boldsymbol{\epsilon}}
\newcommand{\bSigma}{\boldsymbol{\Sigma}}

\newcommand{\btheta}{\boldsymbol{\theta}}

\newcommand{\bdelta}{\boldsymbol{\delta}}

\newcommand{\boldeta}{\boldsymbol{\eta}}

\begin{document}

\title{Estimation of Parameters of the Truncated Normal Distribution with Unknown Bounds}
\author[1]{Dylan Borchert\thanks{email: dylan.borchert@jacks.sdstate.edu}}
\author[2]{Semhar Michael}
\author[2] {Christopher Saunders} 
\affil[1]{Department of Mathematical Sciences\\
New Mexico State University} 
\affil[2]{Department of Mathematics and Statistics\\ South Dakota State University}


\maketitle

\begin{abstract}
Estimators of parameters of truncated distributions, namely the truncated normal distribution, have been widely studied for a known truncation region. There is also literature for estimating the unknown bounds for known parent distributions. In this work, we develop a novel algorithm under the expectation-solution (ES) framework, which is an iterative method of solving nonlinear estimating equations, to estimate both the bounds and the location and scale parameters of the parent normal distribution utilizing the theory of best linear unbiased estimates from location-scale families of distribution and unbiased minimum variance estimation of truncation regions. The conditions for the algorithm to converge to the solution of the estimating equations for a fixed sample size are discussed, and the asymptotic properties of the estimators are characterized using results on M- and Z-estimation from empirical process theory. The proposed method is then compared to methods utilizing the known truncation bounds via Monte Carlo simulation.
\end{abstract}



\section{Introduction}
\label{sec:intro}

Estimating the parameters of the truncated normal distribution for fixed truncation bounds is a classical problem \cite{Pearson-on_the_generalized}. Different techniques have been utilized to solve this problem when the truncation bounds or number of censored or truncated samples are known including maximum likelihood estimation \cite{ Cohen-estimating, crain1979, grundy1952, Cohen-on_estimating, Mittal1987}, method of moments \cite{Shah-estimation}, and unbiased estimation \cite{Gupta-estimation} using the linear regression on the order statistics from a location-scale family of distributions \cite{Lloyd-leastsquares}. Recently, estimating the parameters of a truncated normal random effect under symmetric truncation has been explored \cite{Chen-estimating}. These works present methods for mean testing\cite{canonne-gaussian}, maximum likelihood estimation with known truncation \cite{Bhattacharyya-Learning, Daskalakis-efficient}, and moment matching with unknown truncation \cite{Kontonis-Efficient}, as well as a method for estimating the indicator function of the truncation set \cite{Kontonis-Efficient}. Similarly, estimating truncation bounds is another classical problem with work in the area of bias correction of the MLE \cite{Robson-estimation-1964} and uniform minimum variance unbiased estimation \cite{morimoto1967}. In this work, we wish to estimate all four parameters of the truncated normal distribution, namely the mean and standard deviation of the parent normal distribution, as well as the unknown truncation bounds.

The motivation for estimating all four parameters of the truncated normal distribution stems from the identification of source problems within forensic science. In this problem, the task is to give a summary of evidence that allows a decision maker whether two sets of observations came from the same class or two different classes \cite{Ommen-building}. This task can be seen as a non-nested model selection or an open set classification problem. Specifically within the forensic science community, there has been a push for providing this summary of evidence using a likelihood ratio or a Bayes factor \cite{Aitken-practitioner-2010, ENFSI-guidelines}. One way to do this is to set up a random effects model, decomposing the observations into a source-level effect arising from a source-level distribution and a within-source effect or error term coming from a within-source distribution \cite{Ommen-building, Ommen-aproblem}. A potential drawback to this approach stems from the case where there are observations from two sources that are relatively far apart from each other, either in the tails of the between-source distribution or from a subpopulation that was under-represented in the training dataset, which is far away from the bulk of the training data sources. Though the lens of the data seen, it is very unlikely these two sets of observations are from the same source. It is even more unlikely that two rare sources would be selected from the between-source distribution. This is one consequence of Lindley's paradox \cite{Shafer-lindley1982}, where a ratio of small (marginal) likelihoods can still return a very small or very large likelihood ratio or Bayes factor. 

There was a lively discussion about Lindley's paradox following Shafer's reply \cite{Shafer-lindley1982}, to Lindley's paper \cite{lindley-aproblem} that proposed a Bayes factor for a problem in forensic science related to an open set classification problem applied to refractive index glass data. In this discussion, Degroot stated, 
\begin{quote}
    All good Bayesian statisticians reserve a little pinch of probability for the possibility that their model is wrong. Usually this possibility does not have to be acknowledged formally in the statistical analyses that are carried out, but every now and then, when an unusual value of Y is observed, it is important to reconsider the model and to evaluate that little pinch of probability \citep{DeGroot-Lindley}.
\end{quote}
We propose truncating the within-source distribution as one way of incorporating this ``pinch of probability." This would allow for the likelihood that two sets of observations from rare sources become zero if the two sets of observations are far enough from each other in some sense. To begin developing this framework, we first consider studying the closed set classification problem under the frequentist setting, where classes are modeled using the truncated normal distribution. This allows for a new observation to have zero likelihood in all the known classes if it is too far from the exemplars in these classes. To do this, the parameters of the truncated normal distribution, both the mean and standard deviation of the parent normal distribution, as well as the truncation bounds. 

In this paper, we propose estimators of the parameters of the truncated normal distribution, when the truncation bounds are unknown, as the solution to a system of equations set up from the theory of regression on the order statistics of a location-scale family of distributions and uniform minimum variance unbiased estimation of truncation bounds for a known parent distribution. The estimators are found using a numerical solution to the system of equations found with a proposed algorithm under the Expectation-Solution (ES) scheme, which is a generalization of the Expectation-Maximization (EM) algorithm \cite{Dempster-EM}. The EM algorithm is commonly used to provide maximum likelihood estimates of parameters of a distribution when given truncated or censored samples, where the number of truncated/censored samples and/or the truncation/censoring regions are known \cite{mclachlan-EM}. 

The paper is organized as follows. Section \ref{sec:methods} sets up the notation, defines a system of equations to find estimators of the parameters of the truncated normal distribution, and proposes an algorithm to numerically solve this system on equations. Section \ref{sec:fixed_convergence} establishes that the algorithm iterates to a solution of the system of equations. Section \ref{sec:consistency} shows consistency of the resulting estimators and gives the rates of convergence. Section \ref{sec:sims} provides Monte Carlo simulation studies to visualize the large sample behavior of the estimators and gives a brief application of the estimators of the truncated normal distribution applied to a closed set classification problem using the Iris data. Finally, Section \ref{sec:conclusion} concludes the paper.

\section{Truncated Normal Distribution and ES Algorithm}
\label{sec:methods}

Let $X_{1}, \dots X_{n} \overset{\mathrm{iid}}{\sim} \text{TN}(\mu,\sigma,\tau_l,\tau_u)$, the truncated normal distribution with parameters $\btheta = (\mu, \sigma^2, \tau_l, \tau_u)'$, where $\mu$ is a real number, $\sigma >0$, $\tau_l < \tau_u$ with distribution function $F_{\btheta}$. That is the distribution with density
$$f_{\btheta}(t) = \frac{\frac{1}{\sigma}\phi(\frac{t-\mu}{\sigma})}{\Phi(\frac{\tau_u-\mu}{\sigma}) - \Phi(\frac{\tau_l-\mu}{\sigma})}, t \in [\tau_l, \tau_u ],$$ 
where $\phi$ and $\Phi$ are the density and distribution function of the standard normal distribution, respectively.

Define the order statistic $\bX_{(n)} = (X_{n:1}, X_{n:2}, \dots, X_{n:n})'$ to be the vector of ordered observations where $X_{n:k}$ is the $k^{th}$ ordered observation from our sample of size $n$. Let $\tau_l^* = \frac{\tau_l-\mu}{\sigma}$, and $\tau_u^* = \frac{\tau_u - \mu}{\sigma}$. The goal of this work is to provide estimates $\mu, \sigma, \tau_l$ and $\tau_u$ and characterize their asymptotic properties. 

The proposed algorithm will follow the expectation-solution (ES) algorithm framework \cite{Elashhoff2004}. The ES algorithm is a method for providing parameter estimates that solve a system of estimating equations involving a latent variable. The estimating equations used for the algorithm will be discussed one at a time before stating the full system of equations and the formal statement of the ES algorithm.

\subsection{Estimating Equations for Parent Normal Distribution Parameters}
 This algorithm will hinge on an alternative sampling model for a sample of size $n$ from a truncated normal distribution, where observations are repeatedly sampled from the parent normal distribution until $n$ observations lie within the truncation bounds. If we had kept track of the number of observations that fell below $\tau_l$, $n_l$, and the number of observations that were above $\tau_u$, $n_u$, we can think of our vector of order statistics $\bX_{(n)}$ as a ``middle portion" of a larger vector of order statistics from a normal distribution. We will denote this as $\bY_{(n)} = (X_{n_l + n_u + n: n_l + 1}, \dots, X_{n_l + n_u + n: n_l + n})$. Thus, our $X_{n:k}$ lines up with $Y_{(n_l + n_u + n):(n_l+k)}$ from the larger sample from the full normal distribution. We can then set up the regression model
\begin{equation}
   X_{n:k} = \mu + \sigma \mathbb{E}(Z_{(n_l + n_u + n):(n_l+k)}) + \epsilon_k \label{eqn:reg},
\end{equation} where $\mathbb{E}(\beps) = 0$ and $\mathbb{V}(\beps) = \bSigma_\epsilon$. For ease of computation, we can replace $\mathbb{E}(Z_{(n_l + n_u + n):(n_l+k)})$ with $\xi(n, n_l,n_u,k)$ which is the median of the $k^{th}$ order statistic \cite{Filliben-probabilityplot} and use the identity matrix in place of $\bSigma_\epsilon$ \cite{Gupta-estimation, Shapiro-anapproximate}. We have $\xi(n, n_l,n_u,k) = \Phi^{-1}(G_{n, n_l, n_u, k}^{-1}(0.5))$ where $G_{n, n_l, n_u, k}^{-1}$ is quantile function of the $\text{Beta}(n_l + k, n_u + n + 1 - k)$ distribution. It is worthwhile to note that $\xi(n, n_l,n_u,k)$ is defined for non-integer values of $n_l$ and $n_u$. With this set-up, ordinary least squares can be used to obtain estimates of $\mu$ and $\sigma$. 

\subsection{The Distributions of Latent Variables}

The regression model in \eqref{eqn:reg} depends on the realizations of the latent variables $N_l$ and $N_u.$ We can characterize the distributions of $N_l$ and $N_u$ with the following negative binomial distributions
\begin{align*}
    &N_l \sim \text{NB}\left(n, \frac{\Phi(\tau_u^*) - \Phi(\tau_l^*)}{\Phi(\tau_u^*)}\right) \quad \text{and}\\
    &N_u \sim \text{NB}\left(n, \frac{\Phi(\tau_u^*) - \Phi(\tau_l^*)}{1 - \Phi(\tau_l^*)}\right).
\end{align*}
with the parametrization of $K \sim NB(r,\pi)$ defined as $K$ failures observed in order to get $r$ successes with probability of success equal to $\pi$. We then have the following expectations,
\begin{align}
    &h_l(\btheta,n) \equiv \mathbb{E}_{\btheta}(N_l) = n \left( \frac{\Phi(\tau_l^*)}{\Phi(\tau_u^*) - \Phi(\tau_l^*)}  \right) \equiv n c_l \quad \text{and} \label{eqn:expectations1}  \\
    &h_u(\btheta,n) \equiv \mathbb{E}_{\btheta}(N_u) = n \left( \frac{1 - \Phi(\tau_u^*)}{\Phi(\tau_u^*) - \Phi(\tau_l^*)} \right) \equiv n c_u. \label{eqn:expectations2}
\end{align}

\subsection{Estimating Equations for Bound Parameters}

The uniform minimum variance unbiased estimates of $\tau_l$ and $\tau_u$ for known $\mu$ and $\sigma$ \cite{morimoto1967} are 
\begin{align}
    &\hat{\tau}_l = X_{n:1} - \frac{\Phi(W_{n:n}) - \Phi(W_{n:1})}{(n-1)\frac{1}{\sigma}\phi(W_{n:1})} \label{eqn:bounds1} \quad \text{and} \\
    &\hat{\tau}_u = X_{n:n} + \frac{\Phi(W_{n:n}) - \Phi(W_{n:1})}{(n-1)\frac{1}{\sigma}\phi(W_{n:n})}, \label{eqn:bounds2}
\end{align}
where $W_i = \frac{X_i - \mu}{\sigma}$.

\subsection{ES Algorithm} 
\label{sec:ES_alg}

Following notation for the ES algorithm \cite{Elashhoff2004}, we can now define the algorithm. Our ``complete" data is $X^{c} = (X^{obs}, X^{lat})^t$ where our observed data is $X^{obs} = \bX_{(n)}$ and the latent observations are $X^{lat} = (N_l, N_u)^t.$ This leads to the ``complete data estimating equations"
\begin{align} \label{eqn:completeest}
U^c(X^c, \btheta) &= U^{(1)}(X^{obs}, N_l, N_u, \btheta) \\ &=  \begin{pmatrix} n^{-1} \sum_{i=1}^n (X_{n:i} - \mu - \sigma \xi(n, N_l,N_u,k)) \\
n^{-1} \sum_{i=1}^n (X_{n:i} - \mu - \sigma \xi(n, N_l, N_u, k))\xi(n, N_l, N_u,k) \\
X_{n:1} - \frac{\Phi(W_{n:n}) - \Phi(W_{n:1})}{(n-1)\frac{1}{\sigma}\phi(W_{n:1})} - \tau_l \\
X_{n:n} + \frac{\Phi(W_{n:n}) - \Phi(W_{n:1})}{(n-1)\frac{1}{\sigma}\phi(W_{n:n})} - \tau_u
\end{pmatrix}.
\end{align} 
If we had access to the latent variables $N_l$ and $N_u$ we could get parameter estimates by solving 
$U^c(X^c, \btheta) = \boldsymbol{0}$ for $\btheta.$ As these values will not be known  when presented with the observed data, they can be replaced with their expectations under the negative binomial sampling model to get the observed estimating equations 
\begin{align} \label{eqn:obsest}
U^{obs}(X^{obs}, \btheta) &= U^{(1)}(X^{obs}, h_l(n,\btheta), h_u(n,\btheta), \btheta) \\
&=  \begin{pmatrix} n^{-1} \sum_{i=1}^n (X_{n:i} - \mu - \sigma c_{n:i}(\btheta)) \\
n^{-1} \sum_{i=1}^n (X_{n:i} - \mu - \sigma c_{n:i}(\btheta))c_{n:i}(\btheta) \\
X_{n:1} - \frac{\Phi(W_{n:n}) - \Phi(W_{n:1})}{(n-1)\frac{1}{\sigma}\phi(W_{n:1})} - \tau_l \\
X_{n:n} + \frac{\Phi(W_{n:n}) - \Phi(W_{n:1})}{(n-1)\frac{1}{\sigma}\phi(W_{n:n})} - \tau_u
\end{pmatrix},
\end{align} 
where $c_{n:k}(\btheta) = \xi(n, h_l(n,\btheta),h_u(n,\btheta),k).$
Previous work in the area \cite{Elashhoff2004} used an assumption that $U^{obs}(X^{obs}, \btheta)$ could be expressed as a sum of $n$ independent terms, that is $U^{obs}(X^{obs}, \btheta) = \sum_{i=1}^n U^{obs}_i(\btheta)$. Note that in the system of equations above does not satisfy this assumption as the system of equations is a linear weighted combination of transformed order statistics. 

The ES algorithm aims to solve $U^{obs}(X^{obs}, \btheta) = \boldsymbol{0}$ for $\btheta$. After proposing an initial value of $\btheta$, denoted $\btheta^{(0)}$, we iterative between two steps:
\begin{itemize}
    \item E-step: Calculate $n_l^{(t)} = h_l(n, \btheta^{(t-1)})$ and $n_u^{(t)} = h_u(n, \btheta^{(t-1)})$
    \item S-step: Solve $U^{(1)}(X^{obs}, n_l^{(t)}, n_u^{(t)}, \btheta)$ to obtain $\btheta^{(t)}$
\end{itemize}
until some convergence criterion is met.

\section{Numerical Properties of the Algorithm} \label{sec:fixed_convergence}

One property of the proposed algorithm to provide trustworthy parameter estimates of the truncated normal model is the convergence of the algorithm to a solution of $U^{obs}(X^{obs}, \btheta) = \boldsymbol{0}.$ We will follow a similar approach to that in the literature \cite{Elashhoff2004} to establish convergence to a root of $U^{obs}(X^{obs}, \btheta)$ using theory for the Block Gauss-Seidel Process \cite{OrtegaReinboldt-1970}. Define
\begin{equation} \label{eqn:u2}
    U^{(2)}(X^{obs}, n_l, n_u, \btheta) = \begin{pmatrix}
        n_l - h_l(\btheta, n) \\
        n_u - h_u(\btheta, n)
    \end{pmatrix}.
\end{equation}
Then defining a new system of equations combining \eqref{eqn:completeest} and \eqref{eqn:u2}
\begin{equation} \label{eqn:combined u}
    U(X^{obs},\bdelta) = \begin{pmatrix}
        U^{(1)}(X^{obs}, n_l, n_u, \btheta) \\
        U^{(2)}(X^{obs}, n_l, n_u, \btheta)
    \end{pmatrix},
\end{equation}
where $\bdelta = (\btheta, n_l, n_u)^t$. Note that a solution, $\tilde{\bdelta} = (\tilde{\btheta}, \tilde{n}_l, \tilde{n}_u)^t$ satisfying $U(X^{obs},\tilde{\bdelta})=0$ also satisfies $U^{obs}(X^{obs}, \tilde{\btheta}) = 0$, which can be seen by a backwards substitution. Now, consider the differencing function
\begin{equation} \label{eqn:diff_fun}
    g(\bdelta_1, \bdelta_2) = \begin{pmatrix}
    U^{(1)}(X^{obs}, n_{l2}, n_{u2}, \btheta_1) \\
        U^{(2)}(X^{obs}, n_{l1}, n_{u1}, \btheta_1)
\end{pmatrix}.
\end{equation}

It can be seen that an iteration of the proposed algorithm gives $\bdelta^{(k+1)}$ by solving $g(\bdelta, \bdelta^{(k)}) = 0$ for $\bdelta.$ An interesting result is that the partial derivative of $g$ with respect to the first argument is invertible as 
\begin{align*}
   \det(D_{\bdelta_1}g) &= \det\left(\begin{pmatrix}
        \frac{\partial U^{(1)}(X^{obs}, n_{l2}, n_{u2}, \btheta_1)}{d \btheta_1} & \boldsymbol{0} \\
        \frac{\partial U^{(2)}(X^{obs}, n_{l1}, n_{u1}, \btheta_1)}{d \btheta_1} & \frac{\partial U^{(2)}(X^{obs}, n_{l1}, n_{u1}, \btheta_1)}{d (n_{l1}, n_{u1})^t}
    \end{pmatrix}\right)  \\
    &= \det\left(\frac{\partial U^{(1)}(X^{obs}, n_{l2}, n_{u2}, \btheta_1)}{d \btheta_1}\right) \det\left(\frac{\partial U^{(2)}(X^{obs}, n_{l1}, n_{u1}, \btheta_1)}{d (n_{l1}, n_{u1})^t}\right)
\end{align*}
by properties of block matrices. Now, with the form of $U^{(2)}$, it can be seen that $\frac{\partial U^{(2)}}{d (n_{l1}, n_{u1})^t}$ is an identity matrix. Thus, evaluating $\frac{\partial U^{(1)}}{d \btheta }$ it follows that 
\begin{align*}
    \det(D_{\bdelta_1}g) &= \det\left(\begin{pmatrix}
        -1 & -n^{-1}\sum_{i=1}^n \xi(n, n_{l2},n_{u2},i) \\
        -n^{-1}\sum_{i=1}^n \xi(n, n_{l2},n_{u2},i) & -n^{-1}\sum_{i=1}^n \xi(n, n_{l2},n_{u2},i)^2
    \end{pmatrix} \right) \\
    &= n^{-1}\left(\sum_{i=1}^n \xi(n, n_{l2},n_{u2},i)^2 -n^{-1}\left(\sum_{j=1}^n \xi(n, n_{l2},n_{u2},j)\right)^2\right)  > 0.
\end{align*}

To establish convergence of the ES algorithm to a root of $U^{obs}(X^{obs}, \btheta)$, a result and a general proof technique from \cite{Ortega-difference1966} is utilized. In this proof, one step of the ES algorithm is viewed as taking a step in finding a solution to the differencing function $g$. The implicit function theorem will then be used to find a function in $C^2$, which can be used to characterize the difference between the current value of the parameter and the true estimator, as the root of $U^{(1)}$. Taylor's Theorem can then be used to show that, in the limit, each step gives a value closer to the true value of the parameter.

The formal statement of these results is described in Theorem \ref{thm:iterate_convergence}. To start, consider the following lemma. which will be used to characterize the remainder of the Taylor approximation.
\begin{lemma}[Result from \cite{Ortega-difference1966}]
\label{lemma:remainder}
    Let $H$ be a $k\times k$ matrix with spectral radius $\rho(H) \equiv \lambda < 1$ and let $r:\mathcal{M} \subset \mathbb{R}^k \rightarrow \mathbb{R}^k$ be a mapping from a neighborhood of the origin, $\mathcal{M}$, into $\mathbb{R}^l$ such that 
$$\frac{||r(x)||}{||x||} \rightarrow 0, \text{  as  } ||x|| \rightarrow 0.$$
Then given any constant $\epsilon$ satisfying $\lambda < \lambda+\epsilon<1$, there is an open neighborhood of the origin $\mathcal{M}'$ and a positive constant $d$ such that for any initial vector $e^{(0)} \in \mathcal{M}'$ a solution $e^{(k)}$ of $e^{(k)} = He^{(k-1)} + r(e^{(k-1)})$ exists and satisfies 
$$||e^{(k)}|| \le d(\lambda+\epsilon)^k ||e^{(0)}||.$$
\end{lemma}

\begin{theorem}
\label{thm:iterate_convergence}
    Let $\tilde{\bdelta}$ be a solution to $U(X^{obs},\tilde{\bdelta})=0$ and define $g(\bdelta_1, \bdelta_2)$ as in \eqref{eqn:diff_fun}. If $\rho(-D_{\bdelta_1}g(\tilde{\bdelta}, \tilde{\bdelta})^{-1} D_{\bdelta_2}g(\tilde{\bdelta}, \tilde{\bdelta})) < 1$ and $g\in C^2$, then there exists a neighborhood $\mathcal{N}$ of $\tilde{\bdelta}$ such that $\bdelta^{(0)}\in \mathcal{N}$ implies that the proposed algorithm iterates towards $\tilde{\bdelta}.$
\end{theorem}

\begin{proof}
As $\tilde{\bdelta}$ satisfies $U(X^{obs},\tilde{\bdelta})=0$ it also satisfies $g(\tilde{\bdelta},\tilde{\bdelta}) = 0.$ By the implicit function theorem \cite{Apostol-1975, Dieudonne-foundations}, there exists a neighborhood $\mathcal{N}$ of $\tilde{\bdelta}$ and a unique continuous function $u(\bdelta)$ defined on $\mathcal{N}$ such that $u(\tilde{\bdelta}) = \tilde{\bdelta}$ and $(u(\bdelta), \bdelta)^t$ satisfies $g(u(\bdelta), \bdelta) = 0$ for all $\bdelta \in \mathcal{N}.$

By Taylor's Theorem 
$u(\bdelta) = u(\tilde{\bdelta}) + Du(\tilde{\bdelta})(\bdelta - \tilde{\bdelta}) + r(\bdelta - \tilde{\bdelta}).$ As g is twice continuously differentiable $u$ is also twice continuously differentiable \cite{Dieudonne-foundations} so, $||r(\bdelta - \tilde{\bdelta})|| \le c||\bdelta - \tilde{\bdelta})||^2$ for some $c>0$ and all $\bdelta \in \mathcal{N}.$ 
Now as $\bdelta^{(k+1)}$ satisfies $g(\bdelta^{(k+1)}, \bdelta^{(k)})=0,$ it follows that $\bdelta^{(k+1)} = u(\bdelta^{(k)}),$ and so 
$$u(\bdelta^{(k)}) = u(\tilde{\bdelta}) + Du(\tilde{\bdelta})(\bdelta^{(k)} - \tilde{\bdelta}) + r(\bdelta^{(k)} - \tilde{\bdelta})$$
yielding
$$\bdelta^{(k+1)} - \tilde{\bdelta} = Du(\tilde{\bdelta})(\bdelta^{(k)} - \tilde{\bdelta}) + r(\bdelta^{(k)} - \tilde{\bdelta}).$$
Using Lemma \ref{lemma:remainder}, there is a neighborhood $\mathcal{M}$ such that if $\bdelta^{(0)} \in \mathcal{M}$ then 
$$||\bdelta^{(k+1)} - \tilde{\bdelta}|| \le a b^k ||\bdelta^{(0)} - \tilde{\bdelta}||$$ for $a >0$ and $0 < b < 1.$ Therefore, if $\bdelta^{(0)} \in \mathcal{M}\cap\mathcal{N}$, the algorithm has attraction point $\tilde{\bdelta}.$
\end{proof}
Thus, a sufficient condition for the proposed ES algorithm to converge to a root of $U^{obs}(X^{obs}, \btheta)$ is that $\rho(-D_{\bdelta_1}g(\tilde{\bdelta}, \tilde{\bdelta})^{-1} D_{\bdelta_2}g(\tilde{\bdelta}, \tilde{\bdelta})) < 1$ and $g\in C^2$. The latter is true as the system of equations is a combination of continuous functions element-wise. The first condition will be assumed to be true. The value of the objective function can be evaluated at the solution after the fact to ensure the resulting estimators are indeed a (near) solution to the desired system of equations.

\section{Asymptotic Properties of Estimators} \label{sec:consistency}

\subsection{Consistency}
A desirable property of any estimator is convergence to the parameters of the generative model under some mode of convergence as the sample size tends to infinity. The following results show that when a sample is from a truncated normal distribution, the proposed estimators of the truncated normal distribution provide consistent estimates, as convergence in probability, of the parameters of the truncated normal distribution that generated the data. The proof will use the theory of $Z$-estimators \cite{vdVaart_1998}. From a random sample $X_1,\dots, X_n \sim \text{TN}(\mu,\sigma,\tau_l,\tau_u)$ define the random function $\Psi_n(\btheta) = U^{obs}((X_{n:1},\dots,X_{n:n})^t, \btheta)$. To prove consistency of estimator $\hat{\btheta}_n$, the solution to $\Psi_n(\btheta)=0,$ we we use a result in $Z$-estimation. In the proof the limiting form of $\Psi_n(\btheta)$, $\Psi(\btheta)$, will be needed. To derive $\Psi(\btheta)$, a result on $L$-estimators will be utilized \cite{Shorackandwellner}. The result on $L$-estimators is presented in the following Lemma for completeness.
\begin{lemma}[Result from \cite{Shorackandwellner}] \label{lemma:L_stat_shorackandwellner}
    Let $X_1, \dots, X_n$ be an iid sample from a distribution with distribution function $F$. Define $T_n = \frac{1}{n} \sum_{i=1}^n c_{ni}h(X_{n:i})$ and $J_n(t) = c_{ni}$ for $\frac{i-1}{n} < t \le \frac{i}{n}$ with $J_n(0) = c_{n1}.$ If there is a function $J$ on $[0,1]$ such that $|J| < M$ and $|J_n|<M$ for $M>0$, $h$ is an increasing left continuous function satisfying $|h(F^{-1})| \le M$, $c_{ni} = J(t)$ for some $t\in [\frac{i-1}{n}, \frac{i}{n}], 1\le i\le n$ and $J$ is Lipschitz on $[0,1]$, then $$\sqrt{n}(T_n - \mu(J,F)) \rightsquigarrow N(0, \sigma^2(J,F)),$$ where $$\mu(J,F) = \int_0^1 h(F^{-1}(t))J(t)dt$$ and $$\sigma^2(J,F) = \int_{-\infty}^{\infty}\int_{-\infty}^{\infty} (F(\min(x,y))-F(x)F(y))J(F(x))J(F(y))dxdy.$$
\end{lemma}
Now, with Lemma \ref{lemma:L_stat_shorackandwellner}, the limiting form of $\Psi_n(\btheta)$ can be established, which will be used to establish the consistency of the proposed estimators.
\begin{lemma}
\label{lemma:limiting_system}
    Define the random function $\Psi_n(\btheta) = U^{obs}((X_{n:1},\dots,X_{n:n})^t, \btheta).$ If $$X_1,\dots, X_n \sim \text{TN}(\mu_0,\sigma_0,\tau_{l0},\tau_{u0})$$ and $\boldsymbol{\Theta}$ is compact, then $\sup_{\btheta \in \Theta}|\Psi_n(\btheta) - \Psi(\btheta) |\overset{p}{\rightarrow}0,$ where 
\begin{equation}
    \Psi(\btheta) = \begin{pmatrix}
        \mu_0 + \sigma_0\alpha_{01} - \mu -\sigma\alpha_1 \\
        \mu(J_{\btheta},F_{\btheta_0}) - \mu\alpha_1 - \sigma\alpha_2 \\
        \tau_{l0} - \tau_l \\
        \tau_{u0} - \tau_u
    \end{pmatrix},
\end{equation} 
$J_{\btheta}, \alpha_1$, and $\alpha_2$ are defined in Appendix \ref{app:standard_truncated}.
\end{lemma}

\begin{proof}
By Property 4 in Appendix \ref{app:standard_truncated}, $\mathbb{E}(X_i) = \mu_0 + \sigma_0\alpha_{01}.$ Also, by property 2, in Appendix \ref{app:standard_truncated}, shows that 
\begin{align*}
    & \frac{1}{n}\sum_{k=1}^n \xi(n, h_l(n,\btheta), h_u(n,\btheta),k) \rightarrow \int_0^1 J_{\btheta}(p)dp \quad \text{and} \\
    & \frac{1}{n}\sum_{k=1}^n \xi(n, h_l(n,\btheta), h_u(n,\btheta),k)^2 \rightarrow \int_0^1 J^2_{\btheta}(p)dp, \text{ as } n\rightarrow \infty
\end{align*} 
uniformly over $\boldsymbol{\Theta}.$ Also, the strong law of large numbers gives
\begin{equation}
    \sum_{k=1}^n X_{n:k} \overset{wp1}{\rightarrow} \mu_0 + \sigma_0\alpha_{01}
\end{equation}
By properties 1 and 2 in Appendix \ref{app:standard_truncated}, we can apply Lemma \ref{lemma:L_stat_shorackandwellner} using $h(t)=t$ to obtain $$\sqrt n (\frac{1}{n}\sum_{k=1}^n X_{n:k}\xi(n, h_l(n,\btheta), h_u(n,\btheta),k) - \mu(J_{\btheta}, F))\rightsquigarrow N(0, \sigma^2(J_{\btheta}, F)).$$
Looking at the variance term,
\begin{align*}
    \sigma^2(J_{\btheta},F) &= \int_{-\infty}^{\infty}\int_{-\infty}^{\infty} (F(\min(x,y))-F(x)F(y))J_{\btheta}(F(x))J_{\btheta}(F(y))dxdy \\
    &\le  \int_{\tau_{l0}}^{\tau_{u0}}\int_{\tau_{l0}}^{\tau_{u0}} |(F(\min(x,y))-F(x)F(y))| |J_{\btheta}(F(x))J_{\btheta}(F(y))|dxdy \\
\end{align*}
Now, $F(x) \in [0,1]$, which gives that $F(t_1) (1-F(t_2)) \le \frac{1}{4}.$ Also, for compact $\boldsymbol{\Theta}$ $|J_{\btheta}| \le M$ and thus, 
\begin{align*}
    \sigma^2(J_{\btheta},F) \le  \frac{1}{4} M^2 (\tau_{u0} - \tau_{l0})^2,
\end{align*}
which does not depend on $\btheta$ only $\btheta_0.$ Thus $(\frac{1}{n}\sum_{k=1}^n X_{n:k}\xi(n, h_l(n,\btheta), h_u(n,\btheta),k) - \mu(J_{\btheta}, F))$ is $O_p(n^{-1/2})$ not dependent on $\btheta.$ So,
\begin{equation}
\begin{split}
    &n^{-1} \sum_{i=1}^n (X_{n:i} - \mu - \sigma \xi(n, h_l(n,\btheta), h_u(n,\btheta), k))\xi(n, h_l(n,\btheta), h_u(n,\btheta),k) \\
    &\overset{p}{\rightarrow} \mu(J_{\btheta}, F) - \mu \int_0^1 J_{\btheta}(p)dp - \sigma \int_0^1 J^2_{\btheta}(p)dp, \text{ as } n\rightarrow \infty.
\end{split}
\end{equation}

Finally, the third and fourth elements of $\Psi_n(\btheta),$
we have that there exists a positive number $L$ such that $\max\{|W_{n:1}|, |W_{n:2}|\}\le L$ as $\Theta$ is bounded. Thus $\frac{\Phi(W_{n:n}) - \Phi(W_{n:1})}{(n-1)\frac{1}{\sigma}\phi(W_{n:1})} = o_P(1)$ and $\frac{\Phi(W_{n:n}) - \Phi(W_{n:1})}{(n-1)\frac{1}{\sigma}\phi(W_{n:n})} = o_P(1)$ as $n\rightarrow \infty.$ Thus, as the truncated normal distribution is bounded above and below, we have 
\begin{align*}
    & X_{n:1} - \frac{\Phi(W_{n:n}) - \Phi(W_{n:1})}{(n-1)\frac{1}{\sigma}\phi(W_{n:1})} - \tau_l \rightarrow \tau_{l0} - \tau_l \quad \text{and}\\
    & X_{n:n} + \frac{\Phi(W_{n:n}) - \Phi(W_{n:1})}{(n-1)\frac{1}{\sigma}\phi(W_{n:n})} - \tau_u \rightarrow \tau_{u0} - \tau_u,
\end{align*}
as $n\rightarrow \infty.$ Therefore, $\sup_{\btheta \in \Theta}|\Psi_n(\btheta) - \Psi(\btheta) |\overset{p}{\rightarrow}0,$ where 
\begin{equation*}
    \Psi(\btheta) = \begin{pmatrix}
        \mu_0 + \sigma_0\alpha_{01} - \mu -\sigma\alpha_1 \\
        \mu(J_{\btheta},F_{\btheta_0}) - \mu\alpha_1 - \sigma\alpha_2 \\
        \tau_{l0} - \tau_l \\
        \tau_{u0} - \tau_u
    \end{pmatrix}.
\end{equation*} 
\end{proof}
With the limiting form of the system of equations, the main result on $Z$-estimators \citep{vdVaart_1998} is given in Lemma \ref{lemma:z_est_consistency} for completeness before stating the main theorem establishing consistency of the estimators of the truncated normal distribution parameters produced by the given ES algorithm.

\begin{lemma}[Result from \citep{vdVaart_1998})]
\label{lemma:z_est_consistency}
Let $\Psi_n$ be random vector-valued functions and let $\Psi$ be a fixed vector valued function of $\btheta$ such that for every $\epsilon>0$ 
$$\sup_{\btheta \in \Theta}||\Psi_n(\btheta) - \Psi(\btheta)||\overset{P}{\rightarrow}0,$$
$$\inf_{\btheta:d(\btheta,\btheta_0)\ge\epsilon}||\Psi(\btheta)||>0=||\Psi(\btheta_0)||.$$
Then any sequence of estimators $\hat{\btheta}_n$ such that $\Psi_n(\hat{\btheta}_n) = o_P(1)$ converges in probability to $\btheta_0.$
\end{lemma}

Now, the main theorem on consistency can be stated.
\begin{theorem}
\label{theorem:estimator_consistency}
    Define the random function $\Psi_n(\btheta) = U^{obs}((X_{n:1},\dots,X_{n:n})^t, \btheta).$ If $$X_1,\dots, X_n \sim \text{TN}(\mu_0,\sigma_0,\tau_{l0},\tau_{u0}),$$ $\boldsymbol{\Theta}$ is compact, and $(\hat{\btheta}_n)$ is a sequence of estimators such that $\Psi(\hat{\btheta}_n) = o_p(1)$ as $n \rightarrow \infty$, then $\hat{\btheta}_n \overset{p}{\rightarrow} \btheta_0$, as $n \rightarrow \infty$.
\end{theorem}
\begin{proof}
Let $\epsilon > 0$. Lemma \ref{lemma:limiting_system} established that $$\Psi_n(\btheta)\overset{p}{\rightarrow}\Psi(\btheta),$$
as $n \rightarrow \infty,$ where
\begin{align*}
    \Psi(\btheta)&=\begin{pmatrix}
        \mu_0 + \sigma_0\alpha_{01} - \mu -\sigma\alpha_1 \\
        \mu(J_{\btheta},F_{\btheta_0}) - \mu\alpha_1 - \sigma\alpha_2 \\
        \tau_{u0} - \tau_u \\
        \tau_{l0} - \tau_l
    \end{pmatrix}
\end{align*}
where $\mu(J_{\btheta},F_{\btheta_0}) = \int_0^1 (F_{\btheta_0}^{-1}(t))J_{\btheta}(t)dt$. Now in Appendix \ref{app:derivative} it is established that $\frac{\partial\Psi}{\partial \btheta}$ is continuous at $\btheta_0$ and that $\det\left(\frac{\partial\Psi}{\partial \btheta}(\btheta_0)\right) > 0.$
Thus, by the inverse function theorem, there exist neighborhoods $G$ of $\Psi(\btheta_0)=0$ and $H$ of $\btheta_0$ and a continuous function $g:G\rightarrow H$ such that $g$ is a local inverse to $\Psi$. Thus if $\Psi^*({\btheta})\in G$ then $\exists\delta>0$ such that if $||\Psi(\btheta)-\Psi(\btheta_0)||<\delta$ then $||\btheta - \btheta_0|| = ||g(\Psi(\btheta) )- g(\Psi(\btheta_0))|| < \epsilon.$ By contrapositive $||\btheta -\btheta_0|| \ge \epsilon$ implies $||\Psi(\btheta)||\ge \delta > 0= ||\Psi(\btheta_0)||$ as long as $\Psi(\btheta) \in G.$ Note that as $G$ is a neighborhood, this must happen for small enough $\epsilon$. Thus, by Lemma 2 $\hat{\btheta}_n \overset{p}{\rightarrow} \btheta_0$, as $n \rightarrow \infty$.
\end{proof}

\subsection{Rates of Convergence}
Having established consistency of the estimators, the asymptotic distribution of the estimators is also of interest. The marginal asymptotic distributions of the bound estimators $\hat{\tau}_{nl}$ and $\hat{\tau}_{nu}$ as well as the joint asymptotic distribution of the estimators of the mean and standard deviation of the parent normal distribution $(\hat{\mu}_n,\hat{\sigma}_n)^t.$

To begin, the asymptotic distributions of  $\hat{\tau}_{nl}$ and $\hat{\tau}_{nu}$ are derived using the theory of extreme order statistics. The first result from Gnedenko, as presented in David \cite{David-orderstats2003} gives the asymptotic distribution of the minimum and maximum order statistics from the truncated normal distribution, which will be used in Theorem \ref{thm:limiting_dist_of_bounds}, establishing the asymptotic distribution of $\hat{\tau}_{nl}$ and $\hat{\tau}_{nu}$. 
\begin{lemma}[Result from Gnedenko as Presented in \cite{David-orderstats2003}]
\label{lemma:extereme_limiting_dist}
    Let $X_{n:n}$ be the largest order statistic from a sample of size $n$ from a population with distribution function $F$. There exists constants $a_n$ and $b_n$ such that $\frac{1}{a_n} (X_{n:n}-b_n)\rightsquigarrow G_\alpha(x)=\exp\{-(-x)^\alpha\}, x\le0,\alpha>0$, where $G_\alpha$ is the distribution function of the $\text{Weibull}(\alpha)$ distribution if and only if the upper bound $\xi_1$ is finite and $$\lim_{x\rightarrow0+}\frac{1-F(\xi_1-tx)}{1-F(\xi_1-x)}=t^\alpha$$ for every $t>0.$ Moreover, the standardizing constant $a_n$, making $$\lim_{n \rightarrow \infty }\text{Pr}(X_{n:n}\le\xi_1+a_nx)=G_\alpha(x)$$ may be chosen as $$a_n = \xi_1 - F^{-1}\left(1-\frac{1}{n}\right).$$
\end{lemma}

\begin{theorem}
\label{thm:limiting_dist_of_bounds}
    Define the random function $\Psi_n(\btheta) = U^{obs}((X_{n:1},\dots,X_{n:n})^t, \btheta).$ If $$X_1,\dots, X_n \sim \text{TN}(\mu_0,\sigma_0,\tau_{l0},\tau_{u0}),$$ and $\hat{\btheta}_n$ is a sequence of estimators such that $\Psi(\hat{\btheta}_n) = o_p(1)$ as $n \rightarrow \infty$, then 
$$nf_{\btheta_0}(\tau_{u0})(\hat{\tau}_{nu}-\tau_{u0}) \rightsquigarrow W + 1$$ and $$nf_{\btheta_0}(\tau_{l0})(\hat{\tau}_{nl}-\tau_{l0}) \rightsquigarrow -(W + 1),$$ as $n \rightarrow \infty$, where $W \sim \text{Weibull(1)}.$ 
\end{theorem}

\begin{proof}
By Theorem \ref{theorem:estimator_consistency}, we have $\hat{\btheta}_n \overset{p}{\rightarrow} \btheta_0$, as $n \rightarrow \infty$. Now the truncated normal distribution is bounded above by $\tau_{u0}$. Consider 
\begin{align*}
    \frac{1-F(\tau_{u0}-tx)}{1-F(\tau_{u0}-x)} &= \frac{1-\frac{\Phi(\frac{\tau_{u0}-tx-\mu_0)}{\sigma_0})-\Phi(\tau_{l0}^*)}{\Phi(\tau_{u0^*})-\Phi(\tau_{l0}^*)}}{1-\frac{\Phi(\frac{\tau_{u0}-x-\mu_0)}{\sigma_0})-\Phi(\tau_{l0}^*)}{\Phi(\tau_{u0}^*)-\Phi(\tau_{l0}^*)}} \\
    &= \frac{\Phi(\tau_{u0}^*)-\Phi(\tau_{u0}^* - \frac{tx}{\sigma_0})}{\Phi(\tau_{u0}^*)-\Phi(\tau_{u0}^*-\frac{x}{\sigma_0})}.
\end{align*}
By L'Hopital's Rule, it follows that 
\begin{align*}
    \lim_{x\rightarrow0+}\frac{1-F(\tau_{u0}-tx)}{1-F(\tau_{u0}-x)} &= \lim_{x\rightarrow0+}\frac{\frac{t}{\sigma_0}\phi(\tau_{u0}^* - \frac{tx}{\sigma_0})}{\frac{1}{\sigma_0}\phi(\tau_{u0}^* - \frac{x}{\sigma_0})} \\
    &=t.
\end{align*}
Thus, by Lemma \ref{lemma:extereme_limiting_dist}, there exists a sequence $(a_n)$ such that $\frac{X_{n:n}-\tau_{u0}}{a_n}$ converges in distribution to the $\text{Weibull}(1)$ distribution, where $a_n = \tau_{u0} - F^{-1}\left(1-\frac{1}{n}\right).$ Now consider again using L'Hopital's rule 
\begin{align*}
    \lim_{x\rightarrow0+}\frac{F^{-1}(1)-F^{-1}(1-x)}{x} &= \lim_{x\rightarrow0+}\frac{1}{f(F^{-1}(1-x))} \\
    &= \frac{1}{f(\tau_{u0})}.
\end{align*}
This gives that $\lim_{n\rightarrow \infty} \frac{nf(\tau_{u0})}{\frac{1}{a_n}}=1,$ and so $nf(\tau_{u0})(X_{n:n}-\tau_{u0})\rightsquigarrow W$ as $n\rightarrow\infty$. Then, as $\hat{\tau}_{nu} = X_{n:n} + \frac{\Phi\left( \frac{X_{n:n} - \hat{\mu}_n }{\hat{\sigma}_n} \right) - \Phi\left( \frac{X_{n:n} - \hat{\mu}_n }{\hat{\sigma}_n} \right)}{(n-1) \frac{1}{\hat{\sigma}_n} \phi\left( \frac{X_{n:n} - \hat{\mu}_n }{\hat{\sigma}_n} \right)} + o_P(1)$, as $n \rightarrow \infty$ and by Slutsky's Theorem,
$nf_{\btheta_0}(\tau_{u0})(\hat{\tau}_{nu}-\tau_{u0}) \rightsquigarrow W + 1$, as $n \rightarrow \infty.$

Finally, for $Y \sim \text{TN}(\mu_0,\sigma_0,\tau_{l0},\tau_{u0})$ it follows that $-Y \sim \text{TN}(-\mu_0,\sigma_0,-\tau_{u0},-\tau_{l0}).$ Thus, repeating the above argument gives $-nf(\tau_{l0})(X_{n:1}-\tau_{l0}) \rightsquigarrow W$ as $n\rightarrow\infty$. Then, as $\hat{\tau}_{nl} = X_{n:1} - \frac{\Phi\left( \frac{X_{n:n} - \hat{\mu}_n }{\hat{\sigma}_n} \right) - \Phi\left( \frac{X_{n:n} - \hat{\mu}_n }{\hat{\sigma}_n} \right)}{(n-1) \frac{1}{\hat{\sigma}_n} \phi\left( \frac{X_{n:1} - \hat{\mu}_n }{\hat{\sigma}_n} \right)} + o_P(1)$, as $n\rightarrow\infty$ again using Slutsky's Theorem and the Continuous Mapping Theorem, $nf_{\btheta_0}(\tau_{l0})(\hat{\tau}_{nl}-\tau_{l0}) \rightsquigarrow -(W + 1)$, as $n \rightarrow \infty.$
\end{proof}

Next, to derive the joint asymptotic distribution of $(\hat{\mu}_n, \hat{\sigma}_n)^t$, consider the distribution of $\Psi_n(\btheta)$ for $\btheta \in \Theta$. Using the same result on $L$-estimators, it can be shown that 
$\Psi_n(\btheta)$ is asympotically bivariate normal.

\begin{lemma}
\label{lemma:L_stat_mvn_dist}
    If $X_1,\dots,X_n\sim TN(\mu_0,\sigma_0,\tau_{l0},\tau_{u0})$, then
    $$\sqrt{n}\begin{pmatrix}
        \frac{1}{n}\sum_{i=1}^nX_i - \mathbb{E}(X_i)\\
        \frac{1}{n}\sum_{i=1}^n\xi(n, h_l(n,\btheta), h_u(n,\btheta),i)X_i - \mu(J_{\btheta}, F_{\btheta_0})
    \end{pmatrix}\rightsquigarrow N\left(0,\bSigma(\btheta) \right)$$ as $n\rightarrow\infty$ for any $\btheta \in \Theta,$ where $\Sigma = \begin{pmatrix}
        \sigma_{11}(\btheta) & \sigma_{12}(\btheta) \\  
        \sigma_{12}(\btheta) & \sigma_{22}(\btheta)
    \end{pmatrix}$
for 
$$\sigma_{11}(\btheta)=\int_{\tau_{l0}}^{\tau_{u0}}\int_{\tau_{l0}}^{\tau_{u0}} [ F_{\btheta_0}(\min(x,y)) -F_{\btheta_0}(x)F_{\btheta_0}(y) ] dx dy $$
$$ \sigma_{12}(\btheta)=\int_{\tau_{l0}}^{\tau_{u0}}\int_{\tau_{l0}}^{\tau_{u0}} J_{\btheta}(F_{\btheta_0}(y)) [ F_{\btheta_0}(\min(x,y)) -F_{\btheta_0}(x)F_{\btheta_0}(y) ] dx dy \quad \text{and} $$
$$ \sigma_{22}(\btheta)=\int_{\tau_{l0}}^{\tau_{u0}}\int_{\tau_{l0}}^{\tau_{u0}} J_{\btheta}(F_{\btheta_0}(x))J_{\btheta}(F_{\btheta_0}(y)) [F_{\btheta_0}(\min(x,y)) -F_{\btheta_0}(x)F_{\btheta_0}(y) ] dx dy.$$
\end{lemma}

\begin{proof}
    Let $X_1,\dots,X_n\sim TN(\mu_0,\sigma_0,\tau_{l0},\tau_{u0})$. 
    
    By Lemma \ref{lemma:L_stat_shorackandwellner}, $\sqrt{n}(\frac{1}{n}\sum_{i=1}^n\xi(n, h_l(n,\btheta), h_u(n,\btheta),i)X_i - \mu(J_{\btheta}, F_{\btheta_0})) \rightsquigarrow N(0,\sigma^2(J_{\btheta,F_{\btheta_0}})),$ where $$\sigma^2(J_{\btheta},F_{\btheta_0}) = \int_{\tau_{l0}}^{\tau_{u0}}\int_{\tau_{l0}}^{\tau_{u0}} J_{\btheta}(F_{\btheta_0}(x))J_{\btheta}(F_{\btheta_0}(y)) [F_{\btheta_0}(\min(x,y)) -F_{\btheta_0}(x)F_{\btheta_0}(y) ] dx dy.$$ By the Central Limit Theorem,
    $\sqrt{n}(\frac{1}{n}\sum_{i=1}^nX_i - \mathbb{E}(X_i)) \rightsquigarrow N(0, \mathbb{V}(X_i)).$ Note that in the appendix, it is shown that 
    $$\mathbb{V}(X_i) = \int_{\tau_{l0}}^{\tau_{u0}}\int_{\tau_{l0}}^{\tau_{u0}} [ F_{\btheta_0}(\min(x,y)) -F_{\btheta_0}(x)F_{\btheta_0}(y) ] dx dy.$$
    Finally, consider 
    \begin{align*}
        \bv^t \begin{pmatrix}
        \frac{1}{n}\sum_{i=1}^nX_i )\\
        \frac{1}{n}\sum_{i=1}^n\xi(n, h_l(n,\btheta), h_u(n,\btheta),i)X_i)
    \end{pmatrix} &= v_1 \frac{1}{n}\sum_{i=1}^nX_i + v_2 \frac{1}{n}\sum_{i=1}^n\xi(n, h_l(n,\btheta), h_u(n,\btheta),i)X_i \\
    &= \frac{1}{n}\sum_{i=1}^n(v_1 + v_2 \xi(n, h_l(n,\btheta), h_u(n,\btheta),i))X_i 
    \end{align*}
    Define $H(p;\bv,\btheta) = v_1 + v_2J_{\btheta}(p).$ For fixed $\bv$ we have that $H(p;\bv,\btheta)$ is bounded, Lipschitz, and that $v_1 + v_2 \xi(n, h_l(n,\btheta), h_u(n,\btheta),i) = H(p_i;\bv,\btheta)$ for $p_i \in (\frac{i-1}{n}, \frac{i}{n}].$
    Thus, again by Lemma \ref{lemma:L_stat_shorackandwellner},
    $$\sqrt{n} (\frac{1}{n}\sum_{i=1}^n(v_1 + v_2 \xi(n, h_l(n,\btheta), h_u(n,\btheta),i))X_i - \mu(H(\cdot;\bv,\btheta),F_{\btheta_0})) \rightsquigarrow N(0, \sigma^2(H(\cdot;\bv,\btheta),F_{\btheta_0}))$$
    Now, the mean follows
    \begin{align*}
       \mu(H_{\bv,\btheta},F_{\btheta_0})) &= \int_0^1 F_{\btheta_0}^{-1}(p)H_{\bv,\btheta}(p) dp \\
       &=\int_0^1 F_{\btheta_0}^{-1}(p)(v_1 + v_2 J_{\btheta}(p)) dp \\
       &= v_1 \mathbb{E}(X_i) + v_2 \mu(J_{\btheta},F_{\btheta_0}),
    \end{align*}
    and the asymptotic variance is,
    \begin{alignat*}{2}
    \sigma^2(H_{\bv,\btheta},F_{\btheta_0}))  &= \int_{\tau_{l0}}^{\tau_{u0}}\int_{\tau_{l0}}^{\tau_{u0}} && H_{\bv,\btheta}(F_{\btheta_0}(x))H_{\bv,\btheta}(F_{\btheta_0}(y)) \\ 
    & && \times [F_{\btheta_0}(x \wedge y) -F_{\btheta_0}(x)F_{\btheta_0}(y) ] dx dy \\
    &= \int_{\tau_{l0}}^{\tau_{u0}}\int_{\tau_{l0}}^{\tau_{u0}} &&(v_1 + v_2 J_{\btheta}(F_{\btheta_0}(x)))(v_1 + v_2 J_{\btheta}(F_{\btheta_0}(y))) \\
   & && \times [F_{\btheta_0}(x \wedge y) -F_{\btheta_0}(x)F_{\btheta_0}(y) ] dx dy \\
    &= v_1^2 \sigma_{11} + v_2^2 \sigma_{22} + 2v_1 v_2 \sigma_{12}. \hspace{-5in}&&
    \end{alignat*}
By the Cram\'{e}r-Wold device  it follows that
     $$\sqrt{n}\begin{pmatrix}
        \frac{1}{n}\sum_{i=1}^nX_i - \mathbb{E}(X_i)\\
        \frac{1}{n}\sum_{i=1}^n\xi(n, h_l(n,\btheta), h_u(n,\btheta),i)X_i - \mu(J_{\btheta}, F_{\btheta_0})
    \end{pmatrix}\rightsquigarrow N\left(0,\bSigma(\btheta) \right).$$ 
\end{proof}

It has been shown that $\btheta_n \overset{P}{\rightarrow}\btheta_0$ and that $\sqrt{n}(\Psi_n(\btheta_0) - \Psi(\btheta_0)$ $\Psi(\btheta_0)) \rightsquigarrow (\bZ,0,0)^t$, where $\bZ\sim MVN_2(0, \bSigma(\btheta_0))$, the following conjecture will be utilized to show $(\hat{\mu}_n, \hat{\sigma}_n)$ is asymptotically bivariate normal.
\begin{lemma} \label{lemma:drift_term}
    Define the random function $\Psi_n(\btheta) = U^{obs}((X_{n:1},\dots,X_{n:n})^t, \btheta).$ If $$X_1,\dots, X_n \sim \text{TN}(\mu_0,\sigma_0,\tau_{l0},\tau_{u0}),$$ $\boldsymbol{\Theta}$ is compact, and $\hat{\btheta}_n$ is a sequence of estimators such that $\Psi(\hat{\btheta}_n) = o_p(n^{-1/2})$ as $n \rightarrow \infty$,
    then $\sqrt{n} (\Psi_n(\hat{\btheta}_n) - \Psi(\hat{\btheta_n}) - \sqrt{n} (\Psi_n(\btheta_0) - \Psi(\btheta_0) = o_P(1)$. 
\end{lemma}

\begin{proof}
    Recall the random function

    $$\Psi_n(\btheta) = \begin{pmatrix} n^{-1} \sum_{i=1}^n (X_{n:i} - \mu - \sigma c_{n:i}(\btheta)) \\
n^{-1} \sum_{i=1}^n (X_{n:i} - \mu - \sigma c_{n:i}(\btheta))c_{n:i}(\btheta) \\
X_{n:1} - \frac{\Phi(W_{n:n}) - \Phi(W_{n:1})}{(n-1)\frac{1}{\sigma}\phi(W_{n:1})} - \tau_l \\
X_{n:n} + \frac{\Phi(W_{n:n}) - \Phi(W_{n:1})}{(n-1)\frac{1}{\sigma}\phi(W_{n:n})} - \tau_u
\end{pmatrix},$$
where $c_{n:k}(\btheta) = \xi(n, h_l(n,\btheta),h_u(n,\btheta),k)$ with limiting form
$$\Psi(\btheta)=\begin{pmatrix}
        \mu_0 + \sigma_0\alpha_{01} - \mu -\sigma\alpha_1 \\
        \mu(J_{\btheta},F_{\btheta_0}) - \mu\alpha_1 - \sigma\alpha_2 \\
        \tau_{u0} - \tau_u \\
        \tau_{l0} - \tau_l
    \end{pmatrix}.$$
As defined in Appendix \ref{app:standard_truncated}, $\alpha_1 = \int_{0}^{1}J_{\btheta}(p)dp$ and  $\alpha_2 = \int_{0}^{1}J^2_{\btheta}(p)dp,$ also $\mu + \sigma J_{\btheta}(p)dp = F^{-1}_{\btheta}(p).$ It was also shown that $c_{n:i}(\btheta)) = J_{\btheta}(p_{n:i}(\btheta))$ for some $p_{n:i}(\btheta) \in [\frac{i-1}{n}, \frac{i}{n}].$ Now, 
\begin{equation*}
    \Psi_n(\btheta) - \Psi(\btheta) = \begin{pmatrix} n^{-1} \sum_{i=1}^n (X_{n:i} - F_{\btheta}^{-1}(p_{n:i}(\btheta))) - \int_0^1(F_{\btheta_0}^{-1}(p) - F_{\btheta}^{-1}(p))dp \\
n^{-1} \sum_{i=1}^n (X_{n:i} - F_{\btheta}^{-1}(p_{n:i}(\btheta)))J_{\btheta}(p_{n:i}(\btheta)) - \mu(J_{\btheta},F_{\btheta_0}) + \mu(J_{\btheta},F_{\btheta}) \\
X_{n:1} - \frac{\Phi(W_{n:n}) - \Phi(W_{n:1})}{(n-1)\frac{1}{\sigma}\phi(W_{n:1})} - \tau_{l0} \\
X_{n:n} + \frac{\Phi(W_{n:n}) - \Phi(W_{n:1})}{(n-1)\frac{1}{\sigma}\phi(W_{n:n})} - \tau_{u0}
\end{pmatrix}.
\end{equation*}
Consider $(\Psi_n(\hat{\btheta}_n) - \Psi(\hat{\btheta_n})) -  (\Psi_n(\btheta_0) - \Psi(\btheta_0))$ by element.

The first element is 
\begin{align*}
    \int_0^1(F_{\hat{\btheta}_n}^{-1}(p))dp - n^{-1} \sum_{i=1}^n (F_{\hat{\btheta}_n}^{-1}(p_{n:i}(\hat{\btheta}_n))) - \int_0^1(F_{\btheta_0}^{-1}(p))dp + n^{-1} \sum_{i=1}^n (F_{\btheta_0}^{-1}(p_{n:i}(\btheta_0))),
\end{align*}
which is $o(n^{-1/2})$ as $n\rightarrow \infty$ as shown in Appendix \ref{app:standard_truncated}. 

For The third element define $g_l(\btheta) =\frac{\Phi(\frac{X_{n:n}-\mu}{\sigma}) - \Phi(\frac{X_{n:1}-\mu}{\sigma})}{\frac{1}{\sigma}\phi(\frac{X_{n:1}-\mu}{\sigma})}$ and notice that $$\frac{1}{n-1}(g_l(\hat{\btheta}_n) - g_l(\btheta_0)) = o_P(n^{-1/2})$$ by Thereom \ref{theorem:estimator_consistency} and the Continuous Mapping Theorem. Thus, by Lemma \ref{lemma:extereme_limiting_dist}, that the third element is $o_P(n^{-1/2}).$ A similar argument is used to show that the fourth element is also $o_P(n^{-1/2}).$ 

Finally, the second element is 
\begin{align*}
    & \mu(J_{\hat{\btheta}_n},F_{\hat{\btheta}_n}) - n^{-1} \sum_{i=1}^n J_{\hat{\btheta}_n}(p_{n:i}(\hat{\btheta)}_n)F_{\hat{\btheta}_n}^{-1}(p_{n:i}(\hat{\btheta}_n)) \\
    & + n^{-1} \sum_{i=1}^n J_{\btheta_0}(p_{n:i}(\btheta_0))F_{\btheta_0}^{-1}(p_{n:i}(\btheta_0)) - \mu(J_{\btheta_0},F_{\btheta_0}) \\
    & + n^{-1} \sum_{i=1}^n J_{\hat{\btheta}_n}(p_{n:i}(\hat{\btheta}_n)X_{n:i} - \mu(J_{\hat{\btheta}_n},F_{\btheta_0}) \\
    &+ \mu(J_{\btheta_0},F_{\btheta_0}) - n^{-1} \sum_{i=1}^n J_{\btheta_0}(p_{n:i}(\btheta_0)X_{n:i} - \mu(J_{\btheta_0},F_{\btheta_0}).
\end{align*}
The first two lines in the sum are $o(n^{-1/2})$ as shown in Appendix \ref{app:standard_truncated}. Also, by the Lipschitz condition, the last two lines can be written as
\begin{align*}
    & n^{-1} \sum_{i=1}^n J_{\hat{\btheta}_n}(p_{n:i}(\hat{\btheta}_n)X_{n:i} - \mu(J_{\hat{\btheta}_n},F_{\btheta_0}) \\
    &+ \mu(J_{\btheta_0},F_{\btheta_0}) - n^{-1} \sum_{i=1}^n J_{\btheta_0}(p_{n:i}(\btheta_0)X_{n:i} - \mu(J_{\btheta_0},F_{\btheta_0}) \\
    &= \sum_{i=1}^n \int_{\frac{i-1}{n}}^\frac{i}{n} (J_{\hat{\btheta}_n}(p) - J_{\btheta_0}(p)) (X_{n:i}-F_{\btheta_0}(p))dp.
\end{align*}
By the Mean Value Theorem,
\begin{align*}
    & \sum_{i=1}^n \int_{\frac{i-1}{n}}^\frac{i}{n} (J_{\hat{\btheta}_n}(p) - J_{\btheta_0}) (X_{n:i}-F_{\btheta_0}(p))dp\\
    &=  n^{-1}\sum_{i=1}^n(J_{\hat{\btheta}_n}(q_{n,i}) - J_{\btheta_0}(q_{n,i}) (X_{n:i}-F_{\btheta_0}(q_{n,i}))
\end{align*}
for some $q_{n,i} \in [\frac{i-1}{n}, \frac{i}{n}].$ Now, from the theory of order statistic spacings \citep{David-orderstats2003}, $2f_{\btheta_0}(\xi_p)n(X_{n:r} - X_{n:r-1})\rightsquigarrow\chi^2_2$ as $n \rightarrow \infty,$ where $\frac{r}{n}\rightarrow p.$ Thus, $X_{n:i}-F_{\btheta_0}(q_{n,i})= \hat{\xi}_{q_{n,i}} - \xi_{q_{n,i}}.$ So finally, the absolute value of the second element is 
\begin{align*}
    & |o_P(n^{-1/2}) + n^{-1}\sum_{i=1}^n(J_{\hat{\btheta}_n}(q_{n,i}) - J_{\btheta_0}(q_{n,i}) (X_{n:i}-F_{\btheta_0}(q_{n,i}))| \\
    & \le o_P(n^{-1/2}) + sup_i |J_{\hat{\btheta}_n}(q_{n,i}) - J_{\btheta_0}(q_{n,i})||\hat{\xi}_{q_{n,i}} - \xi_{q_{n,i}}| \\
    & \le o_P(n^{-1/2})+ \sup_{p \in [0,1]}|J_{\hat{\btheta}_n}(p) - J_{\btheta_0}(p)||\hat{\xi}_{q_n^*} - \xi_{q_n^*}|,
\end{align*}
where $q_n^*=\text{argmax}_i |J_{\hat{\btheta}_n}(q_{n,i}) - J_{\btheta_0}(q_{n,i})||\hat{\xi}_{q_{n,i}} - \xi_{q_{n,i}}|$.

Applying the Mean Value Theorem twice,
\begin{align*}
    |J_{\hat{\btheta}_n}(p) - J_{\btheta_0}(p)|
    &= \frac{1}{\phi(\Phi(d_n))}|p(\Phi(\hat{\tau}_{un}^*)-\Phi(\tau_{u0}^*)) + (1-p)(\Phi(\hat{\tau}_{ln}^*)-\Phi(\tau_{l0}^*))| \\
    &= \frac{1}{\phi(\Phi(d_n))}|p\phi(c_{un})(\hat{\tau}_{un}^*-\tau_{u0}^*) + (1-p)\phi(c_{ln})(\hat{\tau}_{ln}^*-\tau_{l0}^*)| \\
    & \le M_u|\hat{\tau}_{un}^*-\tau_{u0}^*| + M_u|\hat{\tau}_{un}^*-\tau_{u0}^*|,
\end{align*}
as $\boldsymbol{\Theta}$ is compact. Thus, $\sup_{p \in [0,1]}|J_{\hat{\btheta}_n}(p) - J_{\btheta_0}(p)| = o_P(1)$ as $n \rightarrow \infty$ by Theorem \ref{theorem:estimator_consistency}.
Finally, for $p\in(0,1)$ define $\sqrt n(\hat{\xi}_p - \xi_p) \rightsquigarrow Z_p$
where $Z_p\sim N(0,\sigma^2_p)$ \citep{serfling-approximation}. Looking at the asymptotic variance
\begin{align*}
    \sigma_p^2 &= \frac{p(1-p)}{f_{\btheta_0}(\xi_p)} \\
    & \le \frac{1}{4}\frac{\Phi(\tau_{u0}^*) - \Phi(\tau_{l0}^*)}{\frac{1}{\sigma^2_0}\phi(\min\{|\tau_{l0}^*|, |\tau_{u0}^*|\})} = \sigma^2_M.
\end{align*}
Thus, $Z_{p}$ is dominated by the random variable $Z_M\sim N(0, \sigma^2_M)$ giving $\hat{\xi}_{q_n^*} - \xi_{q_n^*}=O_p(n^{-1/2})$ if $q_n^*\in(0,1)$. Also, by Theorem \ref{thm:limiting_dist_of_bounds} $\hat{\xi}_{q_n^*} - \xi_{q_n^*} = O_p(n^{-1})$ if $q_n^*=0$ or $q_n^*=1.$ Thus, the second element is $o_P(n^{-1/2}).$ Therefore, $\sqrt{n} (\Psi_n(\hat{\btheta}_n) - \Psi(\hat{\btheta_n}) - \sqrt{n} (\Psi_n(\btheta_0) - \Psi(\btheta_0) = o_P(1)$ as $n\rightarrow\infty.$
\end{proof}

With this conjecture, another result from the theory on $Z$-estimators can be utilized to get the asymptotic distribution of $(\hat{\mu}_n,\hat{\sigma}_n)^t.$ The proof of this theorem will follow closely to Theorem 3.3.1 from Van der Vaart and Wellner \cite{vdvandwellner}.


Finally, with this lemma, we can state the theorem for the asymptotic distribution of the estimators for $\mu$ and $\sigma.$

\begin{theorem}
\label{thm:limiting_dist_of_musigma}
Define the random function $\Psi_n(\btheta) = U^{obs}((X_{n:1},\dots,X_{n:n})^t, \btheta).$ If $$X_1,\dots, X_n \sim \text{TN}(\mu_0,\sigma_0,\tau_{l0},\tau_{u0}),$$ $\boldsymbol{\Theta}$ is compact, and $\hat{\btheta}_n$ is a sequence of estimators such that $\Psi_n(\hat{\btheta}_n) = o_p(n^{-1/2})$ as $n \rightarrow \infty$,
then
$$\sqrt n
   (\hat{\btheta}_n - \btheta_0)\rightsquigarrow \overset{\cdot}{\Psi}(\btheta_0)^{-1}(\bZ,0,0)^t,$$ as $n\rightarrow \infty,$ where $\bZ\sim MVN_2(0, \bSigma(\btheta_0))$.
    
\end{theorem}

\begin{proof}
As $\Psi_n(\hat{\btheta}_n) = o_p(\sqrt{n})$ as $n \rightarrow \infty$ , by Lemma \ref{lemma:z_est_consistency} $\hat{\btheta}_n \overset{P}{\rightarrow} \btheta_0$ and $n \rightarrow \infty.$ Also, as $\Psi(\btheta_0)=0,$ it follows that the first two elements of $\sqrt{n}(\Psi_n -\Psi)(\btheta_0) = \sqrt{n}\Psi_n(\btheta_0) \rightsquigarrow Z,$ as $n \rightarrow \infty$ where $Z \sim MVN_2(0, \bSigma(\btheta_0))$ by Lemma \ref{lemma:L_stat_mvn_dist}. Also, by Theorem \ref{thm:limiting_dist_of_bounds}, $(\hat{\tau}_{nl}-\tau_{l0}) = O_p(n^{-1})$ and $(\hat{\tau}_{nu}-\tau_{u0}) = O_p(n^{-1}),$ as $n\rightarrow \infty,$ giving the last two elements of $\sqrt{n}(\Psi_n -\Psi)(\btheta_0) = o_p(1)$, as $n\rightarrow \infty.$ Putting these results together it follows that  $\sqrt{n}(\Psi_n -\Psi)(\btheta_0) \rightsquigarrow (\bZ,0,0)^t$. 

Using Lemma \ref{lemma:drift_term},  $\Psi(\hat{\btheta}_n) = o_p(n^{-1/2})$ as $n \rightarrow \infty$,
    which gives $\sqrt{n} (\Psi_n(\hat{\btheta}_n) - \Psi(\hat{\btheta_n}) - \sqrt{n} (\Psi_n(\btheta_0) - \Psi(\btheta_0) = o_P(1)$, so
    \begin{align*}
        \sqrt{n}\left(\Psi(\hat{\btheta}_n) - \Psi(\btheta_0)\right) &= -\sqrt{n} \left( \Psi_n(\btheta_0) - \Psi(\btheta_0) \right) + o_P(1).
    \end{align*}
Thus, by the derivation in Appendix \ref{app:derivative}, $$\sqrt{n} 
   \begin{pmatrix}
       \hat{\mu}_n - \mu_0 \\
       \hat{\sigma}_n - \sigma_0
   \end{pmatrix} \rightsquigarrow MVN(0, \Gamma \Sigma\Gamma^t),$$ where 
   $\Gamma = \begin{bmatrix}
       \alpha_{2_0} - \alpha_{1_0}^2 & \alpha_{3_0} - \alpha_{1_0}\alpha_{2_0} \\
       \frac{1}{2}(\alpha_{3_0} - \alpha_{1_0}\alpha_{2_0}) & \frac{1}{2}(\alpha_{4_0} - \alpha_{2_0}^2)
   \end{bmatrix}^{-1}.$
\end{proof}

\section{Numerical Experiments} 
\label{sec:sims}

Monte Carlo simulation will be used to study the rates of convergence related to and visualize the large sample behavior of the proposed estimators. 

\subsection{Setup}

As the theory outlined in Section \ref{sec:fixed_convergence} only guaranteed convergence to a solution of the system of equations in \eqref{eqn:obsest} when the starting point is in a neighborhood of the solution. The simulation will use $\btheta^{(0)}= \btheta_0$. Different cases will be considered for the ground truth models, fixing $\mu_0=0$ and $\sigma_0=1$ with nine different cases for the values $\tau_{l0}$ and $\tau_{u0}$ shown in Table \ref{tab:cases}, which allow for varying to allow for different distribution shapes. Figure \ref{fig:cdf_cases} visualizes the shapes of the different truncated normal distributions considered.

\begin{table}[htb]
    \centering
    \caption{The different combinations of $\tau_{l0}$ and $\tau_{u0}$ used for the simulations.}
    \begin{tabular}{|c|c|c|c|c|c|c|}
        \hline
        Case & 1 & 2 & 3 & 4 & 5 & 6 \\
        \hline
        $\tau_{l0}$ & -3 & -2 & -2 & -1 & -1 & 1\\
        $\tau_{u0}$ & -1 & 1 & 2 & 1 & 2 & 3 \\
        \hline
    \end{tabular}
    \label{tab:cases}
\end{table}

\begin{figure}
    \centering

    \begin{subfigure}[b]{0.5\textwidth}
        \centering
        \includegraphics[width=\textwidth]{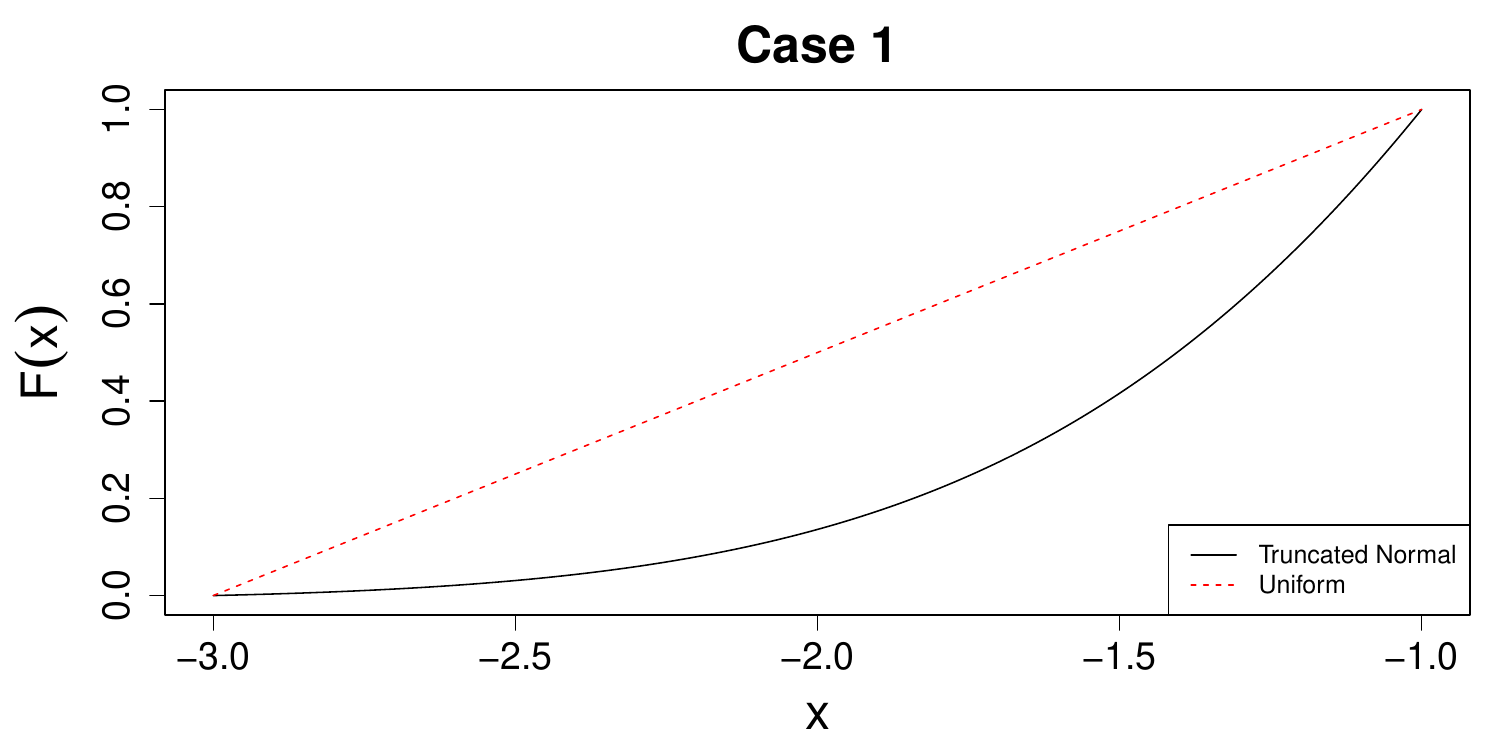}
    \end{subfigure}%
    \begin{subfigure}[b]{0.5\textwidth}
        \centering
        \includegraphics[width=\textwidth]{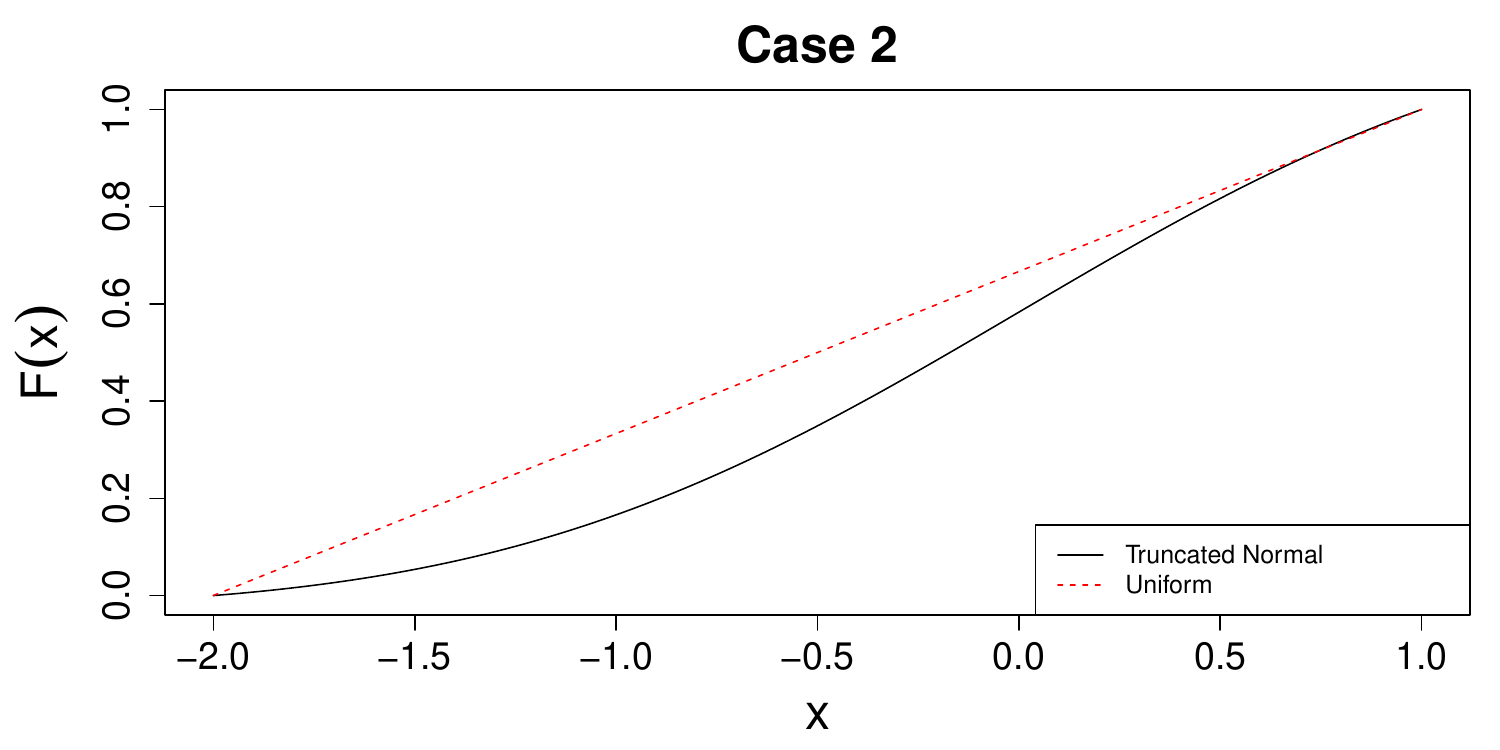}
    \end{subfigure}

    \begin{subfigure}[b]{0.5\textwidth}
        \centering
        \includegraphics[width=\textwidth]{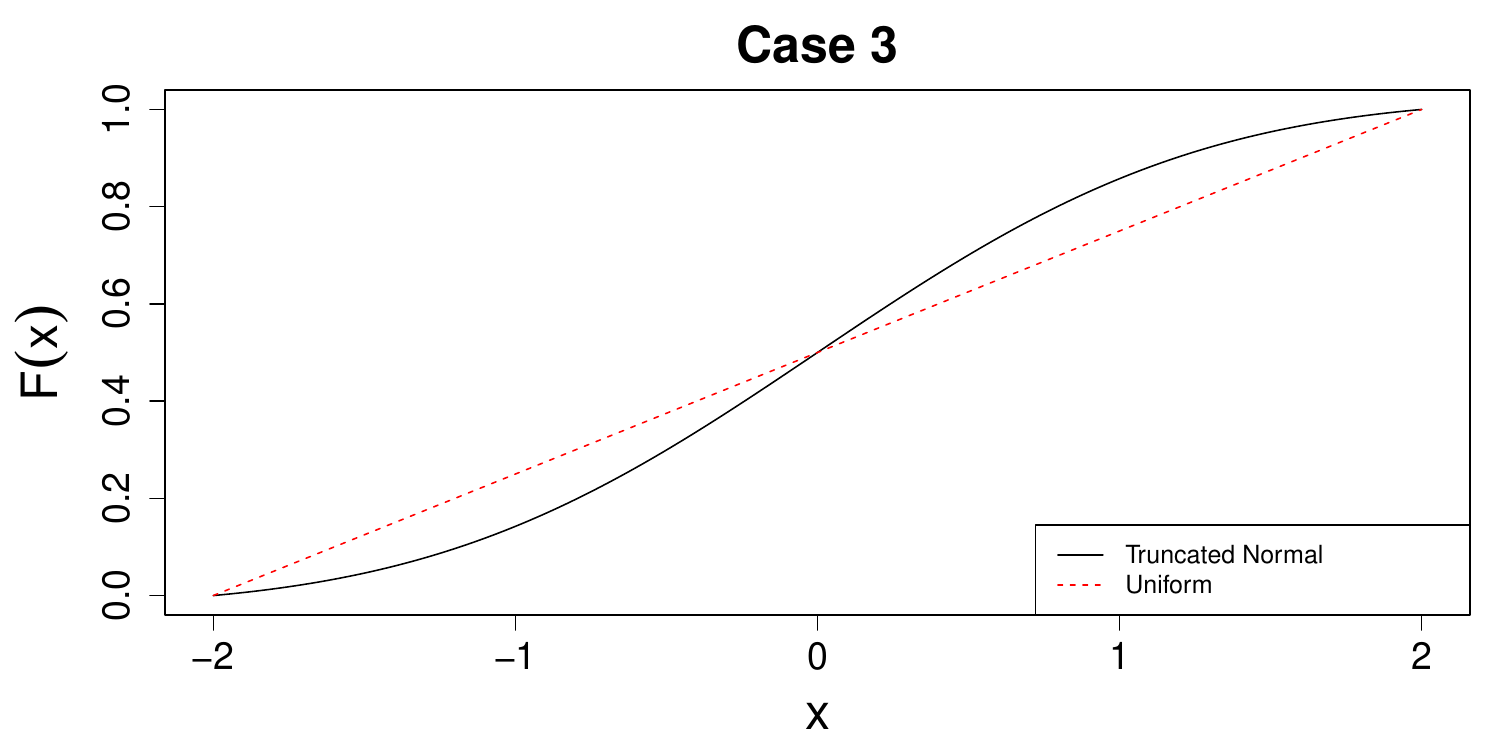}
    \end{subfigure}%
    \begin{subfigure}[b]{0.5\textwidth}
        \centering
        \includegraphics[width=\textwidth]{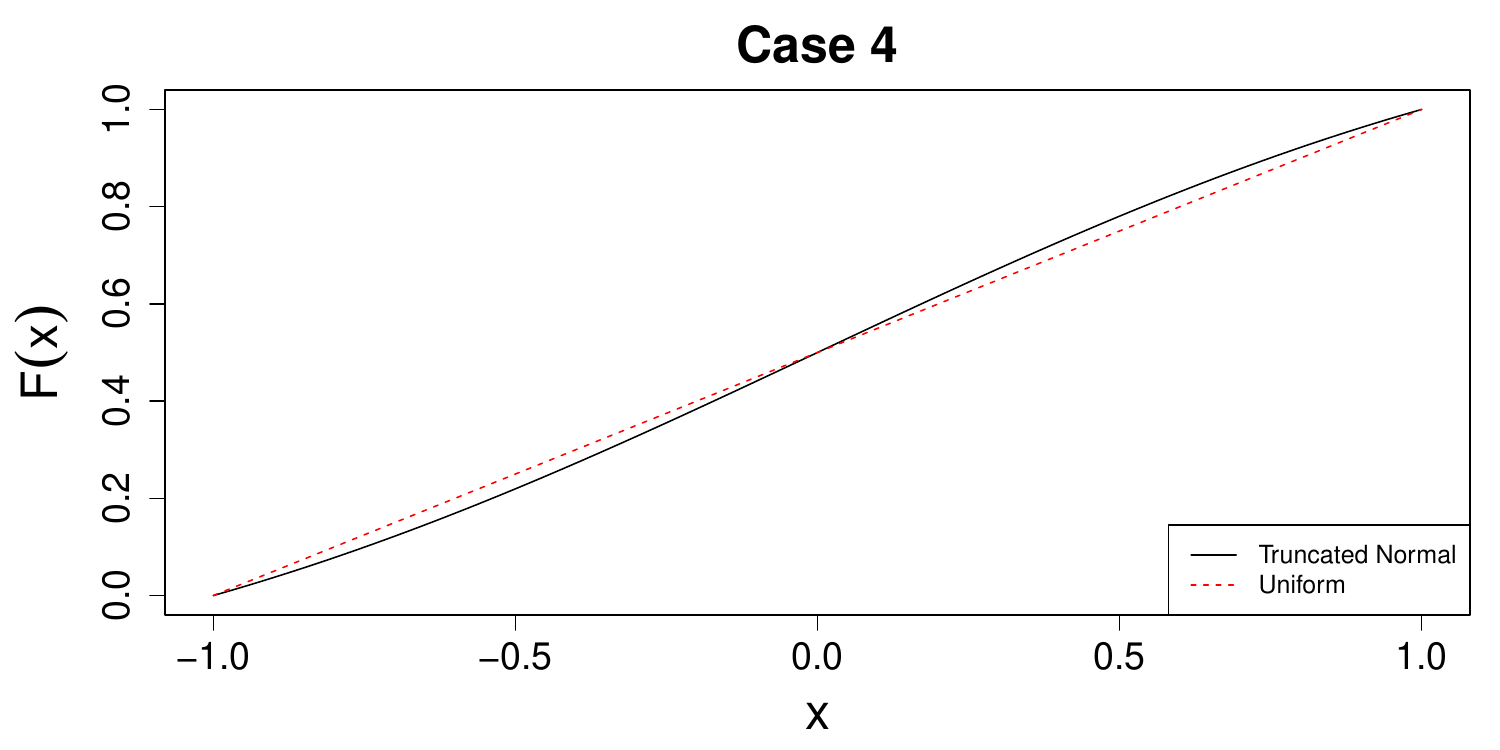}
    \end{subfigure}

    \begin{subfigure}[b]{0.5\textwidth}
        \centering
        \includegraphics[width=\textwidth]{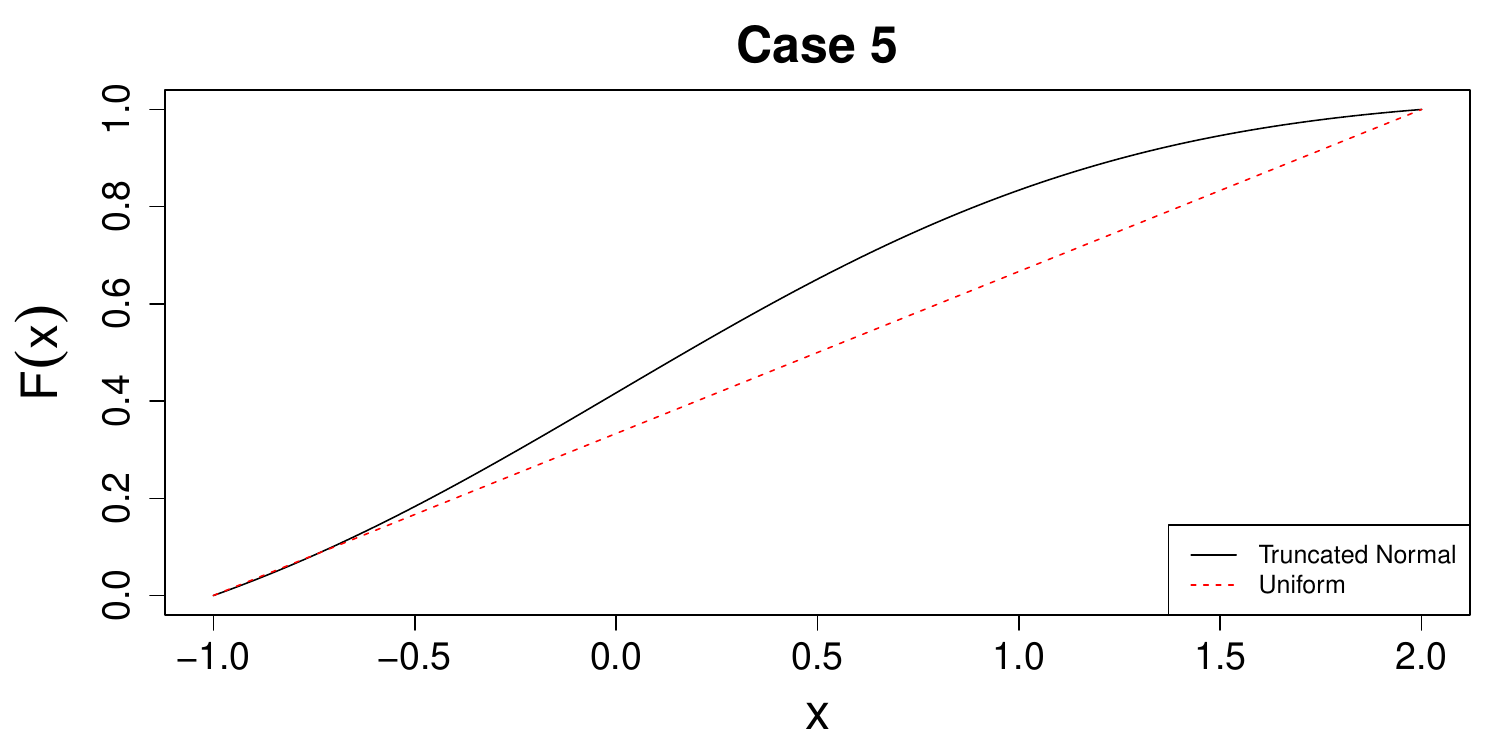}
    \end{subfigure}%
    \begin{subfigure}[b]{0.5\textwidth}
        \centering
        \includegraphics[width=\textwidth]{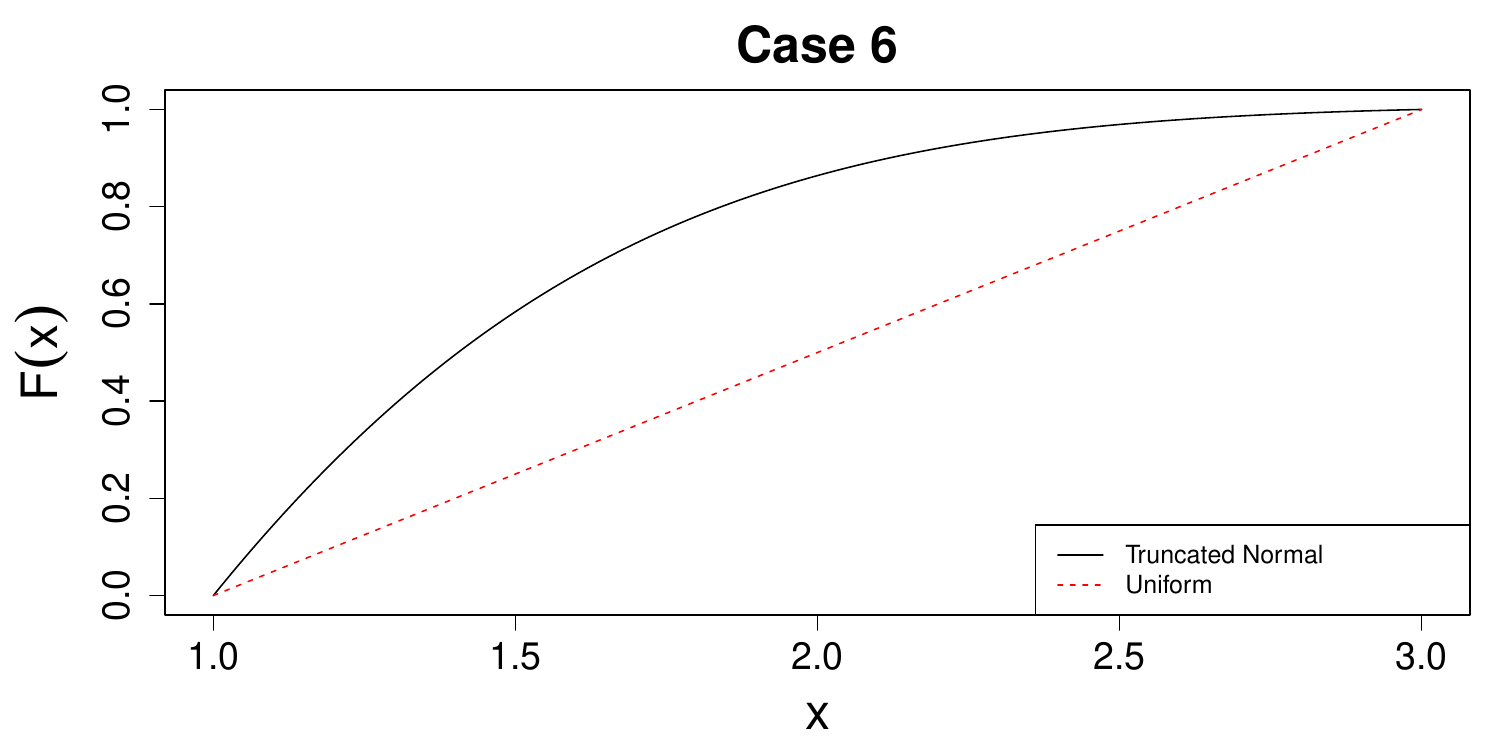}
    \end{subfigure}

    \caption{Visualization of the different cases used for simulation, comparing the distribution function for the truncated normal distribution (black solid line) to the distribution function of the uniform distribution with the same bounds (red dashed line).}
    \label{fig:cdf_cases}
\end{figure}

\begin{algorithm}[H]
\caption{Simulation to Study Large Sample Behavior}
\begin{algorithmic}

\STATE 1. Let $\btheta = (\mu_0 =0, \sigma_0 = 1, \tau_{l0}, \tau_{u0})^t$.
\FOR{Increasing $n$}
\FOR{$b \in {1,2,\dots,B}$}

\STATE 2. Sample $X_1,\dots, X_n \sim \text{TN}(\mu_0, \sigma_0, \tau_{l0}, \tau_{u0})$.

\STATE 3. Set $\btheta^{(0)}=(\bar{x},s, x_{n:1}, x_{n:n})^t$. 

\STATE 4. Run the proposed ES Algorithm until the relative change in log-likelihood is less than $10^{-6}$ to obtain  $\hat{\btheta}_n$.

\STATE 5. Calculate $||U^{obs}(\hat{\btheta}_n, x^{obs})||$.

\STATE 6. Store $\hat{\btheta}_n$, $||U^{obs}(\tilde{\btheta}, x^{obs})||$, and $n$.

\ENDFOR
\ENDFOR

\end{algorithmic}
\label{alg:cons_simulation}
\end{algorithm}

Algorithm \ref{alg:cons_simulation} gives the simulation that will be used to study the large sample behavior. To study the consistency of the estimators, the values of $n$ considered are an evenly spaced sequence of length from 10 to 1000 on the $\log_{10}$ scale with $B=500$. For the study of the asymptotic distributions, the values of $n$ considered will be $n \in \{30, 50, 100 \}$ with $B=10,000$.

\subsection{Results}

Figures \ref{fig:mu_cons}, \ref{fig:sigma_cons}, \ref{fig:tau_l_cons}, and \ref{fig:tau_u_cons} show the results of the Monte Carlo simulation studying the consistency of the proposed estimators. In each of the plots, the distribution of the estimator given the sample size $n$ is visualized with $n$ on the $\log_{10}$ scale. Within the plots, the blue horizontal dashed line is the true value of the parameter. The median, first quartile, and third quartile of the distribution are modeled using a quantile regression utilizing cubic B-splines with 6 degrees of freedom. The plots read similarly to a boxplot for a fixed value of $\log_{10}(n)$, where the length of the shaded region is the interquartile range (IQR) with the median curve plotted as a black line within the shaded region. The lower curve is the point-wise maximum of the lower bound on the parameter induced by the bounding of the parameter space and $1.5*IQR$ subtracted from the first quartile. The upper curve is the point-wise minimum of the upper bound on the parameter and $1.5*IQR$ added to the third quartile. If any of the estimator value lies outside the upper and lower curves, it is deemed an outlier and is plotted as a red point.

Looking at Figure \ref{fig:mu_cons}, which shows the distribution of $\hat{\mu}_n|\log_{10}(n)$, it can be seen that in Cases 1 and 2, the distribution of the parameter is quite wide, with the min and max curves taking the upper and lower bounds put on the $\mu$ when bounding the parameter space. The distribution does begin to collapse around the true value. For Cases 2,3,4, and 5, the distribution collapses around the true value, with quite a few outliers present. It can also be noted that Cases 1 and 2 are near reflections of Cases 3 and 4 about $\mu_0$.

Figure \ref{fig:sigma_cons}, shows the distribution of $\log_{10}(\hat{\sigma}_n)|\log_{10}(n).$ Again, it can be seen in Cases 1 and 2, now with the addition of Case 4, the min and max curves take the values of the upper and lower bounds on $\log_{10}(\sigma)$ induced by the bounding of the parameter space. In all cases, the distribution does begin to collapse around $\log_{10}(\sigma_0)$, with many outliers present, especially in Cases 2,3,4, and 5. An interesting behavior is now seen, that there is similarity in the distributions seen in Cases 1 and 6, as well as a similarity in Cases 2 and 5.

\begin{figure}[htbp]
    \centering

    \begin{subfigure}[b]{0.5\textwidth}
        \centering
        \includegraphics[width=\textwidth]{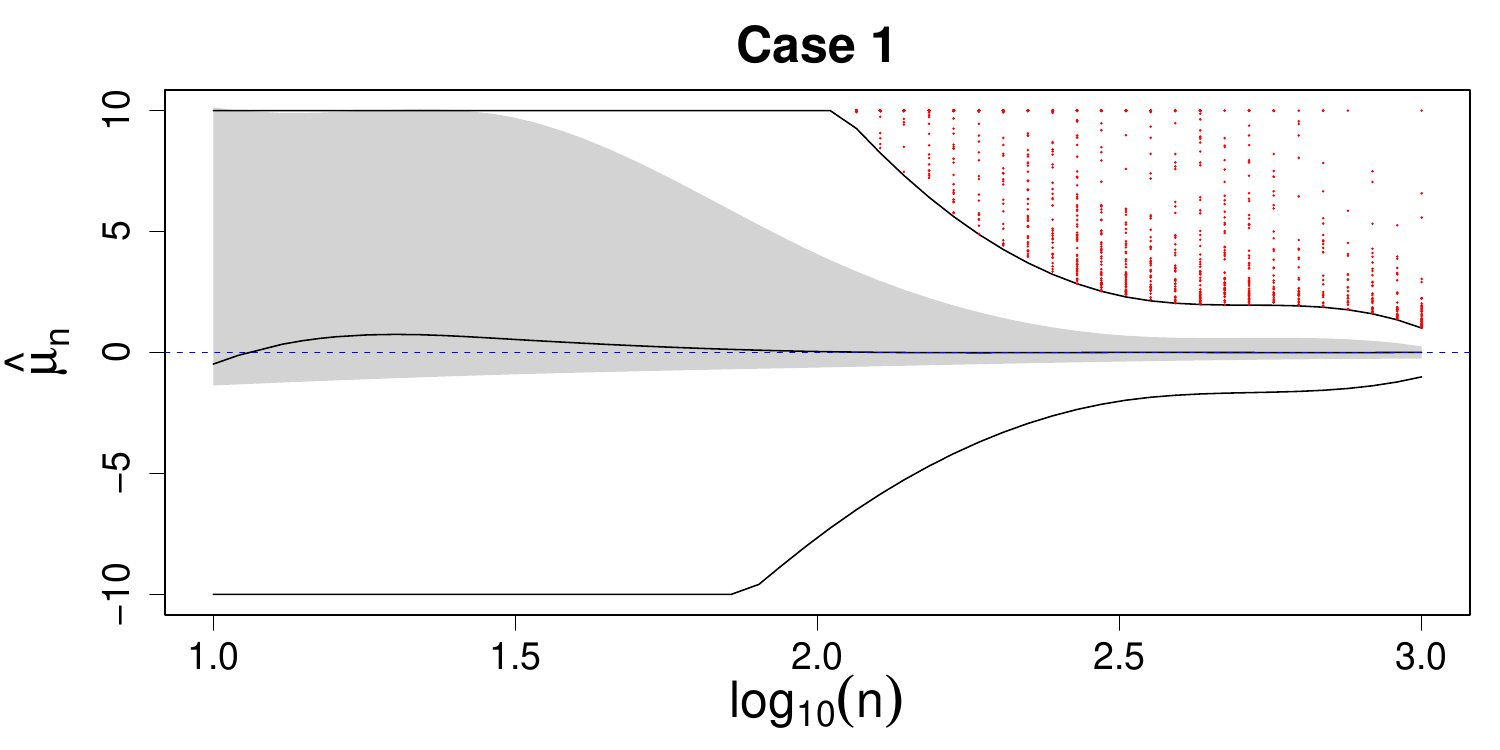}
    \end{subfigure}%
    \begin{subfigure}[b]{0.5\textwidth}
        \centering
        \includegraphics[width=\textwidth]{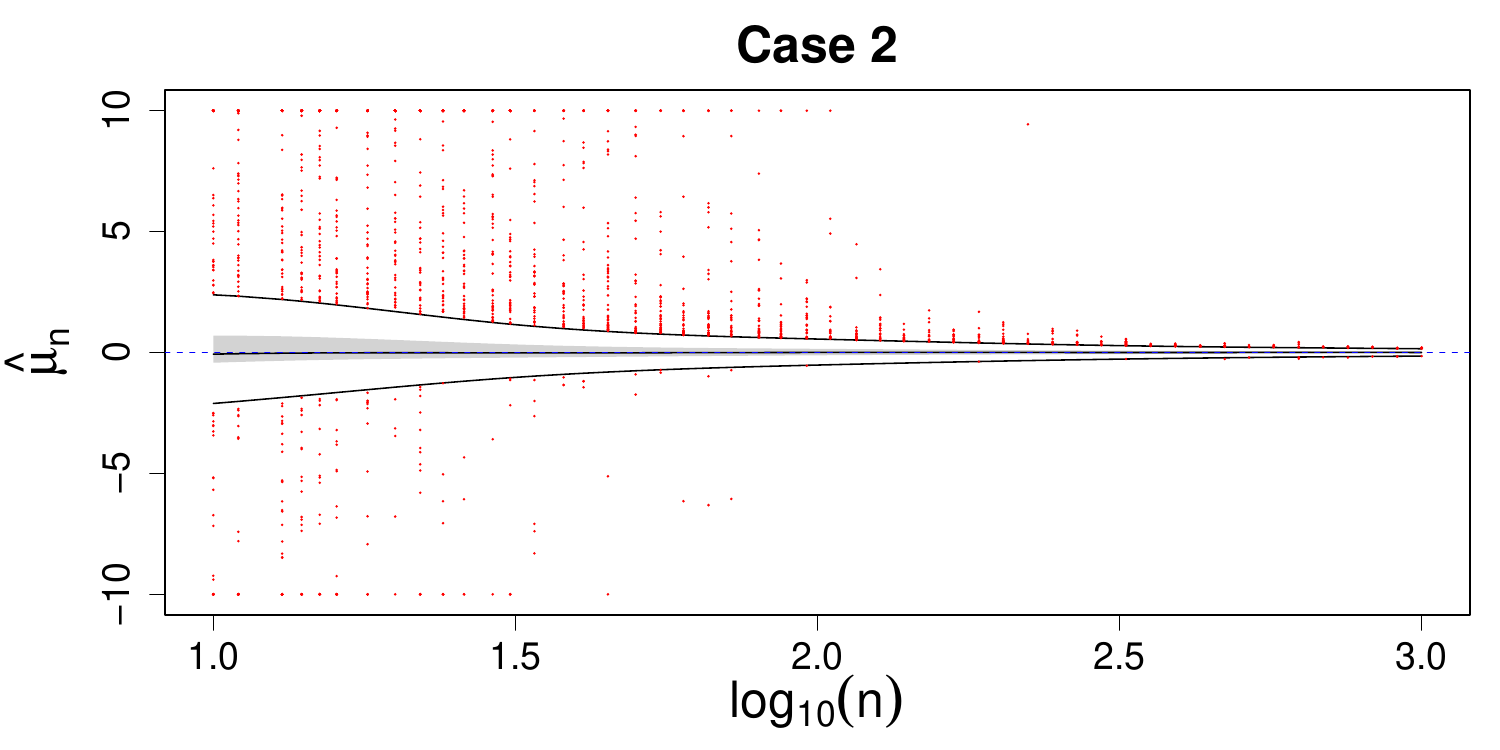}
    \end{subfigure}

    \begin{subfigure}[b]{0.5\textwidth}
        \centering
        \includegraphics[width=\textwidth]{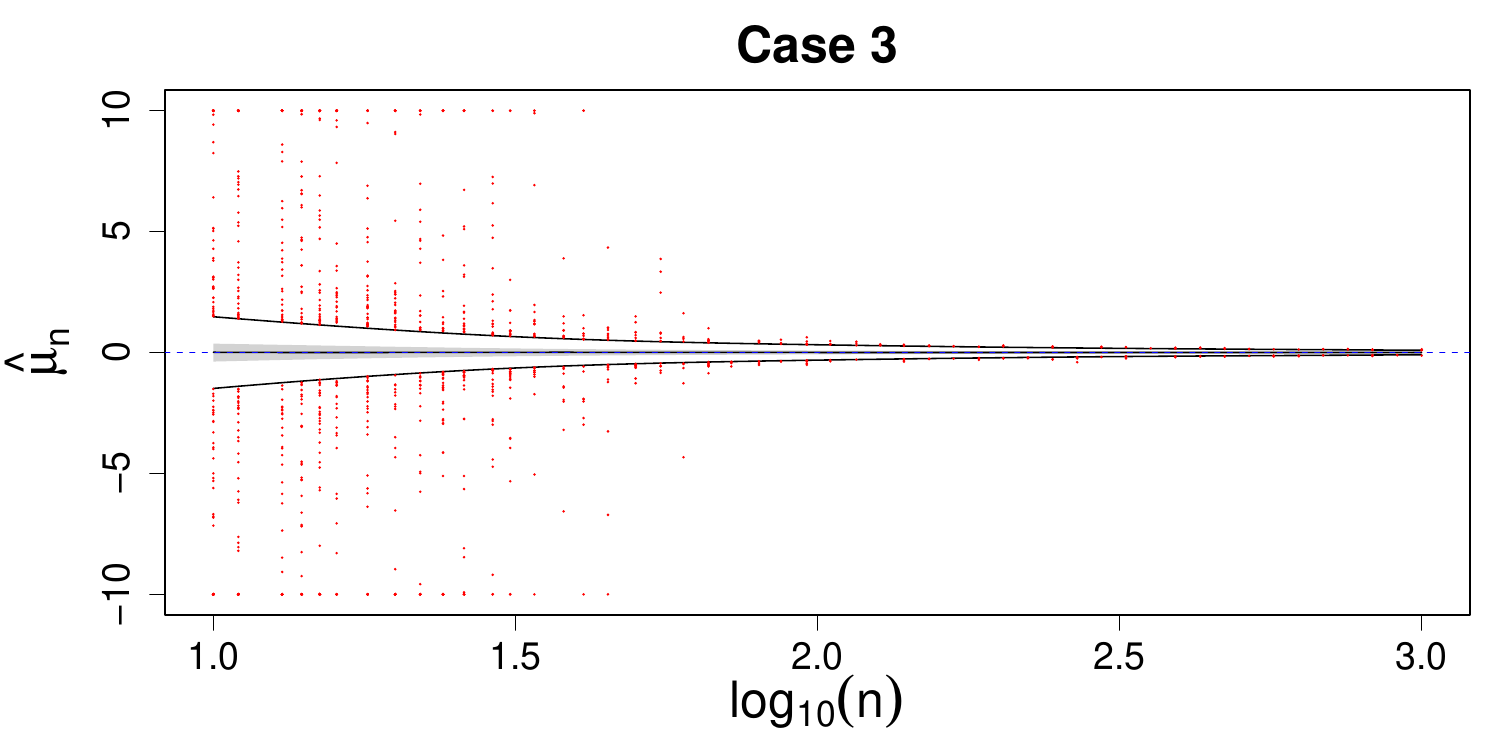}
    \end{subfigure}%
    \begin{subfigure}[b]{0.5\textwidth}
        \centering
        \includegraphics[width=\textwidth]{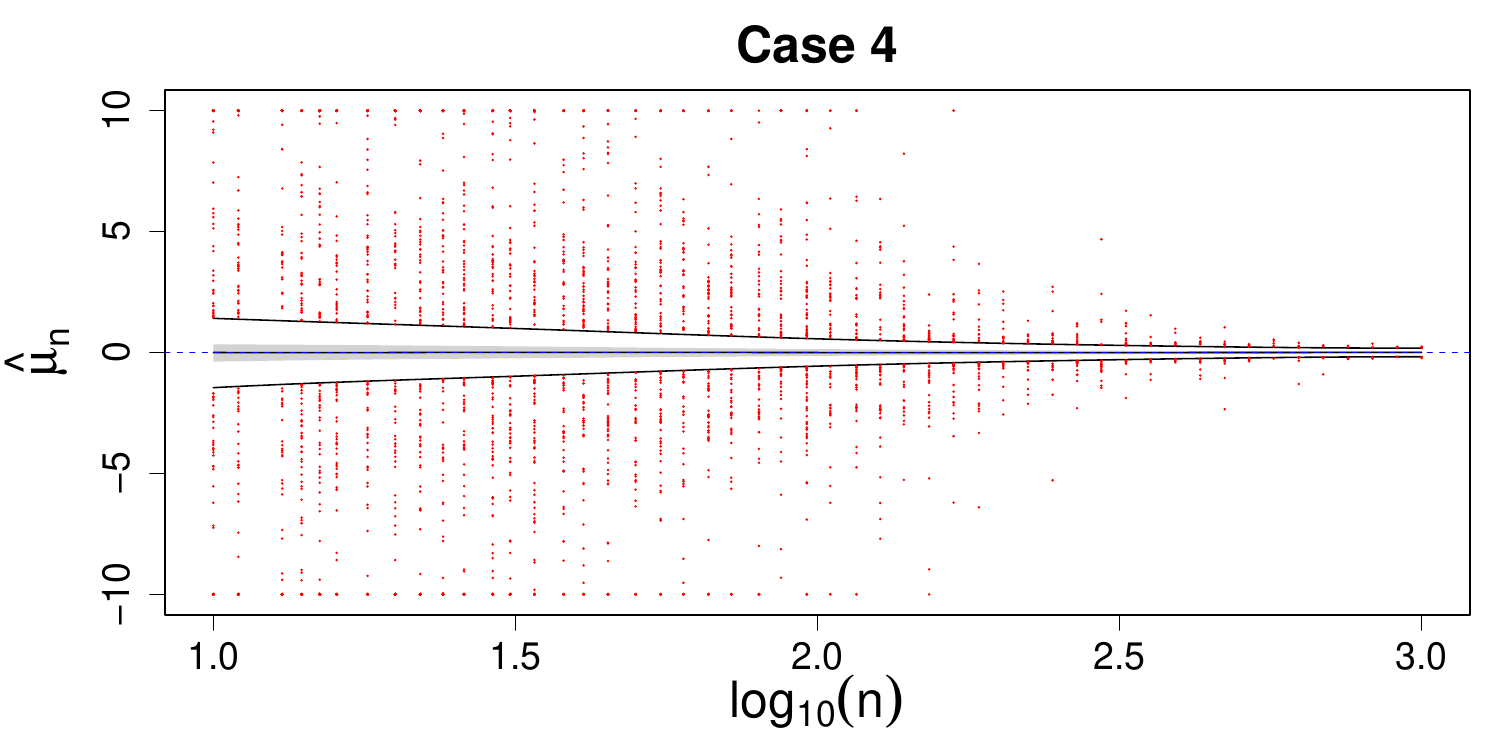}
    \end{subfigure}

    \begin{subfigure}[b]{0.5\textwidth}
        \centering
        \includegraphics[width=\textwidth]{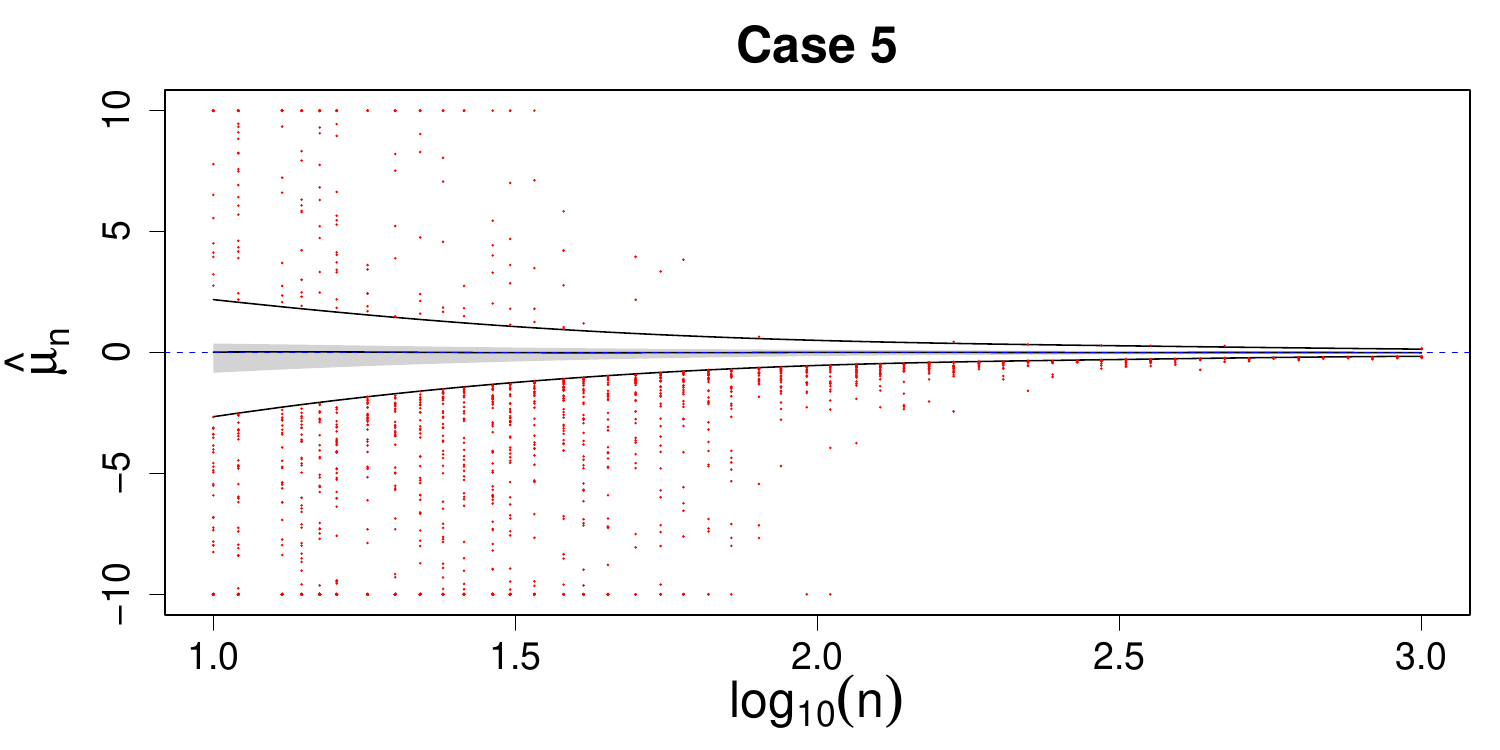}
    \end{subfigure}%
    \begin{subfigure}[b]{0.5\textwidth}
        \centering
        \includegraphics[width=\textwidth]{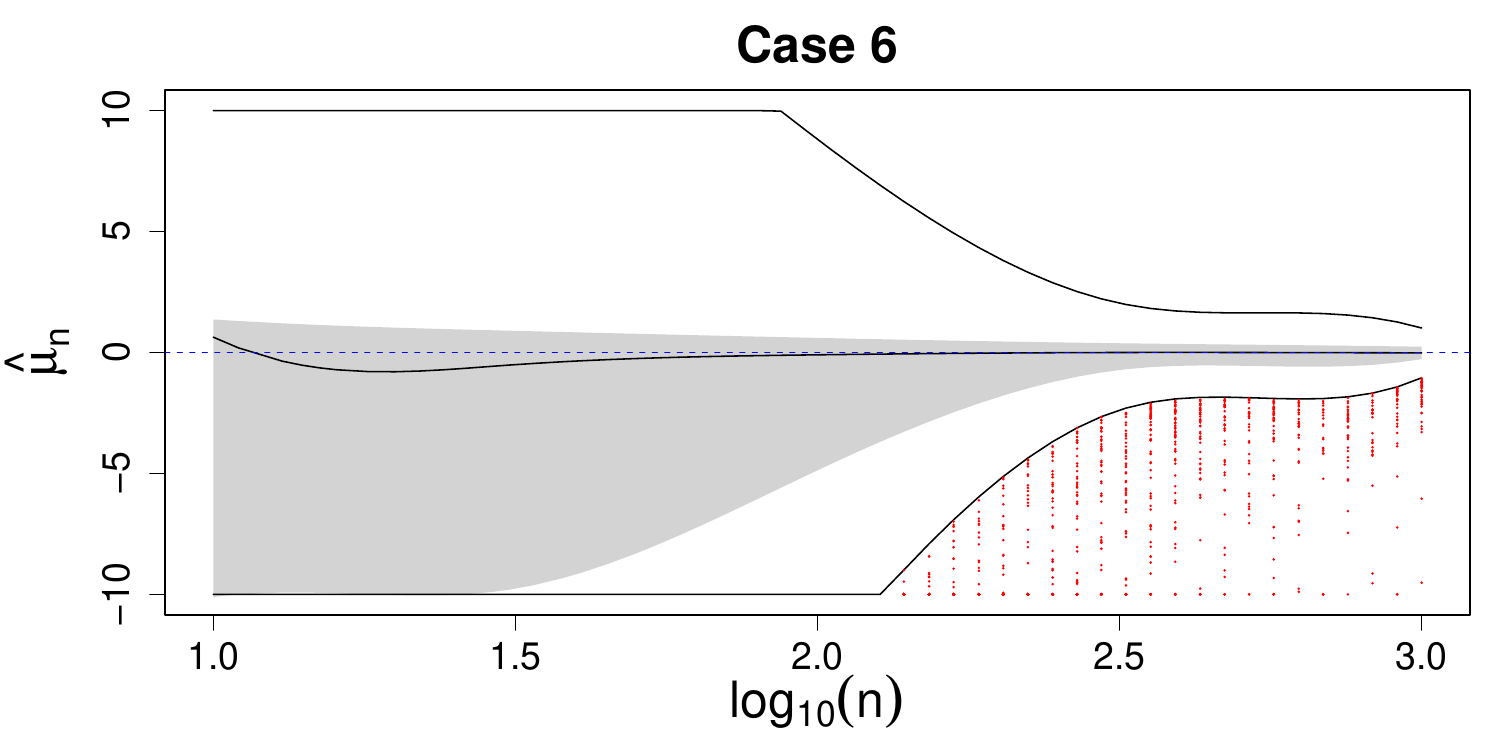}
    \end{subfigure}

    \caption{Results of the simulations illustrating the distribution of $\hat{\mu}_n$ as $n$ increases.}
    \label{fig:mu_cons}
\end{figure}

Finally, Figures \ref{fig:tau_l_cons} and \ref{fig:tau_u_cons} visualize the distribution of $\hat{\tau}_{nl}|\log_{10}(n)$ and $\hat{\tau}_{nu}|\log_{10}(n)$, respectively. In all cases, the distribution of the truncation bound estimators collapses quickly around the true value of the truncation bound with a few outliers present. A notable relation here is that the distributions in Cases 1,2,3,4,5, and 6 from Figure \ref{fig:tau_l_cons} are reflections of Cases 6,5,4,3,2 and 1 from Figure \ref{fig:tau_u_cons}, respectively.

\begin{figure}[htbp]
    \centering

    \begin{subfigure}[b]{0.5\textwidth}
        \centering
        \includegraphics[width=\textwidth]{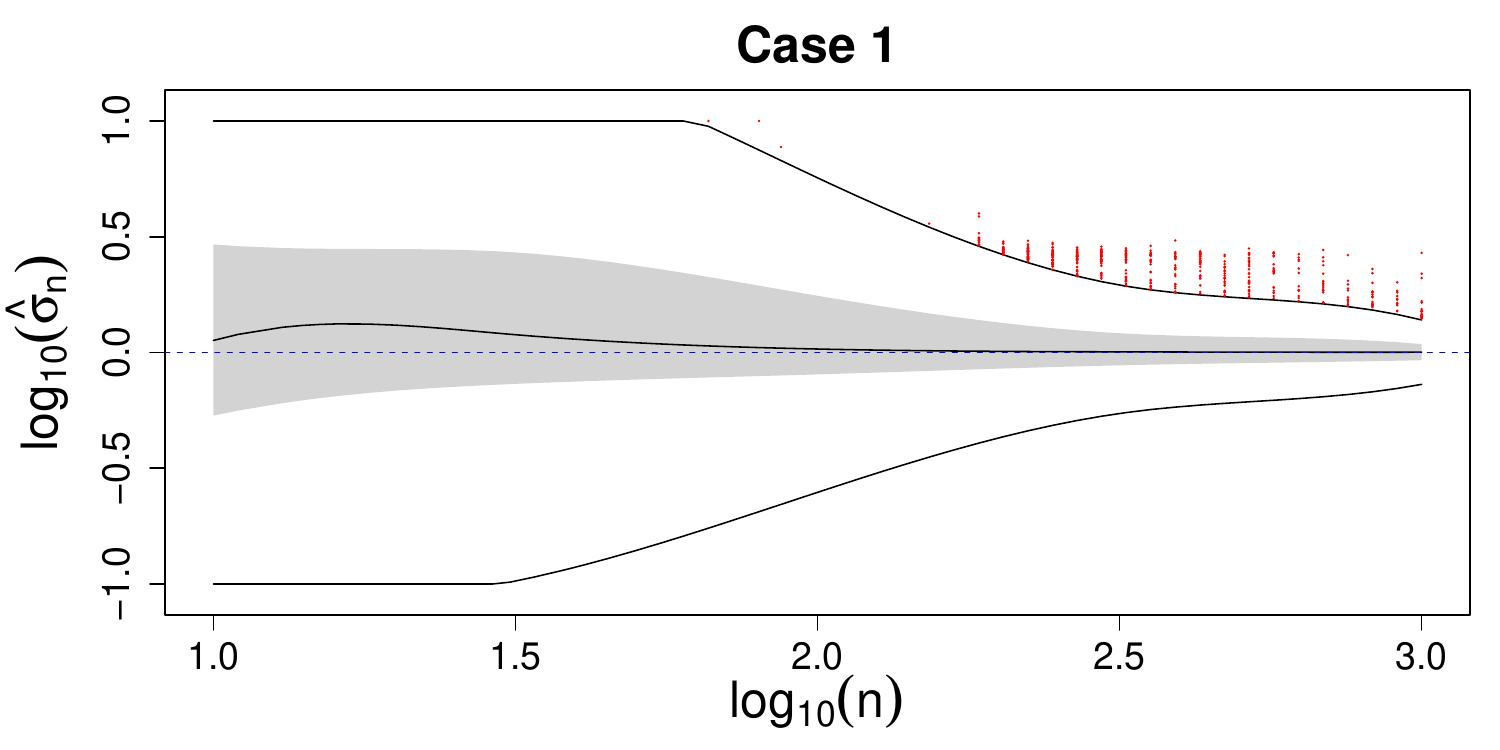}
    \end{subfigure}%
    \begin{subfigure}[b]{0.5\textwidth}
        \centering
        \includegraphics[width=\textwidth]{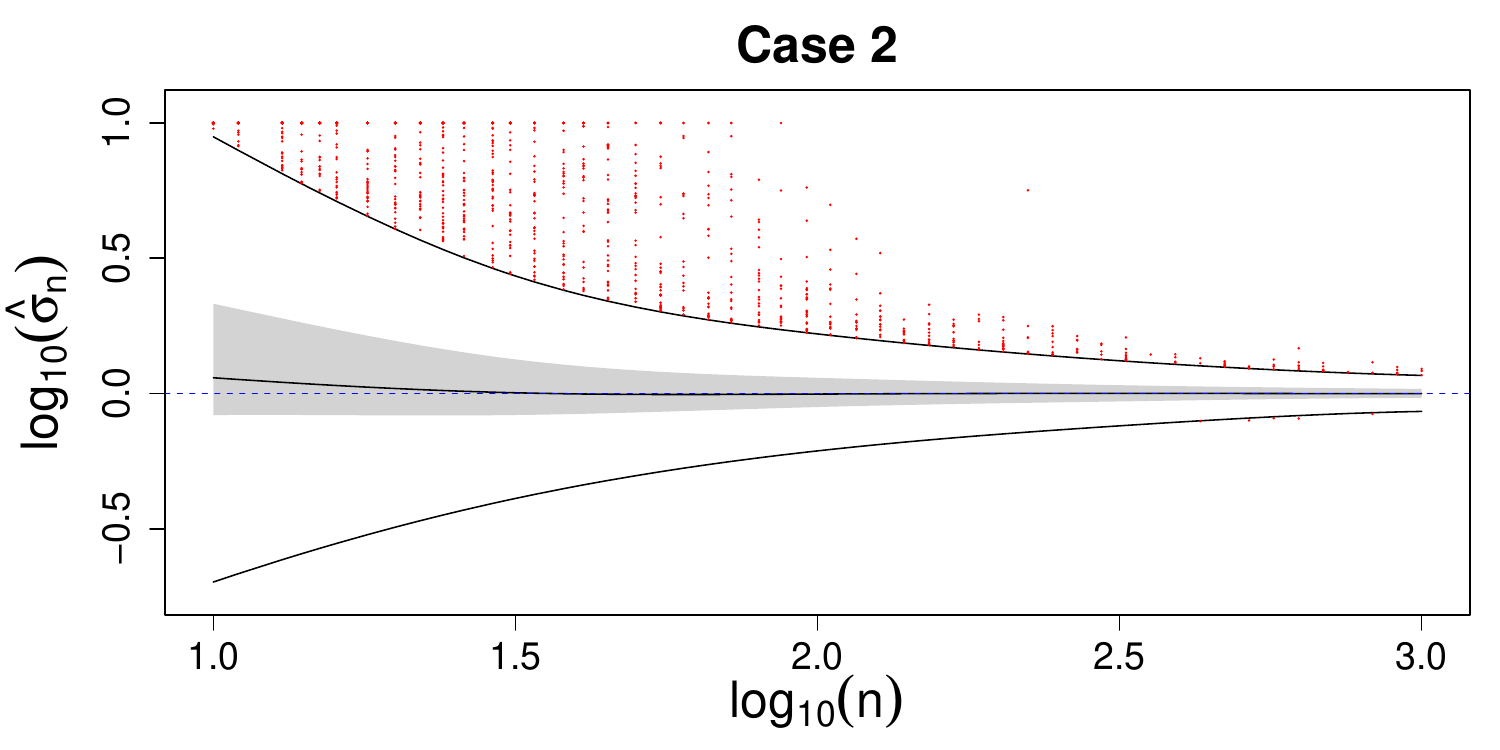}
    \end{subfigure}

    \begin{subfigure}[b]{0.5\textwidth}
        \centering
        \includegraphics[width=\textwidth]{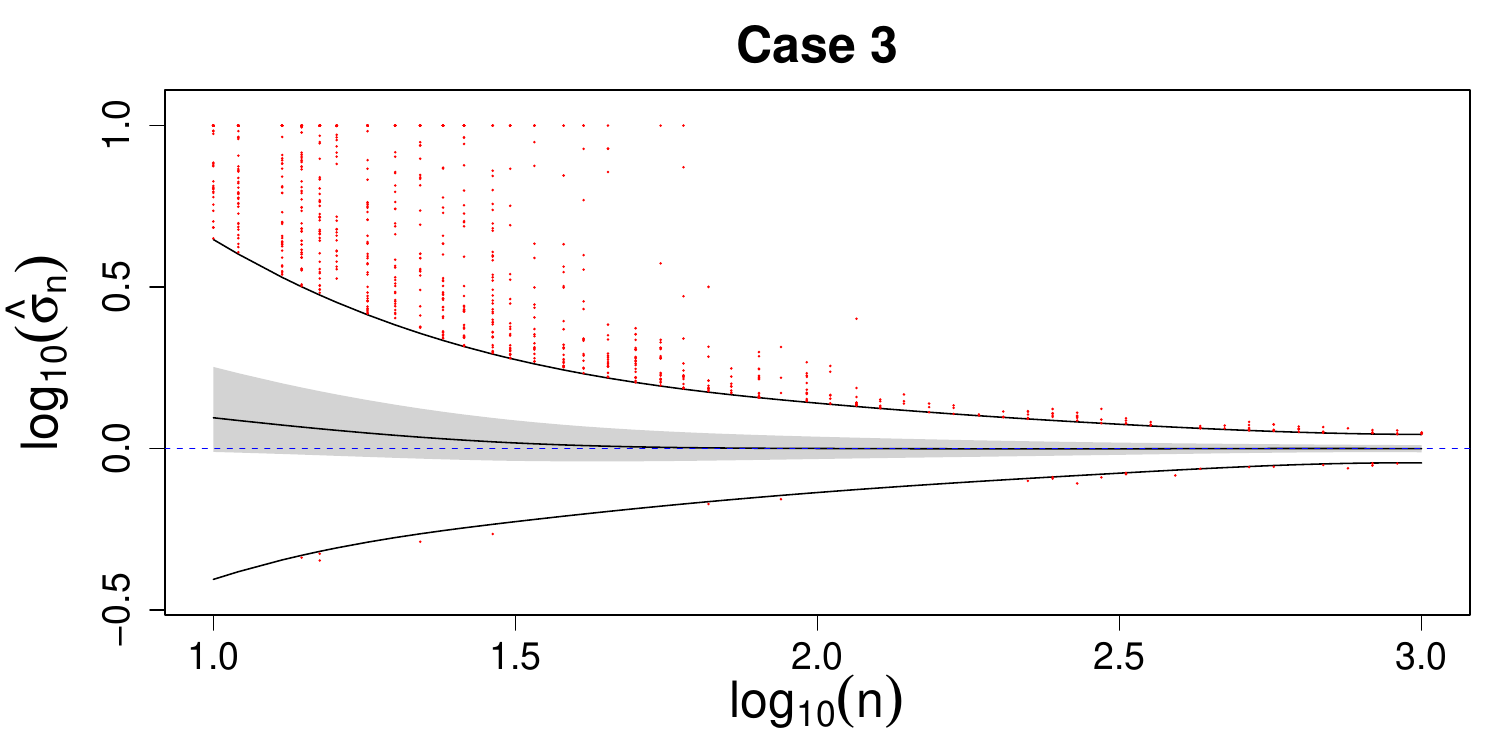}
    \end{subfigure}%
    \begin{subfigure}[b]{0.5\textwidth}
        \centering
        \includegraphics[width=\textwidth]{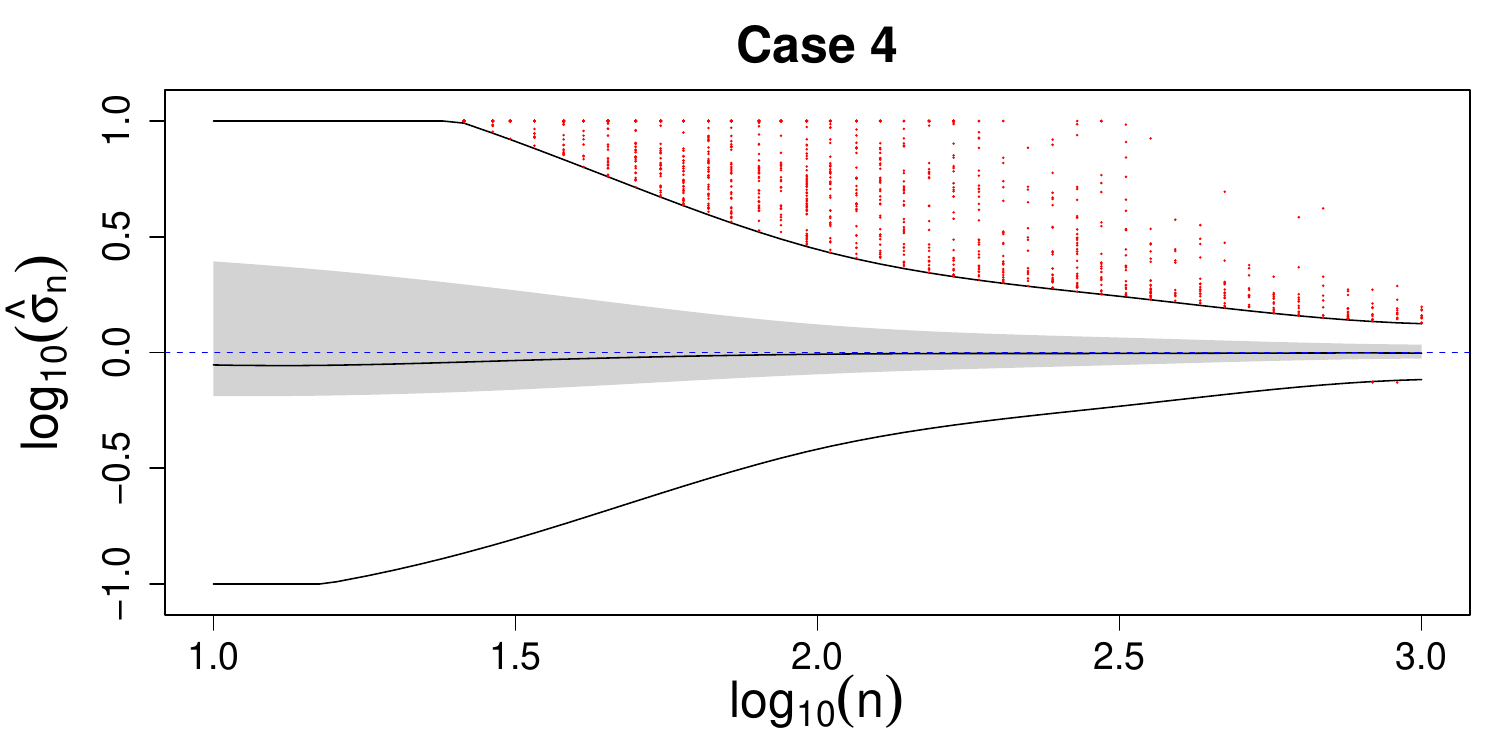}
    \end{subfigure}

    \begin{subfigure}[b]{0.5\textwidth}
        \centering
        \includegraphics[width=\textwidth]{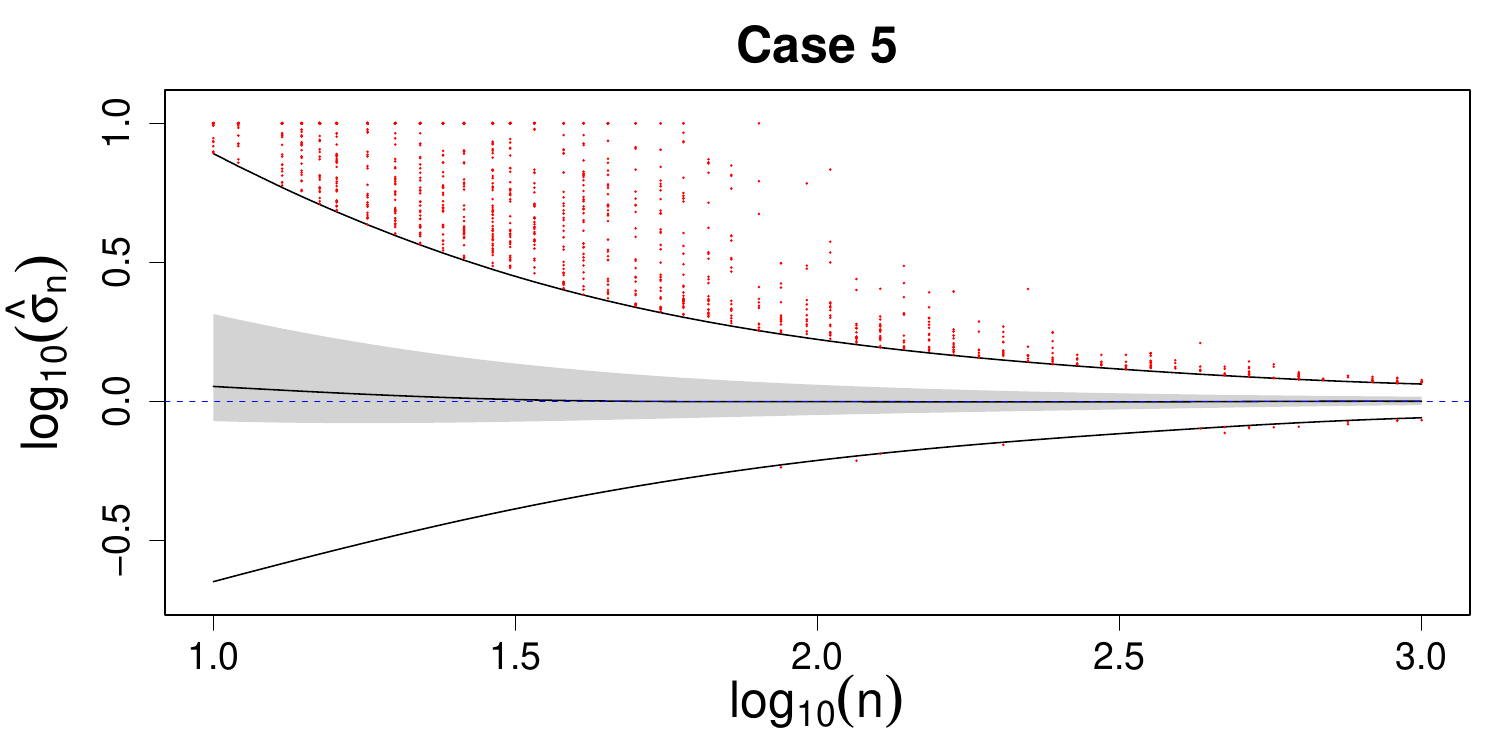}
    \end{subfigure}%
    \begin{subfigure}[b]{0.5\textwidth}
        \centering
        \includegraphics[width=\textwidth]{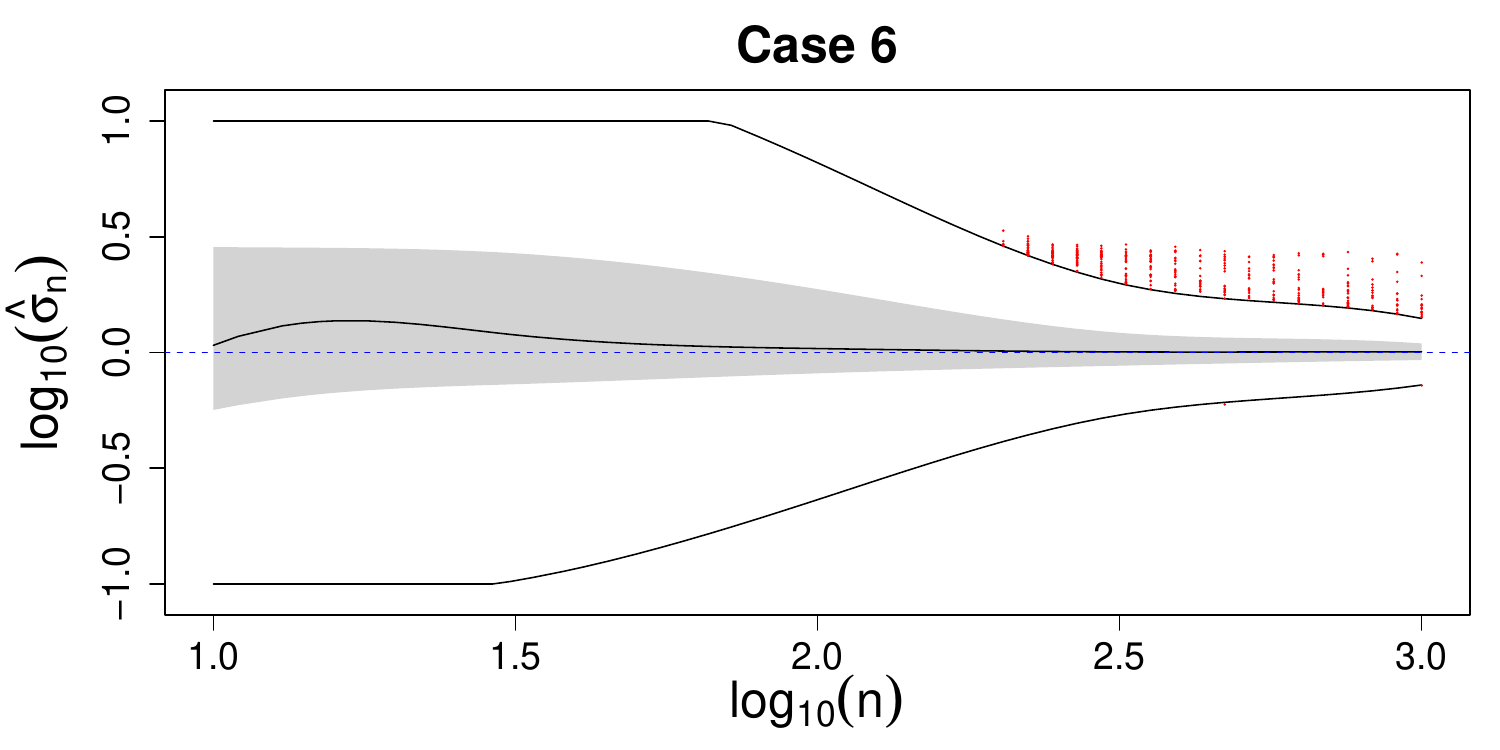}
    \end{subfigure}

    \caption{Results of the simulations illustrating the distribution of $\hat{\sigma}_n$ as $n$ increases.}
    \label{fig:sigma_cons}
\end{figure}

\begin{figure}[htbp]
    \centering

    \begin{subfigure}[b]{0.5\textwidth}
        \centering
        \includegraphics[width=\textwidth]{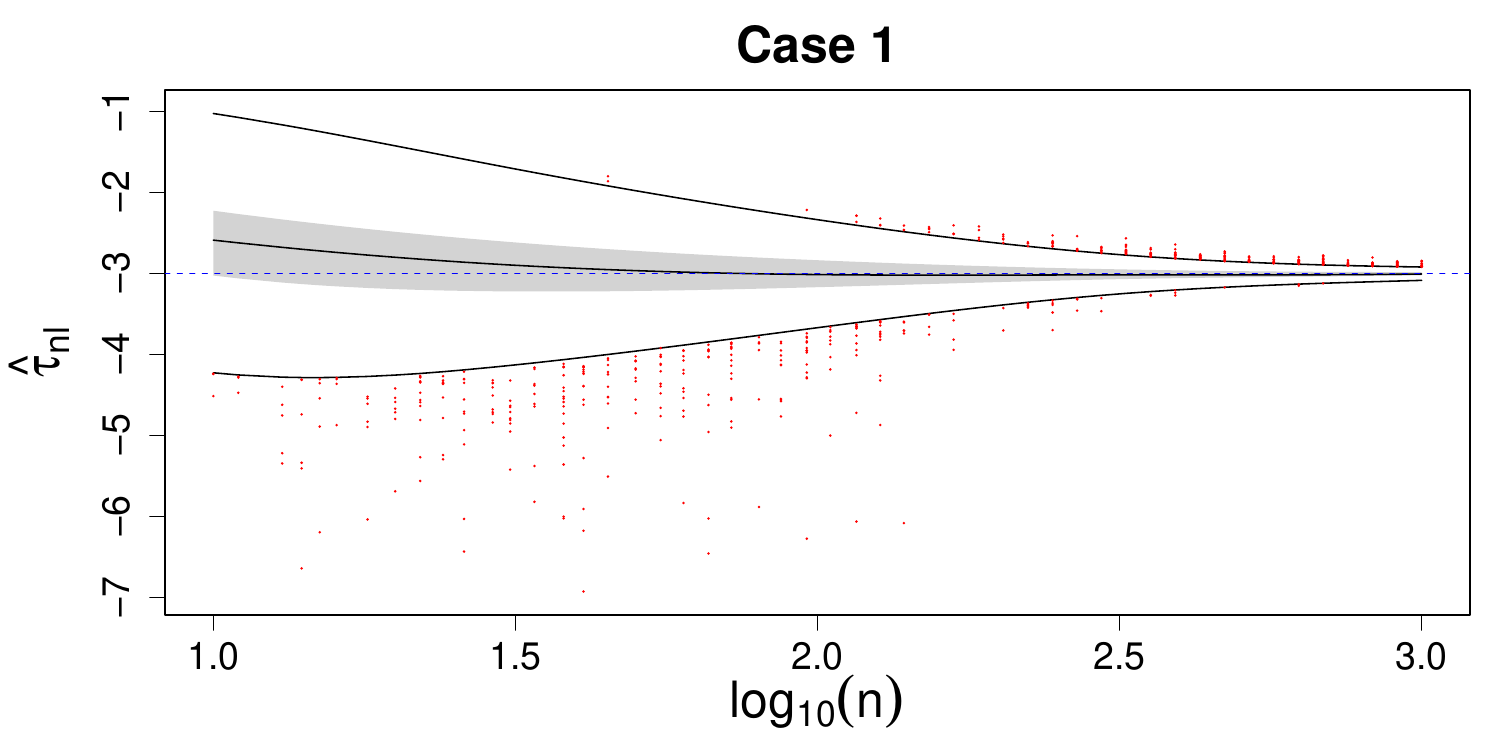}
    \end{subfigure}%
    \begin{subfigure}[b]{0.5\textwidth}
        \centering
        \includegraphics[width=\textwidth]{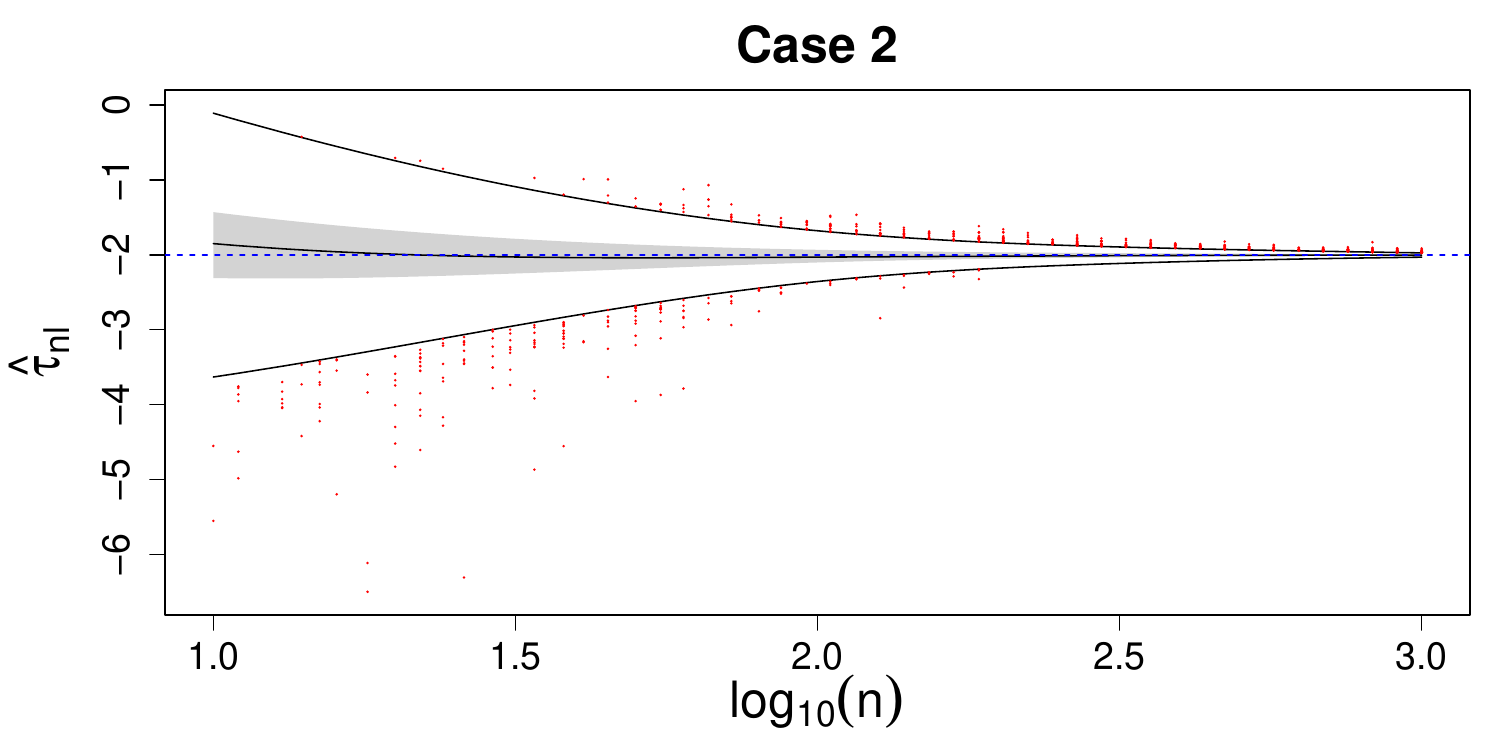}
    \end{subfigure}

    \begin{subfigure}[b]{0.5\textwidth}
        \centering
        \includegraphics[width=\textwidth]{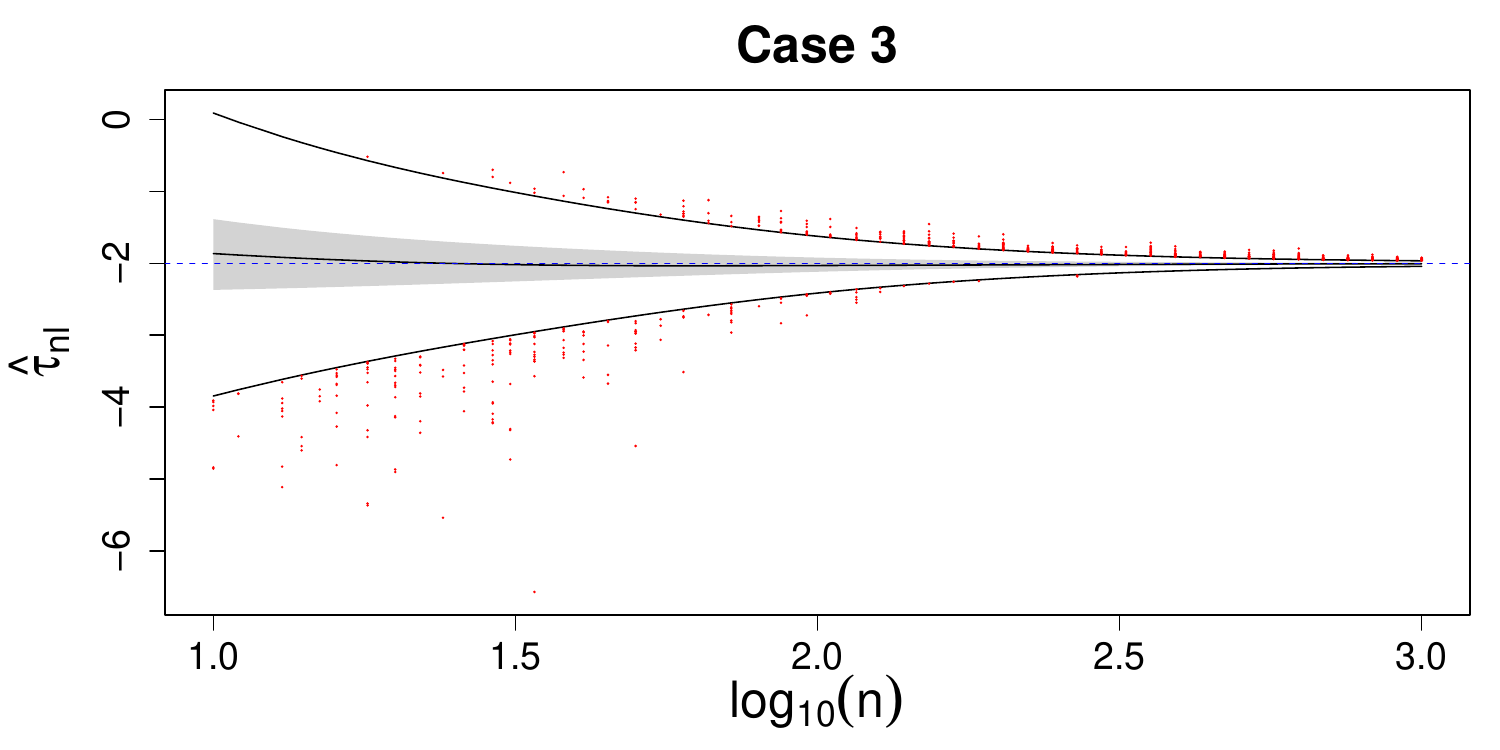}
    \end{subfigure}%
    \begin{subfigure}[b]{0.5\textwidth}
        \centering
        \includegraphics[width=\textwidth]{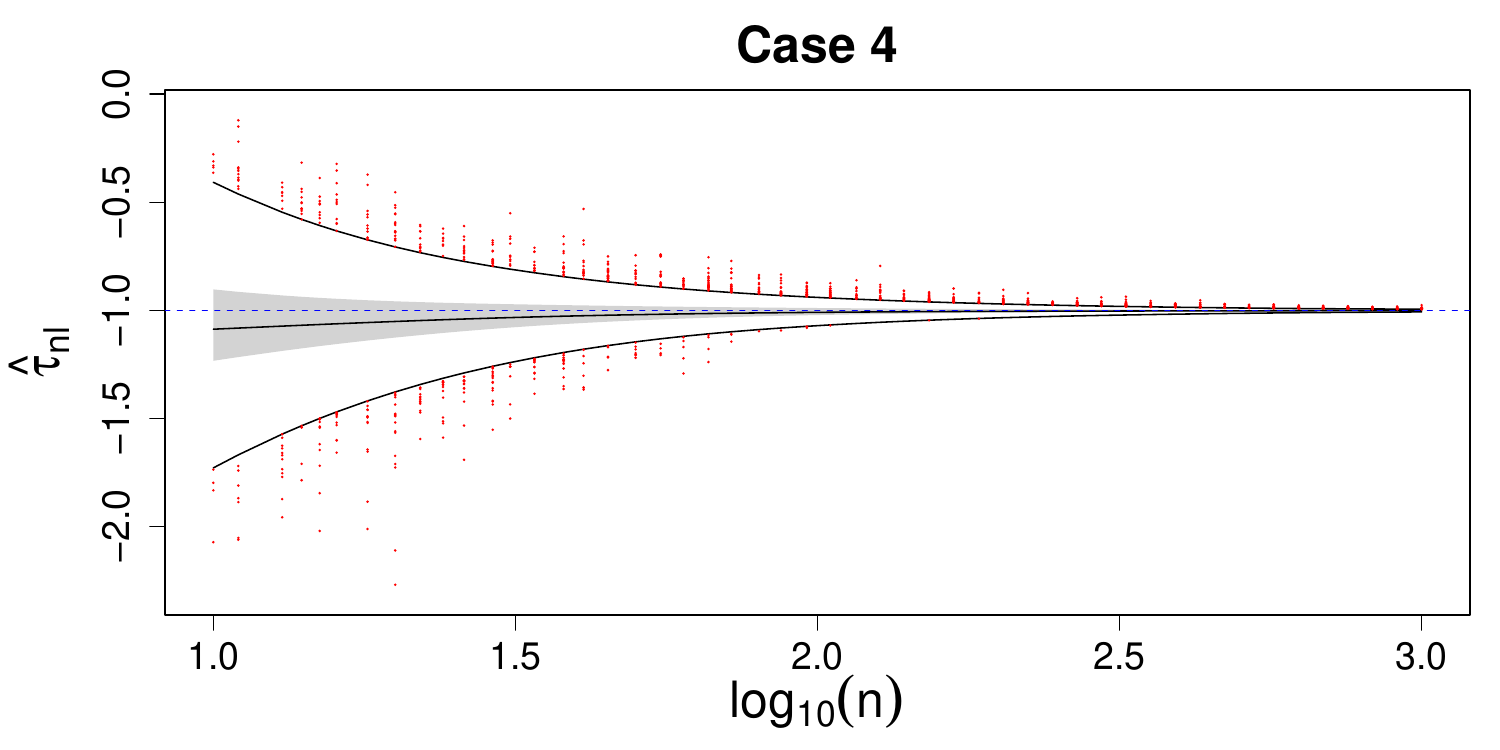}
    \end{subfigure}

    \begin{subfigure}[b]{0.5\textwidth}
        \centering
        \includegraphics[width=\textwidth]{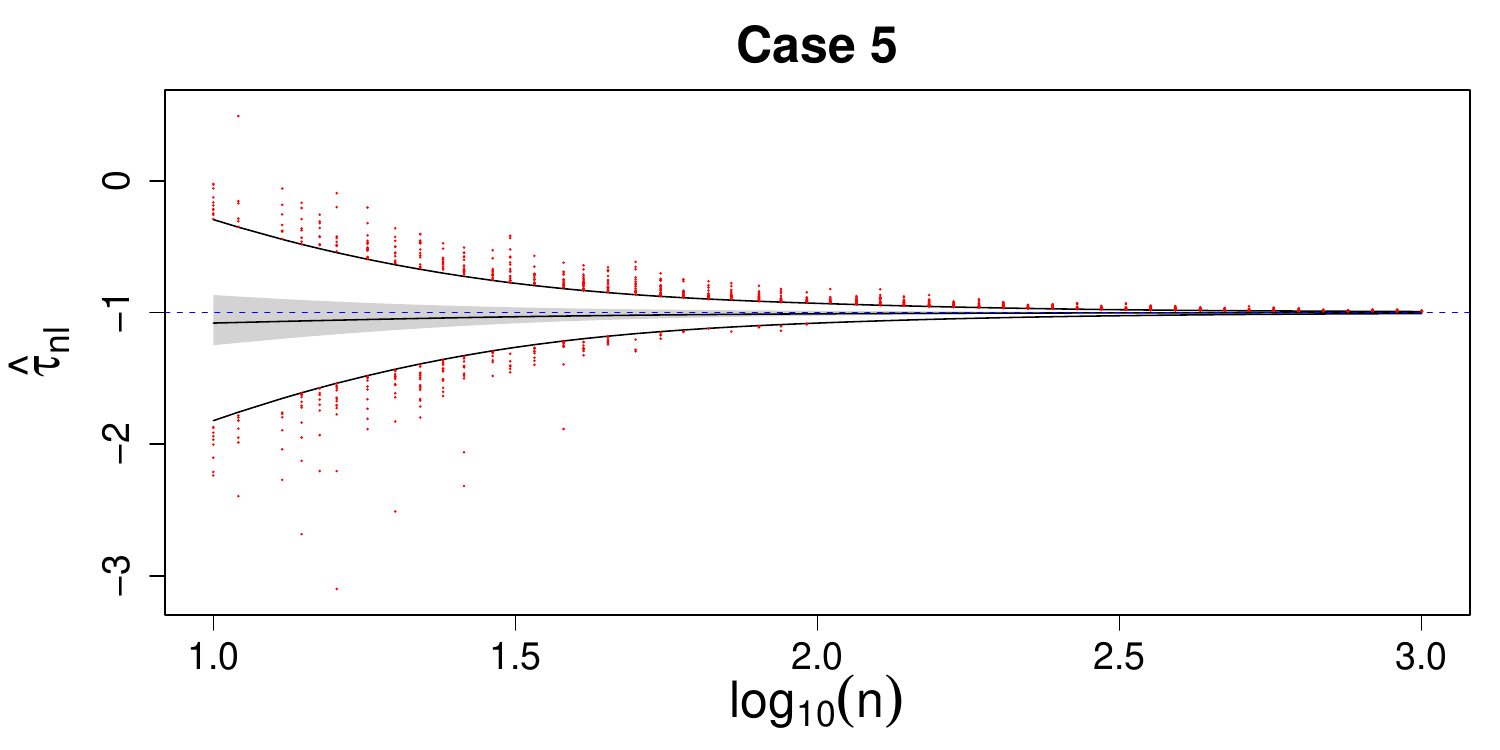}
    \end{subfigure}%
    \begin{subfigure}[b]{0.5\textwidth}
        \centering
        \includegraphics[width=\textwidth]{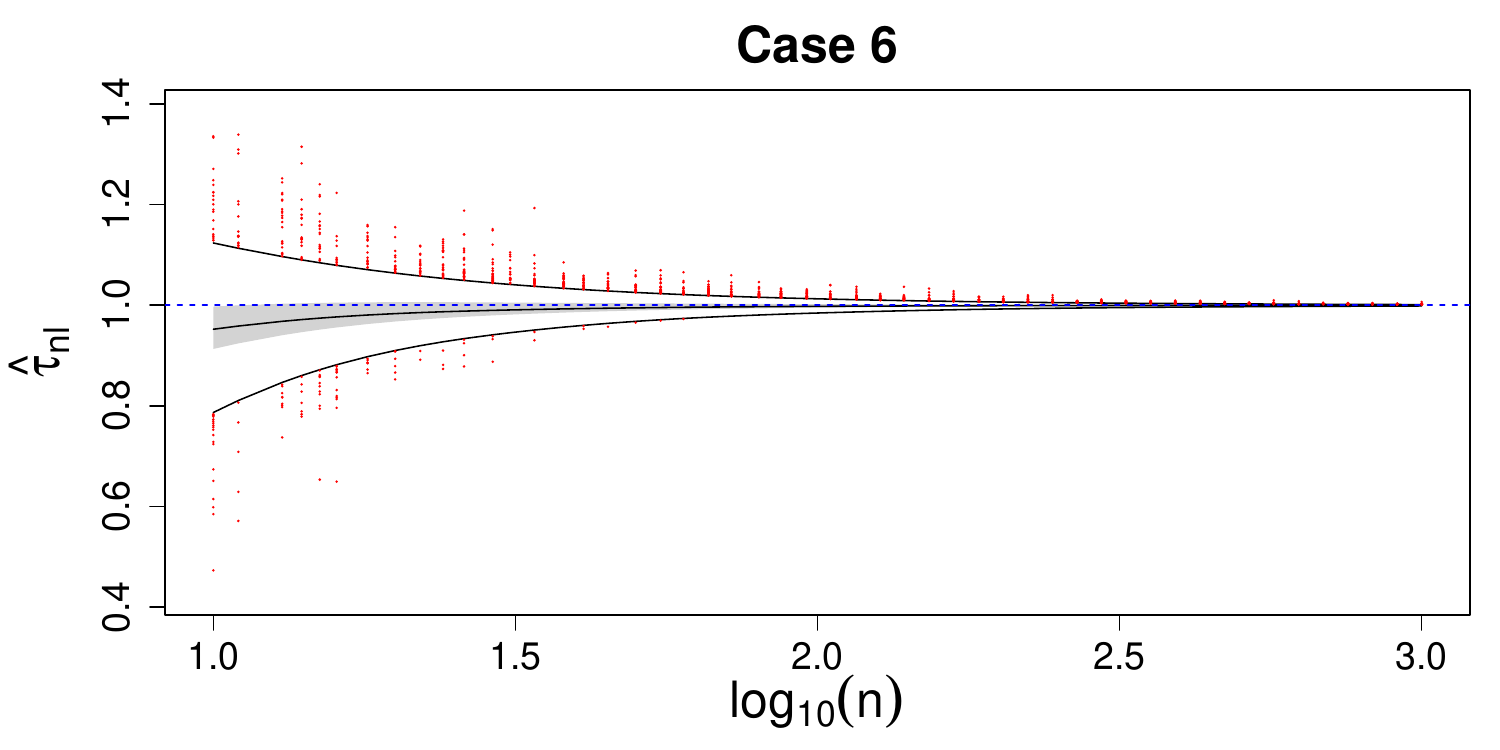}
    \end{subfigure}

    \caption{Results of the simulations illustrating the distribution of $\hat{\tau}_{nl}$ as $n$ increases.}
    \label{fig:tau_l_cons}
\end{figure}

\begin{figure}[htbp]
    \centering

    \begin{subfigure}[b]{0.5\textwidth}
        \centering
        \includegraphics[width=\textwidth]{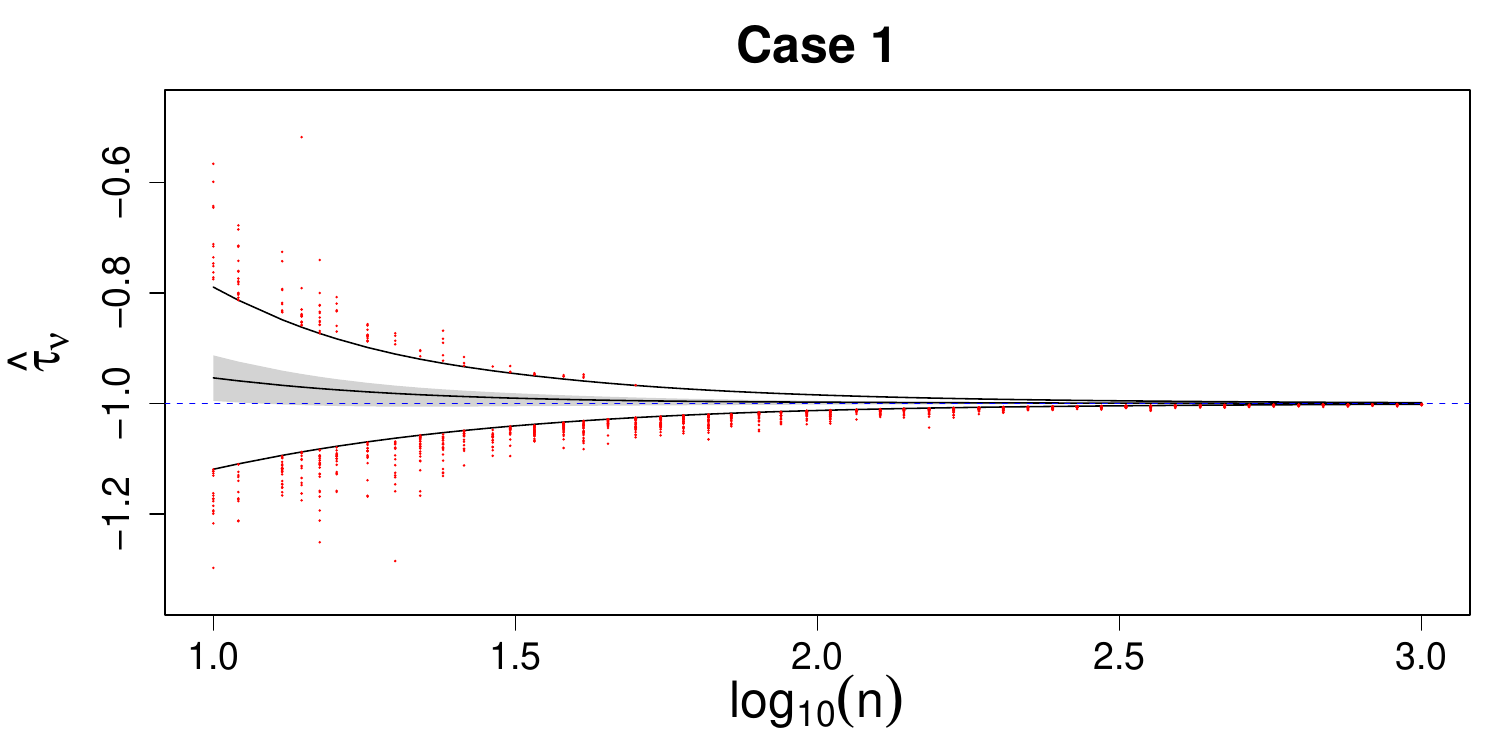}
    \end{subfigure}%
    \begin{subfigure}[b]{0.5\textwidth}
        \centering
        \includegraphics[width=\textwidth]{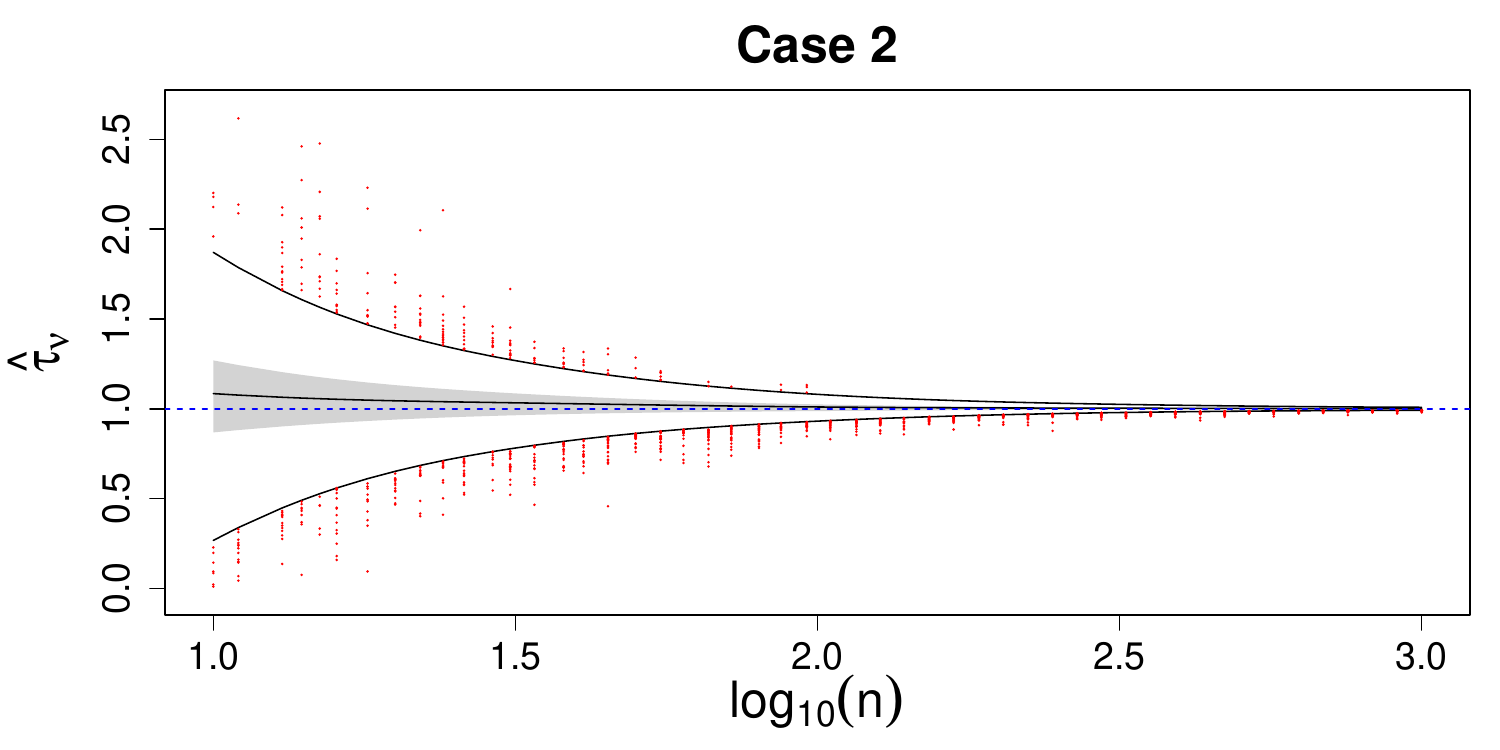}
    \end{subfigure}

    \begin{subfigure}[b]{0.5\textwidth}
        \centering
        \includegraphics[width=\textwidth]{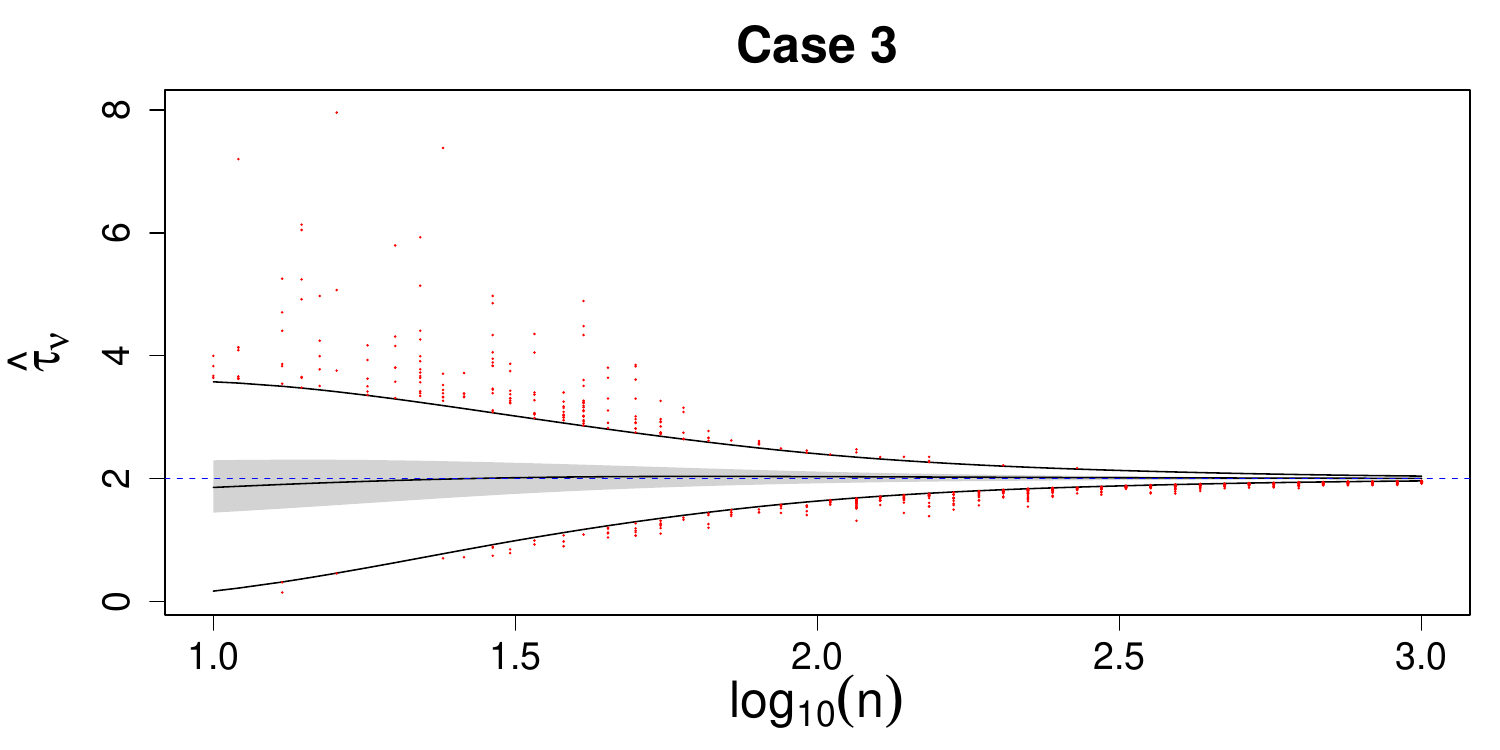}
    \end{subfigure}%
    \begin{subfigure}[b]{0.5\textwidth}
        \centering
        \includegraphics[width=\textwidth]{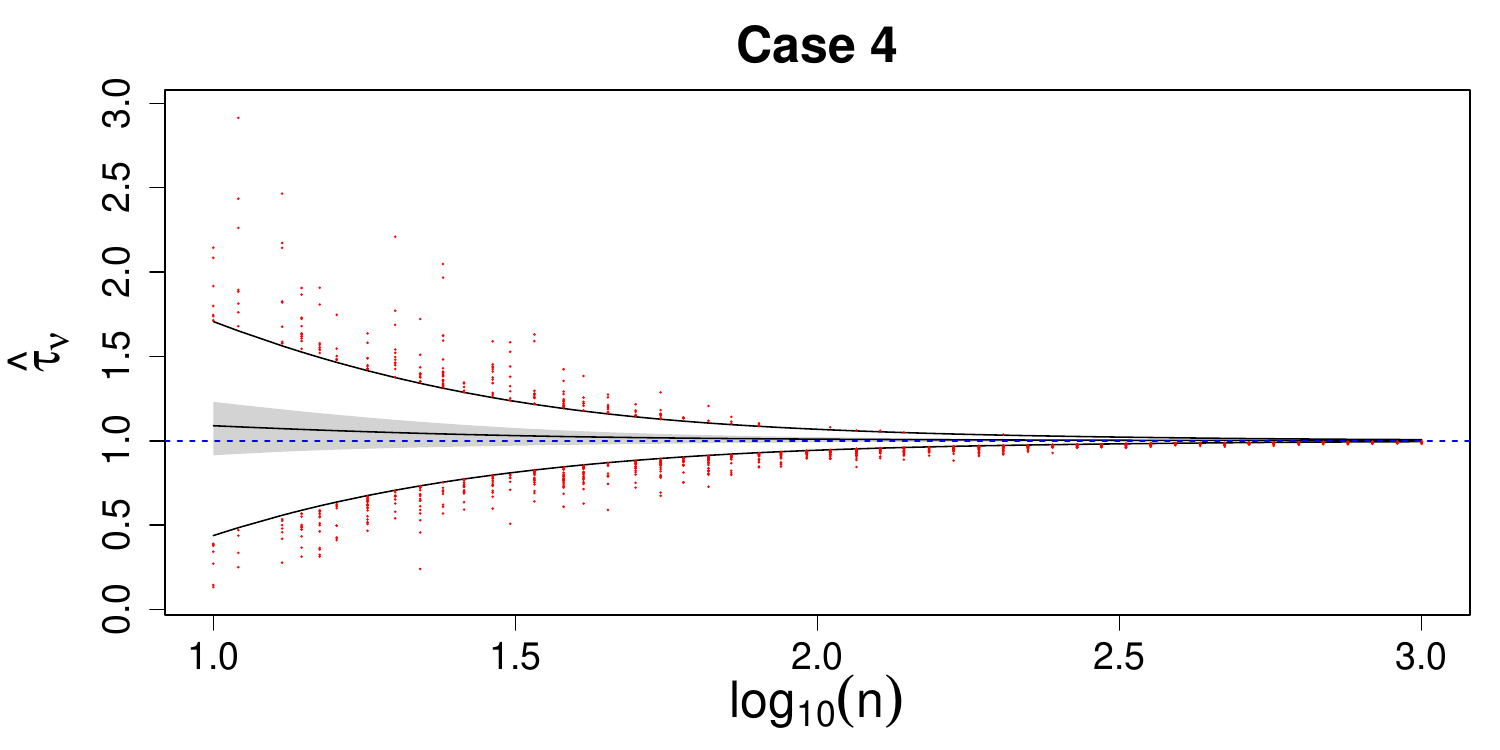}
    \end{subfigure}

    \begin{subfigure}[b]{0.5\textwidth}
        \centering
        \includegraphics[width=\textwidth]{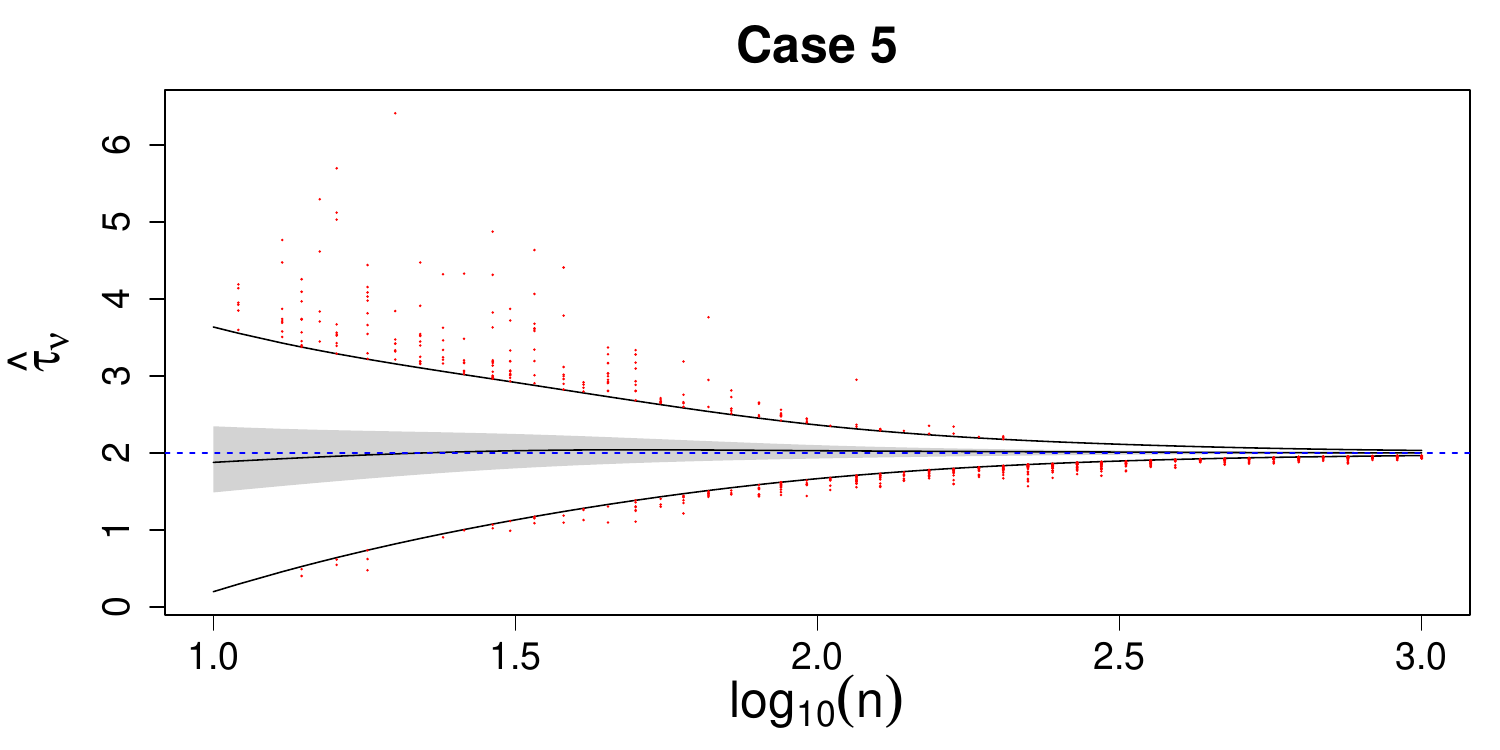}
    \end{subfigure}%
    \begin{subfigure}[b]{0.5\textwidth}
        \centering
        \includegraphics[width=\textwidth]{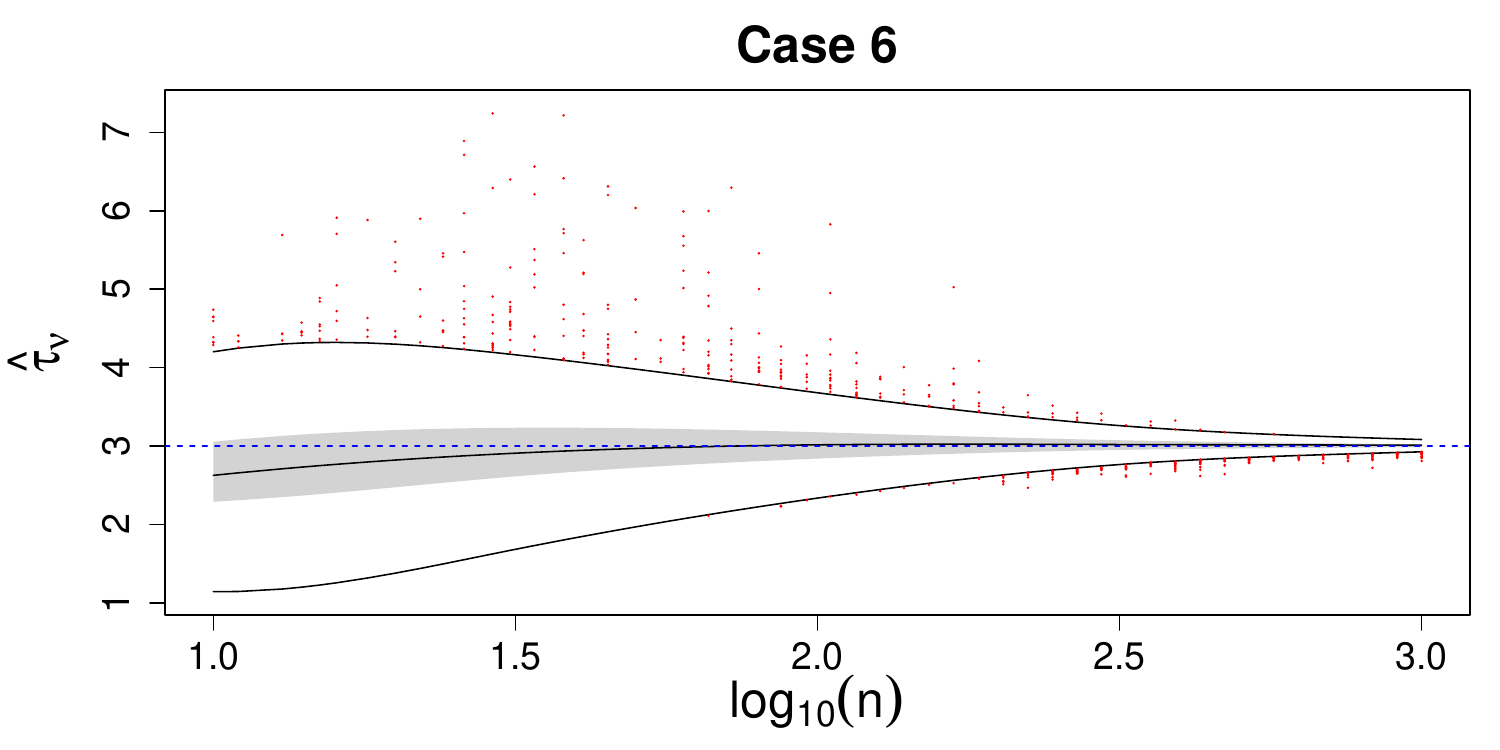}
    \end{subfigure}

    \caption{Results of the simulations illustrating the distribution of $\hat{\tau}_{nu}$ as $n$ increases.}
    \label{fig:tau_u_cons}
\end{figure}

Investigating the scenarios where the estimators lie at the boundary imposed on the parameter space. Figure \ref{fig:weirdcase} shows the comparison of the fitted CDF and True CDF to the ECDF of the sample used to get the estimators in two of these scenarios, illustrating two different boundary behaviors discussed by Mittal \cite{mittal-dissertation}. Where two distributions seen at the boundary of the parameter space are a truncated exponential type distribution with $F(x) \propto \exp\{cx\}$ where $\mu \rightarrow \pm\infty$ and $\sigma \rightarrow \infty$ such that $\frac{\mu}{\sigma^2}\rightarrow c$ and a uniform distribution when $\mu \rightarrow \infty$ and $\sigma \rightarrow \infty$. The left plot shows when the estimators take values $\hat{\mu}_n = 10$, $\hat{\sigma}_n = 2.37$, $\hat{\tau}_{nl}=-2.82$, and $\hat{\tau}_{nu}=-1.01$ and the right plot illustrates when $\hat{\mu}_n = 10$, $\hat{\sigma}_n = 10$, $\hat{\tau}_{nl}=-1.01$, and $\hat{\tau}_{nu}=0.94$. Overall, the fitted CDFs do fit the ECDFs quite well. As a note, there were instances in the simulations within Cases 1 and 6, where the fitted distribution had the correct shape, moving towards a truncated exponential distribution, but was able to capture the exact shape that smoothed the ECDF, likely due to the bounding of the parameter space being too conservative. Figure \ref{fig:weirdcasepoorfit} illustrates two of these scenarios with the truncated exponential type distributions, one with $\hat{\mu}_n = 10$ and one with $\hat{\mu}_n = -10$.

Moving to the results of the Monte Carlo simulations used to study the asymptotic distribution of the estimators, Figures \ref{fig:tau_l_asym_dist} and \ref{fig:tau_u_asym_dist} show the results for sample sizes of 30, 50, and 100. The plots compare the CDF of the limiting distribution given in Theorem \ref{thm:limiting_dist_of_bounds} to the ECDF of the  10,000 replications at each sample size. Looking at Figure \ref{fig:tau_l_asym_dist}, the ECDFs are not close to the limiting CDF in Cases 1 and 2, but are close to the limiting CDF for Cases 5 and 6, which is expected as there is less mass near the lower bound in Cases 1 and 2 and more mass near the lower bounds in these Cases 5 and 6. Figure \ref{fig:tau_u_asym_dist} shows the opposite, which is expected for the same reason. Both of the Figures show similar behavior in Cases 3 and 4, which is again expected as these cases have symmetric truncation with the $\mu$ in the center of the truncation interval.

\begin{figure}[htbp]
    \centering
    
    \begin{subfigure}[b]{0.5\textwidth}
        \centering
        \includegraphics[width=\textwidth]{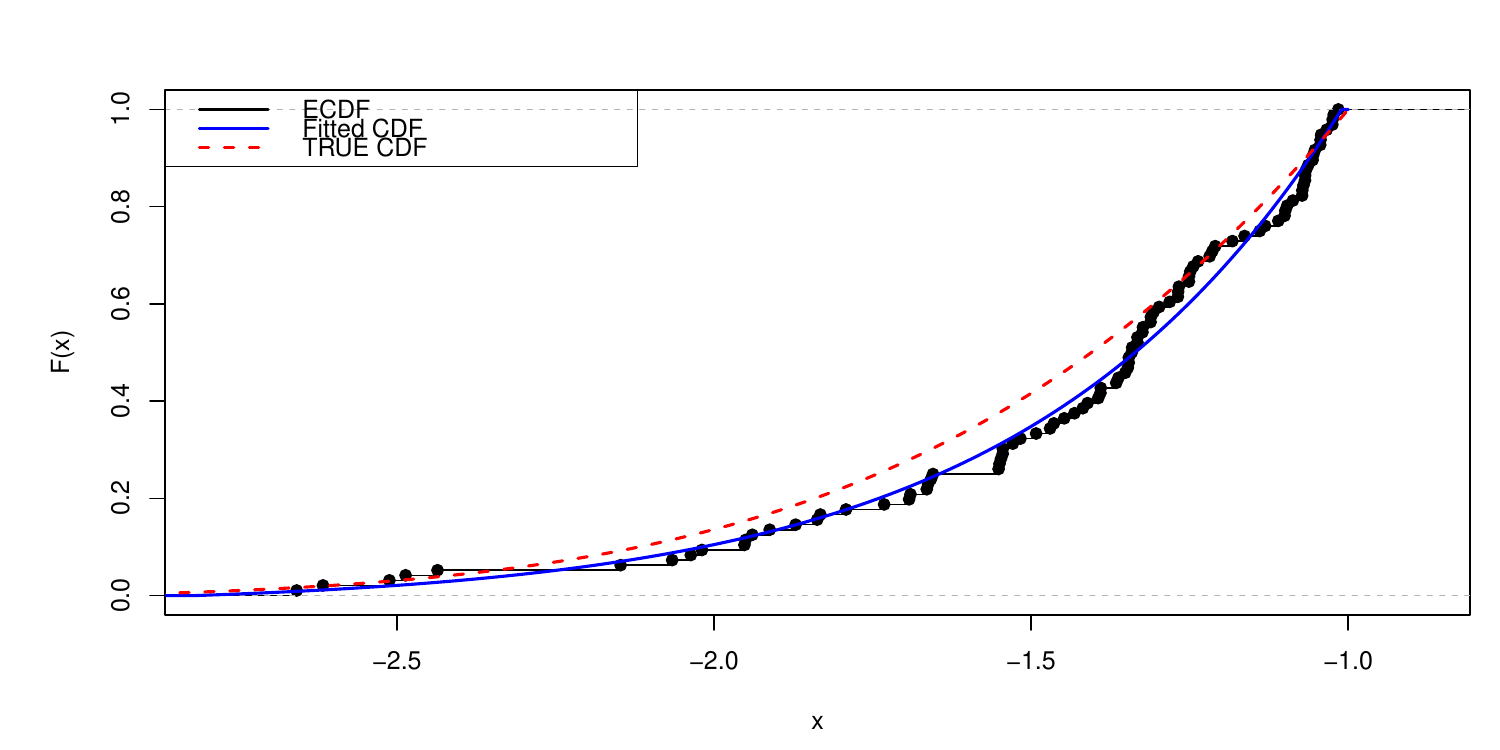}
    \end{subfigure}%
    \begin{subfigure}[b]{0.5\textwidth}
        \centering
        \includegraphics[width=\textwidth]{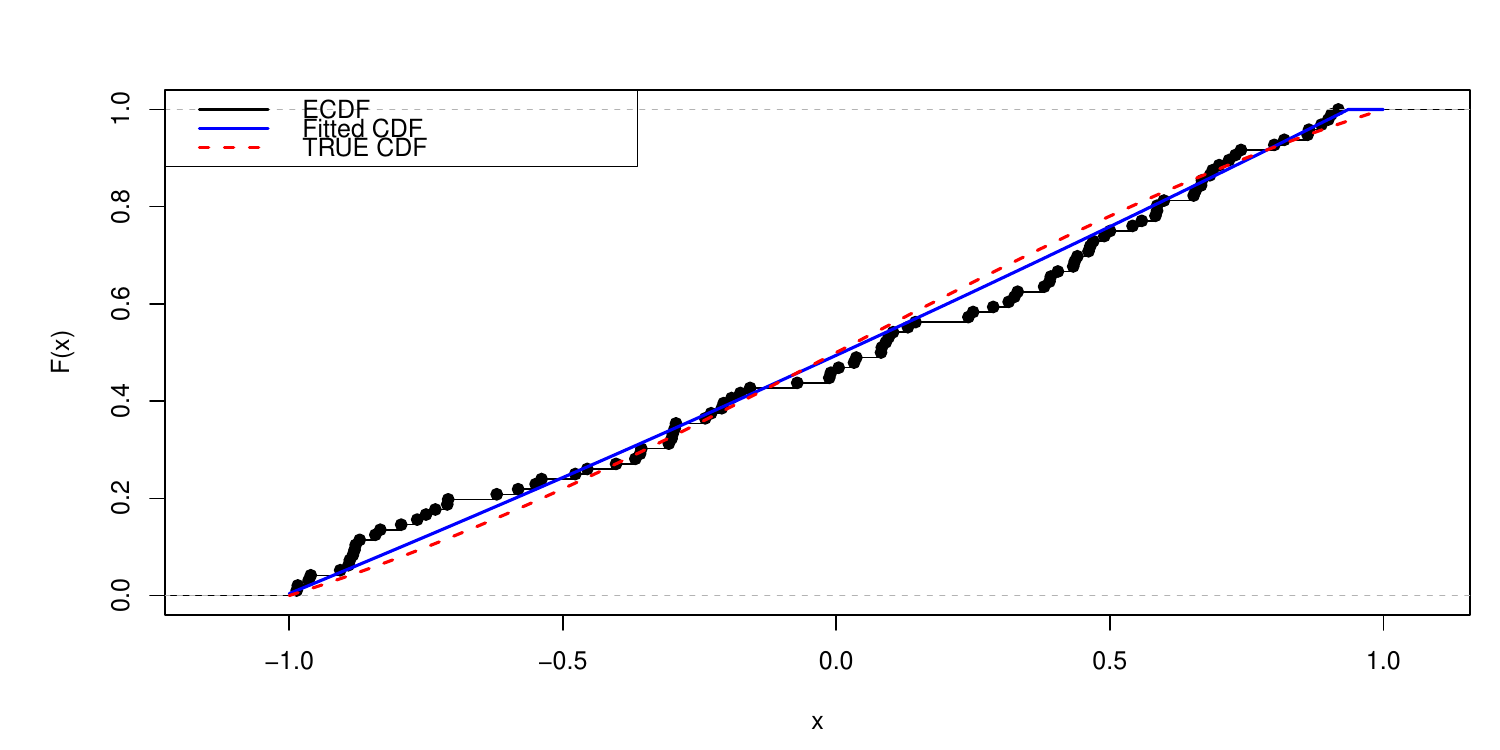}
    \end{subfigure}
    \caption{Visualization of cases where the estimator lies on the boundary imposed on the parameter space, the ECDF, fitted CDF, and true generative CDF are shown in solid black, solid blue, and dashed red lines, respectively.}
    \label{fig:weirdcase}
\end{figure}

\begin{figure}[htbp]
    \centering
    
    \begin{subfigure}[b]{0.5\textwidth}
        \centering
        \includegraphics[width=\textwidth]{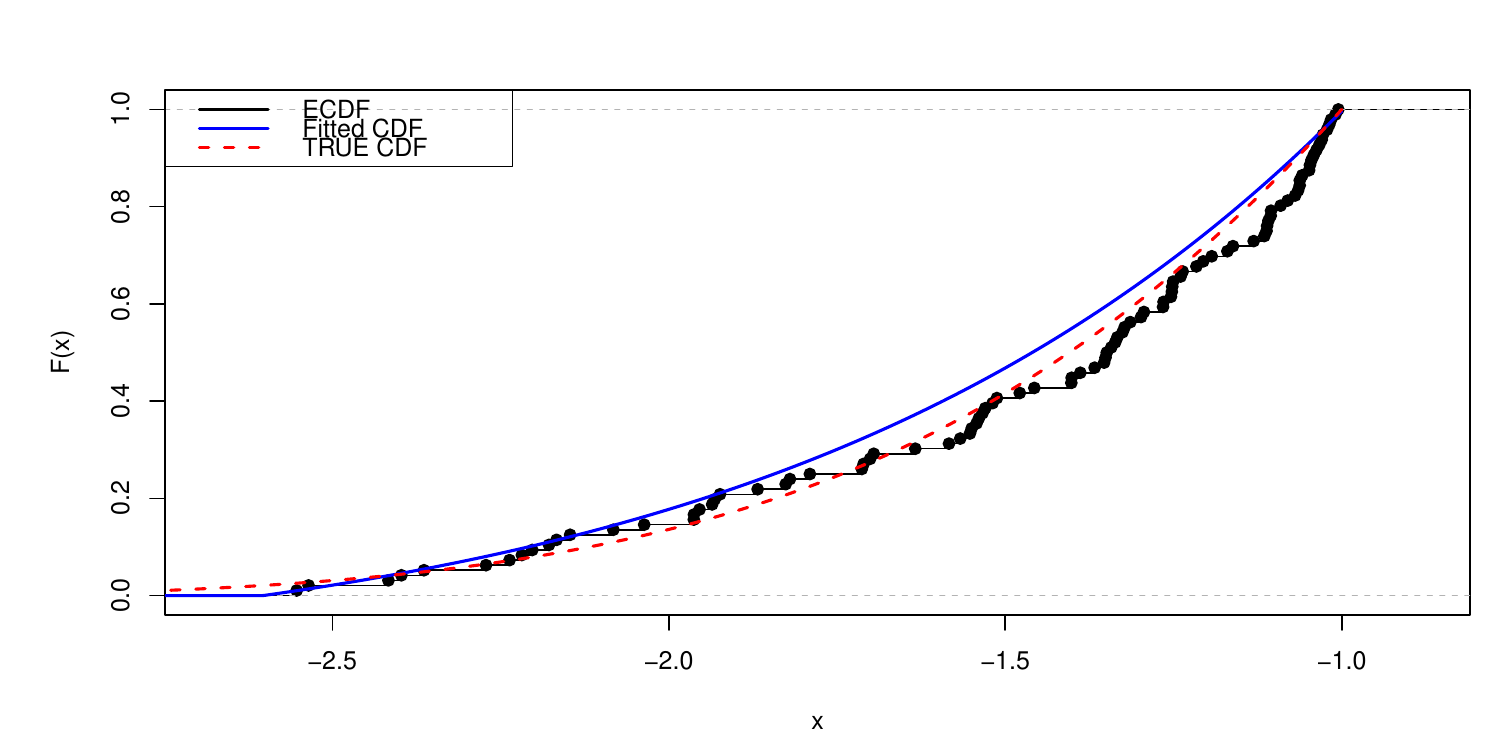}
    \end{subfigure}%
    \begin{subfigure}[b]{0.5\textwidth}
        \centering
        \includegraphics[width=\textwidth]{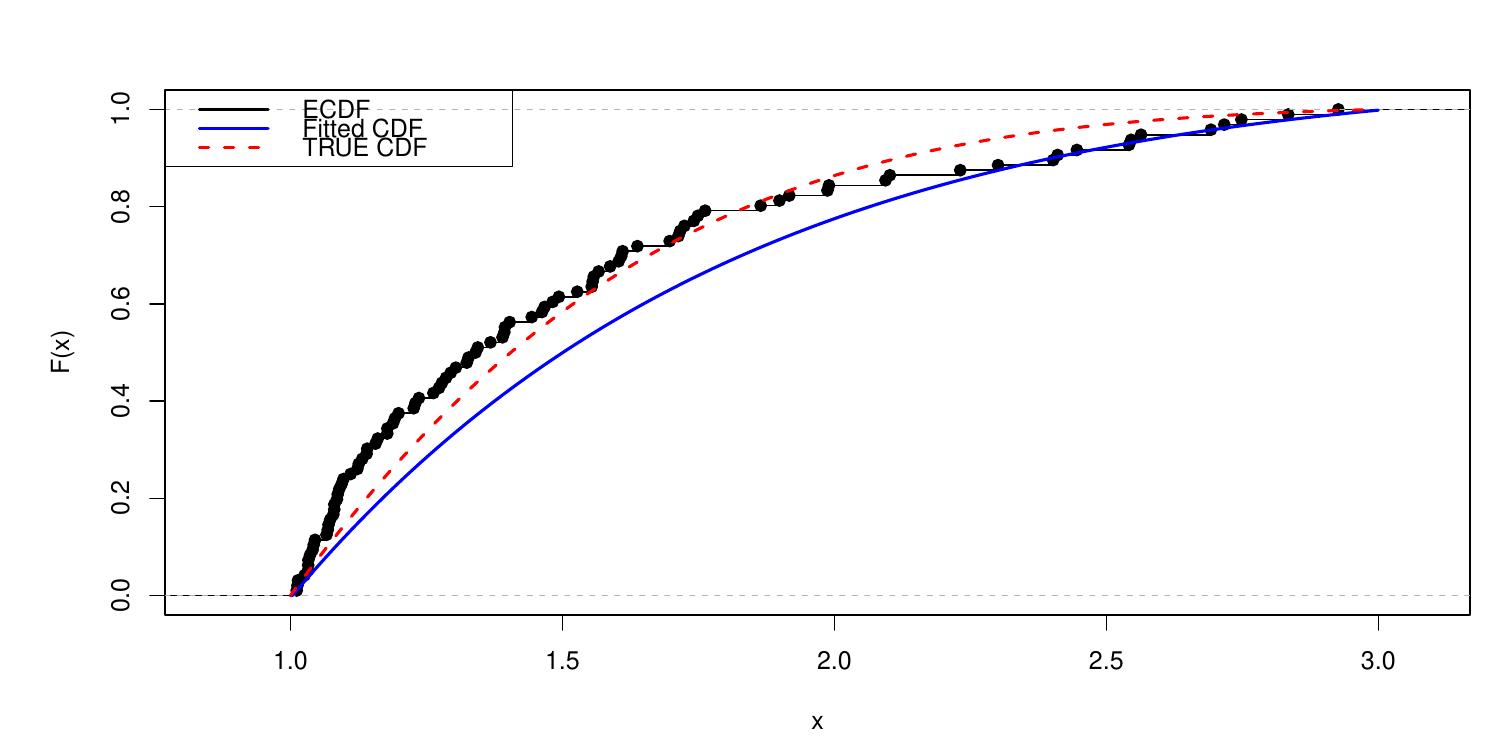}
    \end{subfigure}
    \caption{Visualization of cases where the estimator lies on the boundary imposed on the parameter space, but the fitted distribution visually does not fit well, the ECDF, fitted CDF, and true generative CDF are shown in solid black, solid blue, and dashed red lines, respectively.}
    \label{fig:weirdcasepoorfit}
\end{figure}

\begin{figure}[htbp]
    \centering

    \begin{subfigure}[b]{0.5\textwidth}
        \centering
        \includegraphics[width=\textwidth]{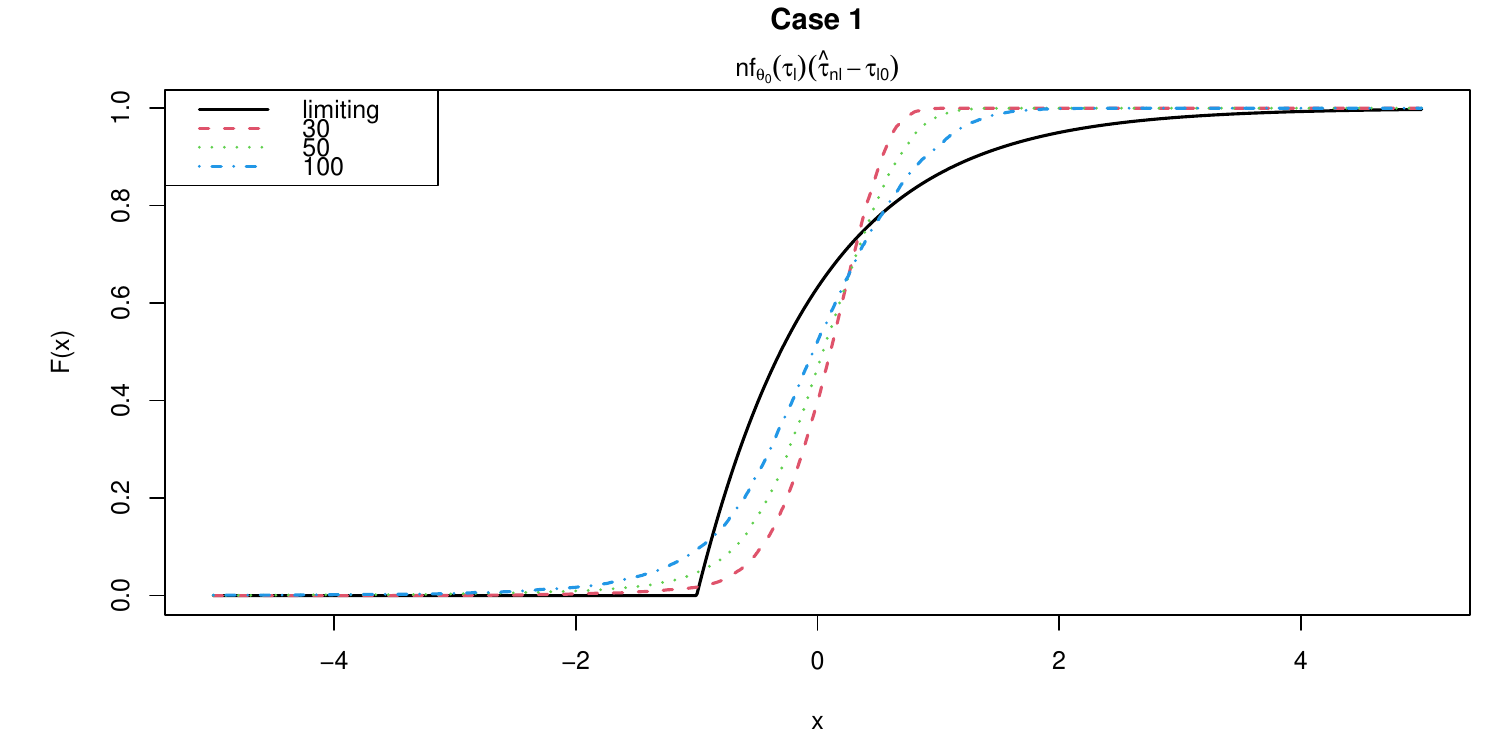}
    \end{subfigure}%
    \begin{subfigure}[b]{0.5\textwidth}
        \centering
        \includegraphics[width=\textwidth]{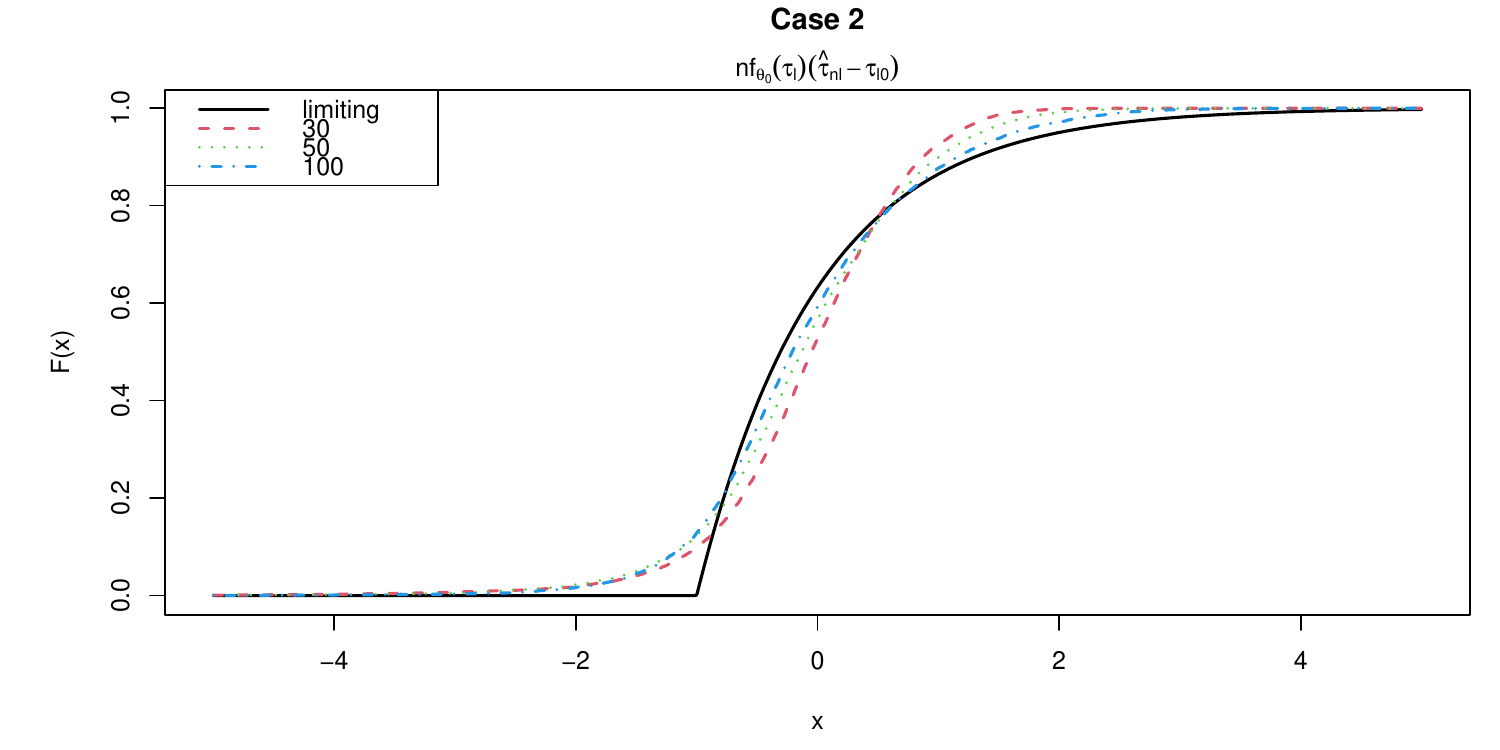}
    \end{subfigure}

    \begin{subfigure}[b]{0.5\textwidth}
        \centering
        \includegraphics[width=\textwidth]{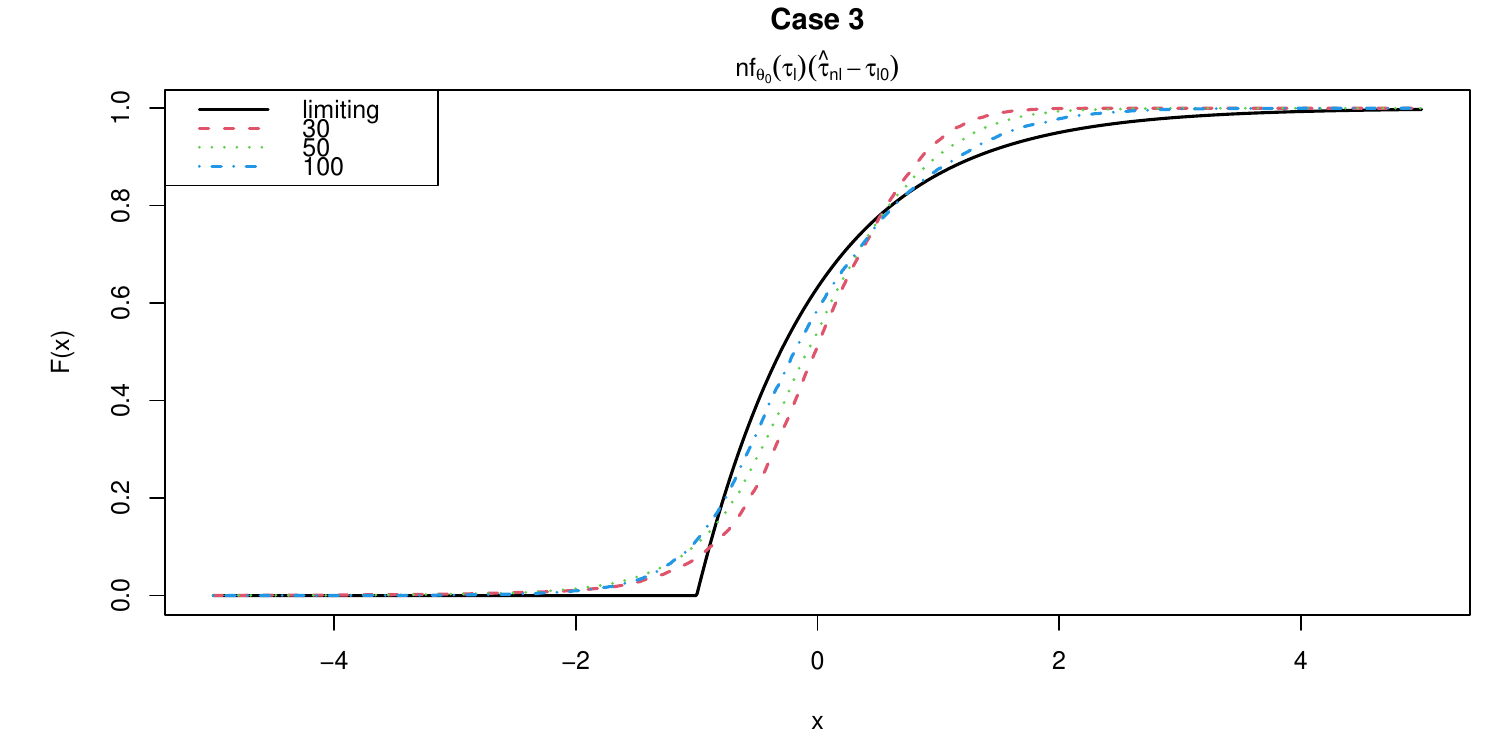}
    \end{subfigure}%
    \begin{subfigure}[b]{0.5\textwidth}
        \centering
        \includegraphics[width=\textwidth]{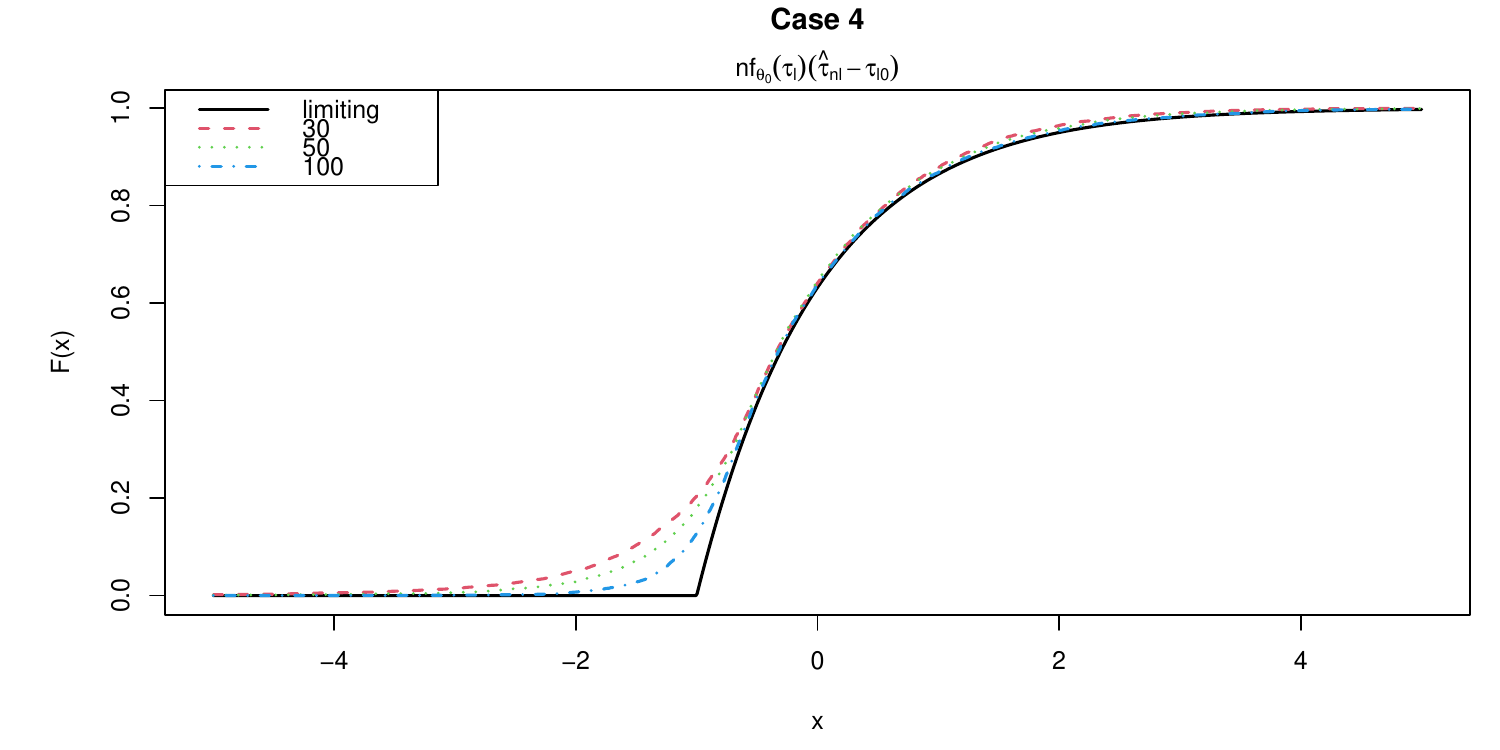}
    \end{subfigure}

    \begin{subfigure}[b]{0.5\textwidth}
        \centering
        \includegraphics[width=\textwidth]{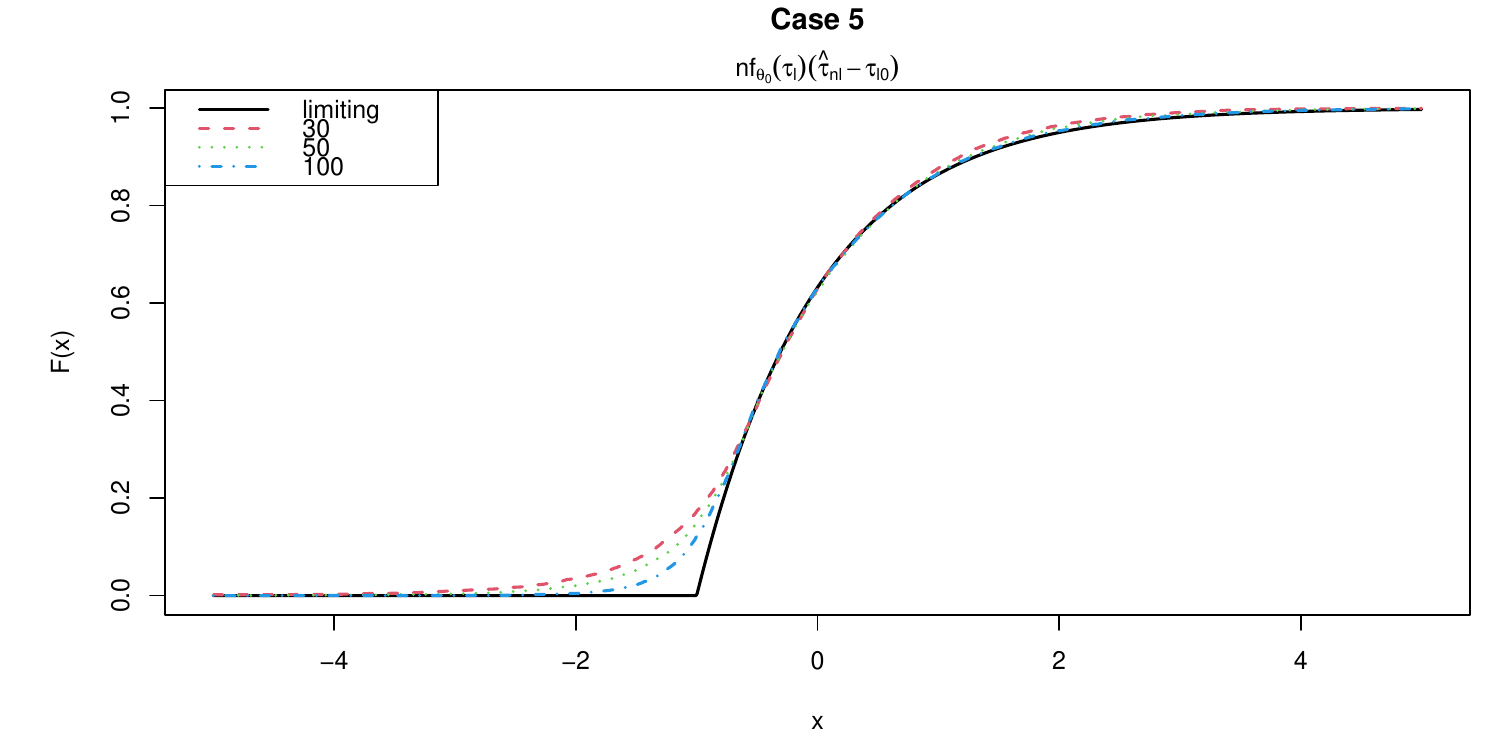}
    \end{subfigure}%
    \begin{subfigure}[b]{0.5\textwidth}
        \centering
        \includegraphics[width=\textwidth]{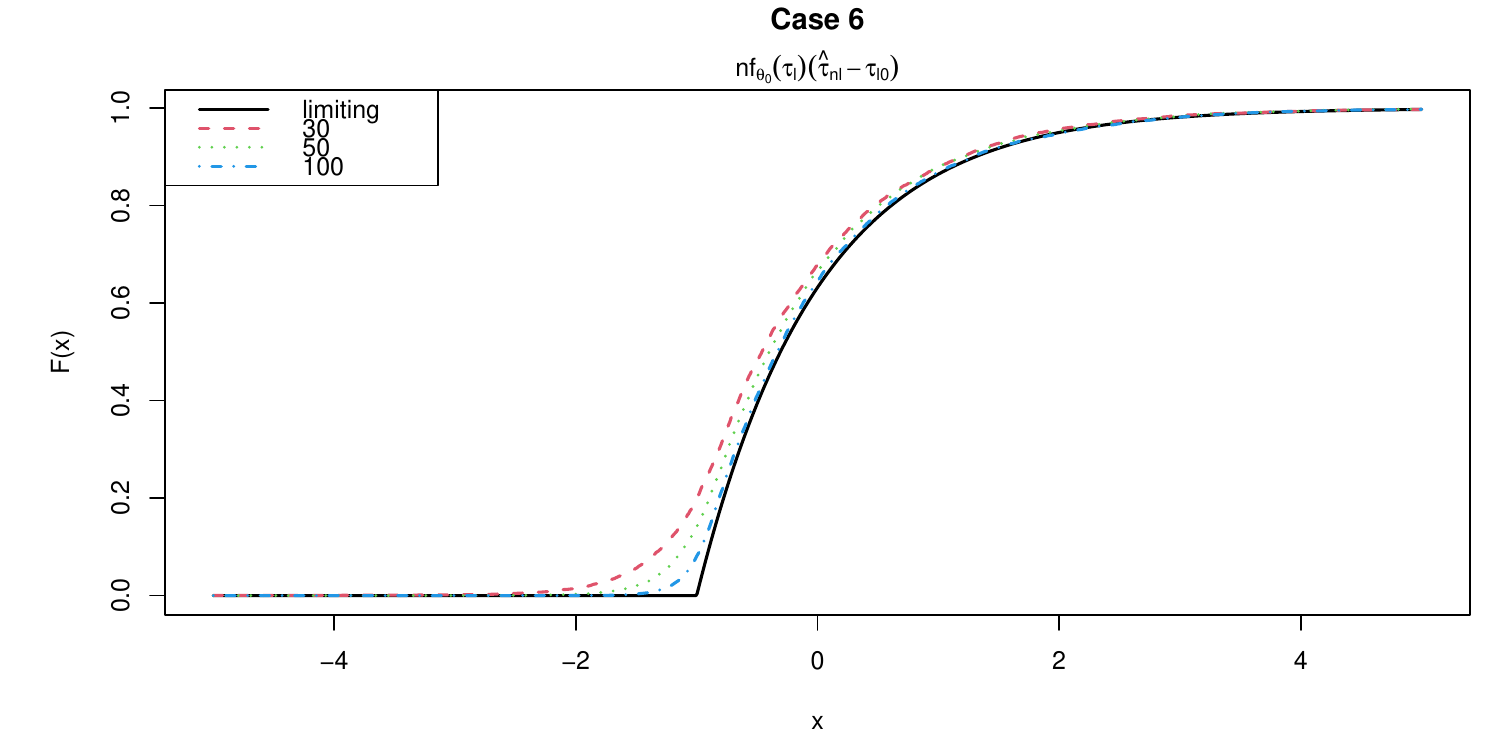}
    \end{subfigure}

    \caption{Results of the simulations illustrating the asymptotic distribution of $nf_{\btheta_0}(\tau_{l0})(\hat{\tau}_{nl}-\tau_{l0})$ as $n\rightarrow \infty$. The black solid, red dashed, green dotted, and blue dot-dashed lines show the limiting CDF, and ECDF from 10,000 repetitions of the sample sizes of 30, 50, and 100, respectively.}
    \label{fig:tau_l_asym_dist}
\end{figure}

\begin{figure}[htbp]
    \centering

    \begin{subfigure}[b]{0.5\textwidth}
        \centering
        \includegraphics[width=\textwidth]{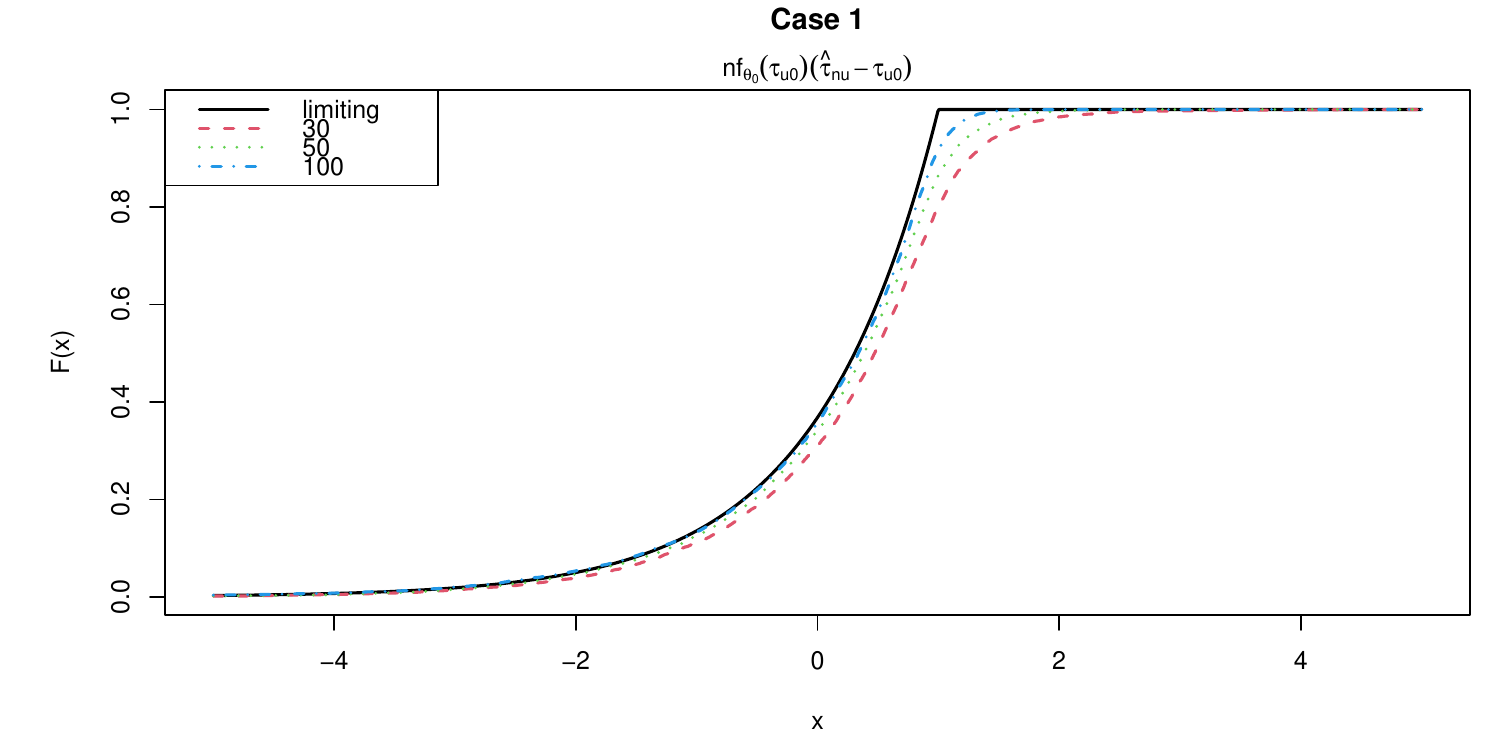}
    \end{subfigure}%
    \begin{subfigure}[b]{0.5\textwidth}
        \centering
        \includegraphics[width=\textwidth]{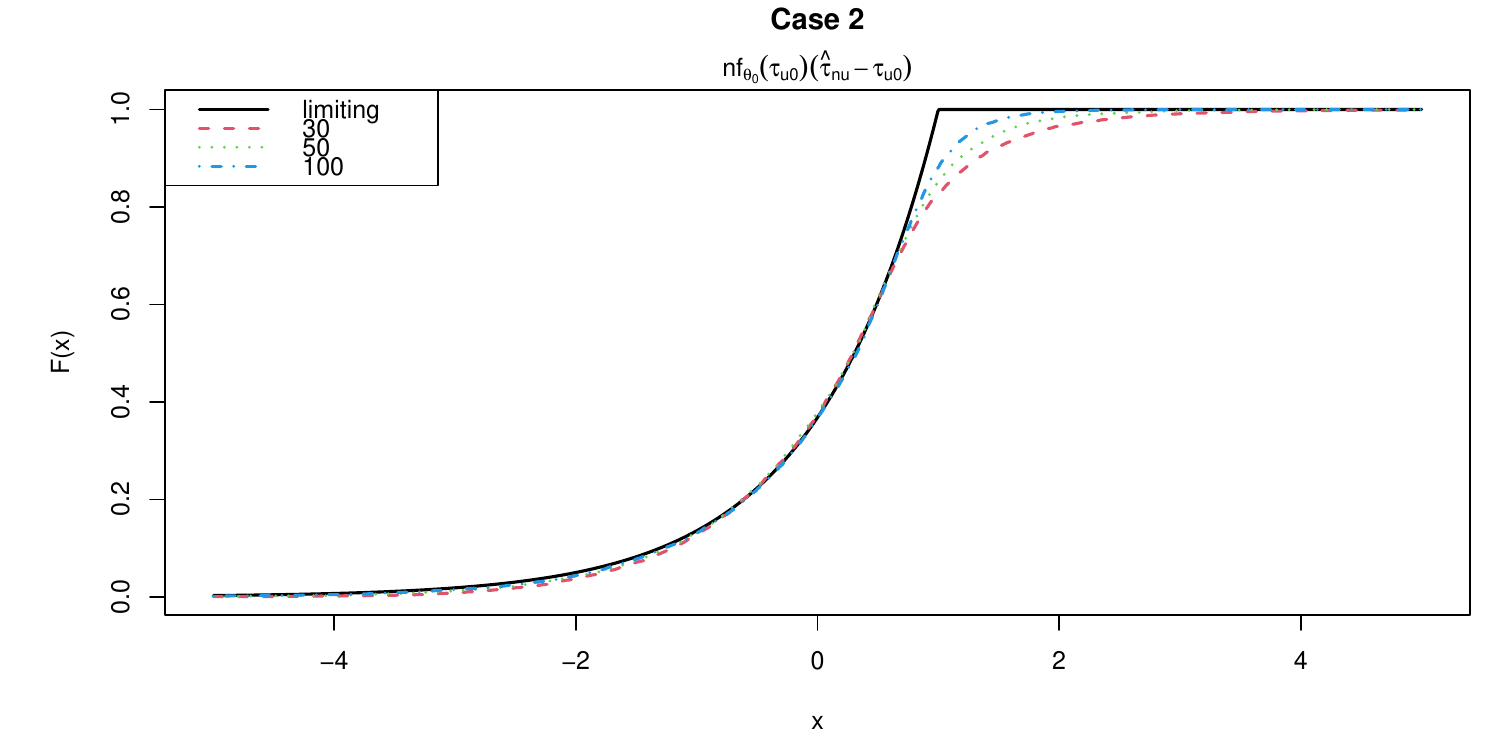}
    \end{subfigure}

    \begin{subfigure}[b]{0.5\textwidth}
        \centering
        \includegraphics[width=\textwidth]{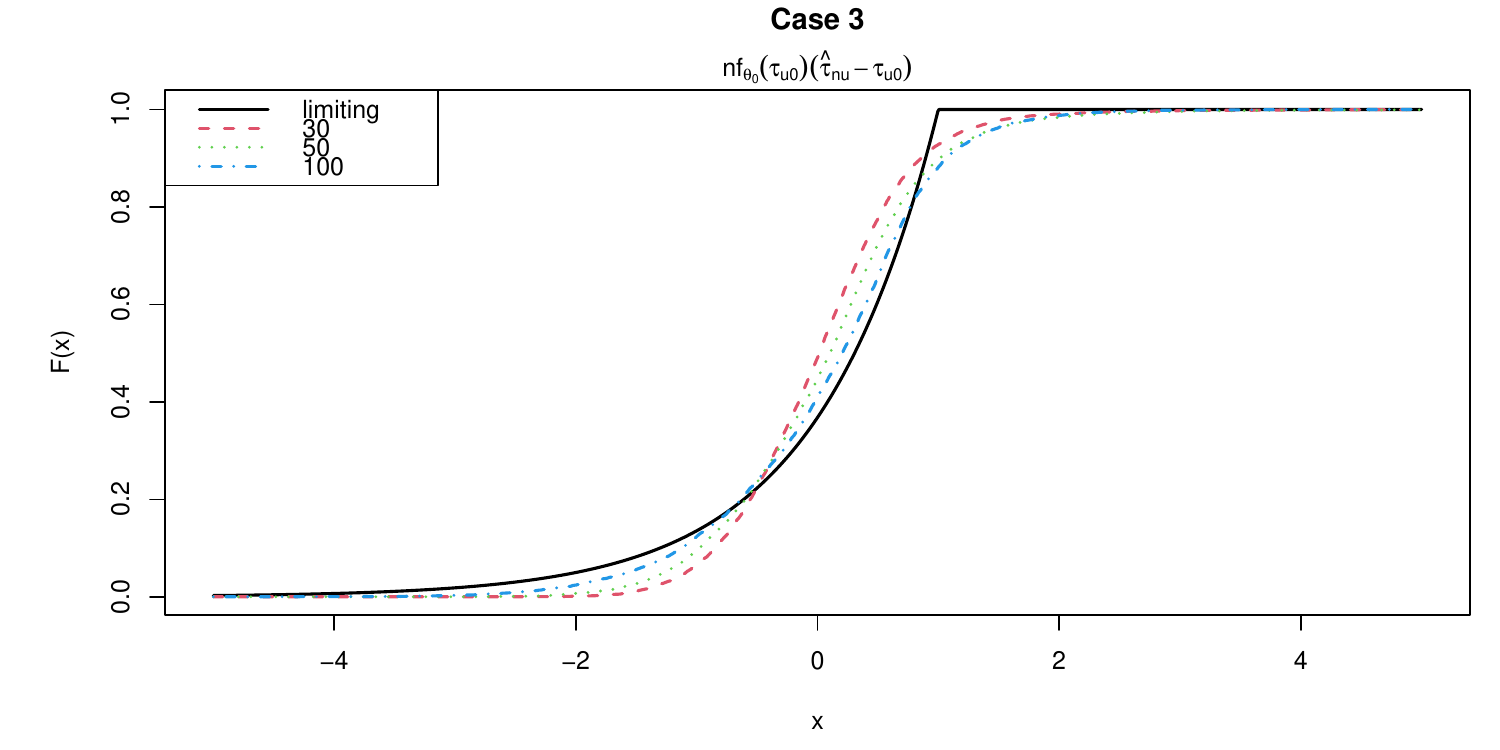}
    \end{subfigure}%
    \begin{subfigure}[b]{0.5\textwidth}
        \centering
        \includegraphics[width=\textwidth]{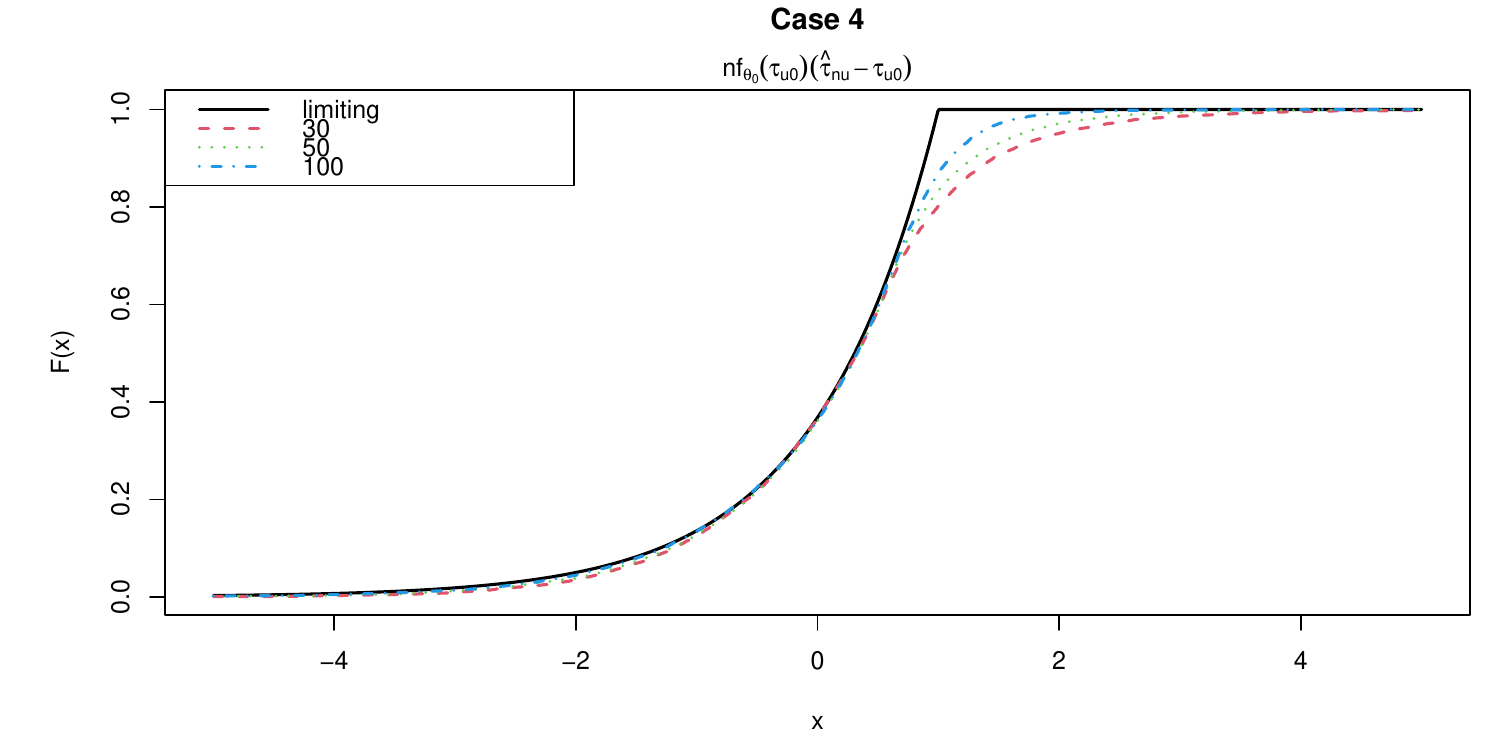}
    \end{subfigure}

    \begin{subfigure}[b]{0.5\textwidth}
        \centering
        \includegraphics[width=\textwidth]{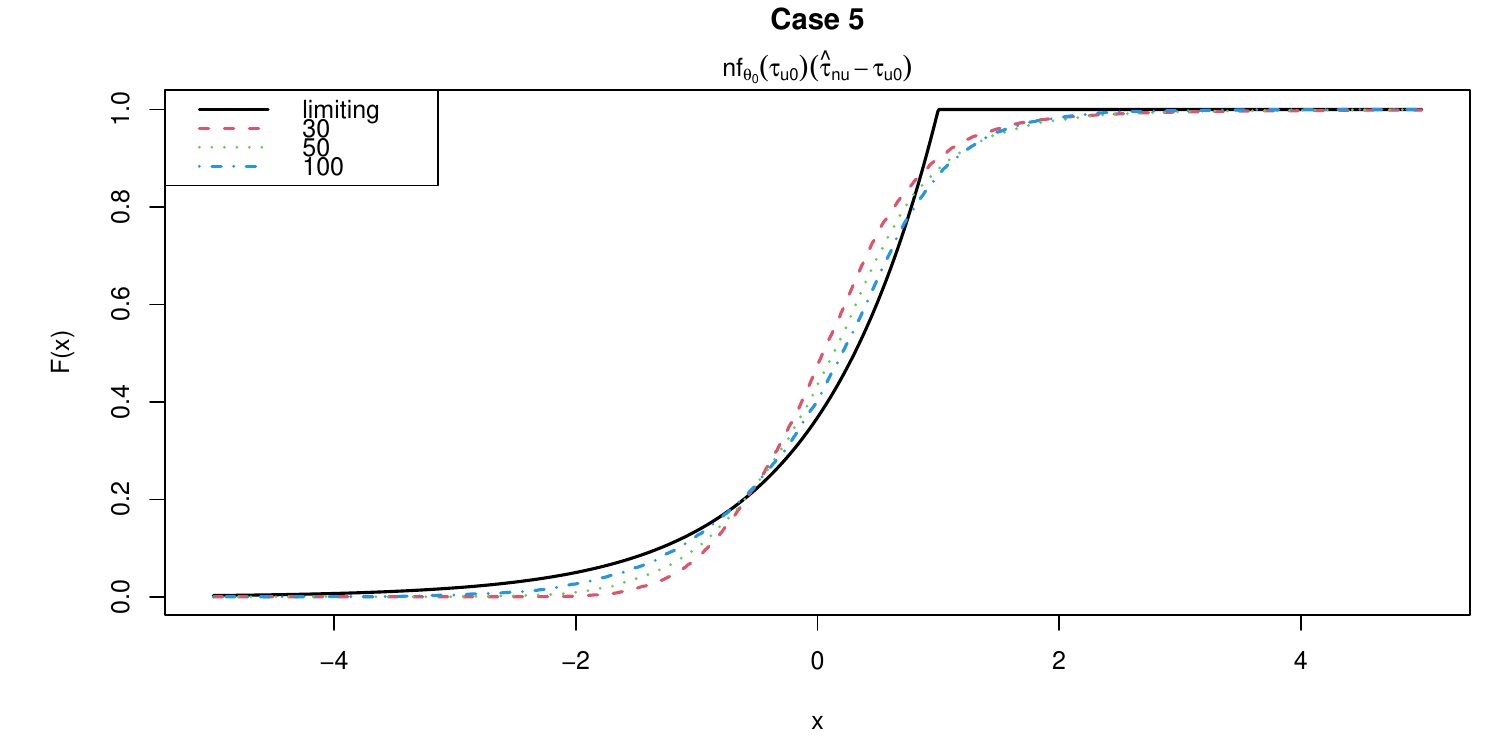}
    \end{subfigure}%
    \begin{subfigure}[b]{0.5\textwidth}
        \centering
        \includegraphics[width=\textwidth]{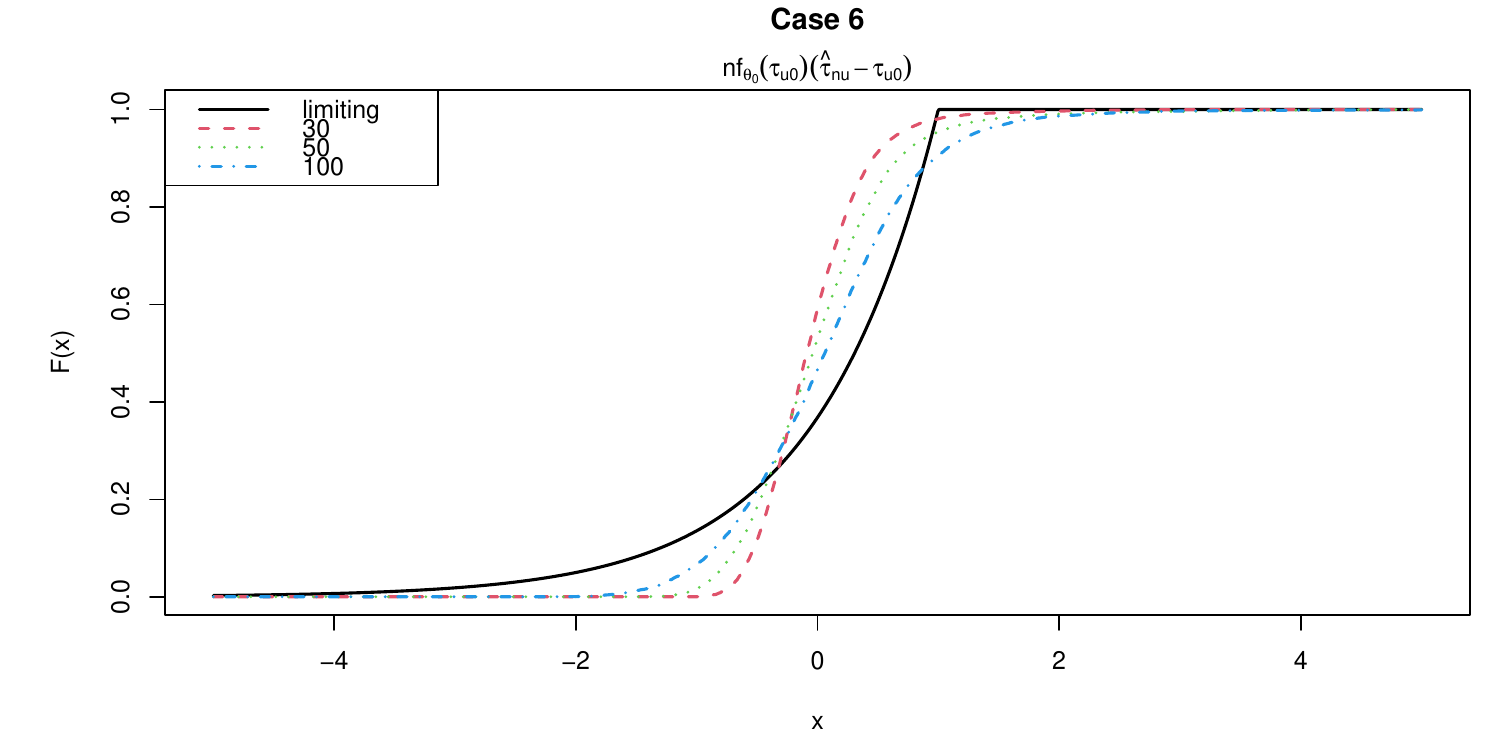}
    \end{subfigure}

    \caption{Results of the simulations illustrating the asymptotic distribution of $nf_{\btheta_0}(\tau_{u0})(\hat{\tau}_{nu}-\tau_{u0})$ as $n\rightarrow \infty$. The black solid, red dashed, green dotted, and blue dot-dashed lines show the limiting CDF and ECDF from 10,000 repetitions of the sample sizes of 30, 50, and 100, respectively.}
    \label{fig:tau_u_asym_dist}
\end{figure}

\subsection{Application} \label{sec:data_ex}

As discussed earlier, the motivation for want to estimate all four of the parameters of the truncated normal distribution lies in an application to forensic science in the form of an open set classification problem in hopes that incorporating truncation into the models will alleviate Lindley's paradox seen in those applications. Here, a simplification of this problem will be shown in the form of a closed set classification problem, highlighting a desirable property using truncation. The idea is that truncation will allow for the likelihood to be zero when a new observation is too far from the exemplars for each class used to fit the models. This truncated normal discriminant analysis is closely related to quadratic discriminant analysis with atypicality scoring  \citep{McLACHLAN-Discriminant}, where the atypicality is built into the sampling models. 

The example uses the iris data \citep{Anderson-Iris1935, Fisher-Iris1936} focusing only on the variable petal length. The classes Setosa and Versicolor are split evenly into training and testing sets, with all of the Virginica added to the testing set. Using the training set, the parameters of the truncated normal distribution and the normal distribution are calculated for both classes. The parameters of the truncated normal distribution were estimated using the proposed algorithm with initial value $\btheta^{(0)}=(\bar{x}, s, x_{n:1}, x_{n:n})$ for each class. Predictions are then made on the testing set using the truncated normal discriminant analysis and QDA with an atypicality cutoff of 0.05. This process is repeated 10,000 times with different training/testing splits on the Setosa and Versicolor classes.

\begin{table}[H]
    \centering
    \caption{Summary of the 10000 confusion matrices in the iris data example}
    \begin{tabular}{|c|c|c|c|c|}
        \hline
        Method & Truth/Prediction & No Training Class & Setosa & Versicolor\\
        \hline
        \multirow{3}{*}{QDA with Atypicality} & Virginica & 32.5 (3.0) & 0 (0) & 17.5 (3.0) \\
        & Setosa & 1.5 (0.7) & 23.5 (0.7) & 0 (0) \\
        & Versicolor & 0.8 (0.5) & 0 (0) & 24.2 (0.5)\\
        \hline
        \multirow{3}{*}{TQDA} & Virginica & 36.3 (5.6) & 0 (0) & 13.7 (5.6)\\
        & Setosa & 0 (0) & 25 (0) & 0 (0) \\
        & Versicolor & 0 (0) & 0 (0) & 25 (0) \\
        \hline
    \end{tabular}
    \label{tab:Data_example}
\end{table}

Table \ref{tab:Data_example} summarize the results of the classification examples on the 10,000 training/testing splits for the QDA with atypicality and truncated normal discriminant analysis (TQDA), respectively. The table gives an element-wise mean with an element-wise standard deviation in parentheses for the confusion matrices. Looking at the first block, it can be seen that QDA with atypicality gives nearly perfect classification on the known training classes, Setosa and Versicolor, while accurately finding that the Virginica observations were from a class other than those seen in the training set. The second block shows similar results for the truncated normal discriminant analysis, where there is now perfect classification on Setosa and Versicolor, and the correct identification that the Virginica observations were not from a class seen in the training set.  These results suggest that the addition of truncation to the models and that extending this idea to the open set classification problem may provide a likelihood ratio that is resistant to Lindley's paradox.

\section{Conclusion}
\label{sec:conclusion}

In this work, an approach to estimate the truncation bounds and the location and scale parameters of the parent normal distribution was presented. The estimation technique utilizes an alternative sampling model, in which the positions of the vector of order statistics from the truncated normal distribution are described within a larger vector of order statistics from the normal distribution by incorporating latent variables. Using this alternative sampling model, the inclusion of the normal distribution in the location-scale family is used to estimate the parameters using a regression on the order statistics and the Filliben plotting positions.

The proposed estimators are consistent, and there is an interesting behavior, where the bound estimators converge faster than the estimates of the location and scale parameters of the parent normal distribution. Also, the simulation study suggests that a larger sample size is needed for harder problems, where the distribution function has the form of a sharper curve or is flat. The iterative algorithm, under the ES framework, allows for a numerical solution to the system of equations. 

The use of the truncated normal distribution in statistical pattern recognition is promising, where the real data example in a univariate classification problem exemplifies that the use of this distribution in the sampling models for the classes provides a built-in detection of whether or not a new observation has arisen from one of the classes seen in the training set. Though currently, the results are only for the univariate case. Further work is needed to extend the methods to the multivariate case to allow for a more thorough comparison with more common classification techniques. Finally, it is our view that extending this modeling technique to hierarchical sampling models, which involve the truncated normal distribution, will allow for a built-in protection to Lindley's Paradox regarding a value of evidence in the form of a likelihood ratio. The introduction of truncation into these models can serve as one method to incorporate the ``pinch of probability" discussed by Degroot.
\newpage
\begin{appendix}

\section{Standard Truncated Normal Distributions}
\label{app:standard_truncated}

Some results from the class of standard truncated normal distributions are used for various results of the algorithm and the resulting estimators. The results used are listed below. 

\subsection{Quantile Function}

The quantile function for the standard truncated normal distribution $\text{TN}(0,1,\tau_l^*, \tau_u^*)$ takes the form

\begin{equation} \label{eqn:std_trunc_quant}
   J_{\btheta}(p) = \Phi^{-1}\left(\Phi(\tau_l^*) + p(\Phi(\tau_u^*)- \Phi(\tau_l^*))\right), \; p\in [0,1],
\end{equation}
where $\Phi$ is the distribution function of the standard normal distribution. Now, some useful properties of $J_{\btheta}$ are as follows.

\begin{property} \label{prop:Lipschitz}
    $J_{\btheta}$ is Lipschitz.
\end{property}

\begin{proof}
Consider $J_{\btheta}$ as defined in \eqref{eqn:std_trunc_quant}. We have
\begin{equation*}
    \frac{\partial J_{\btheta}}{\partial p} = \frac{1}{\phi(\Phi^{-1}\left(\Phi(\tau_l^*) + p(\Phi(\tau_u^*)- \Phi(\tau_l^*))\right))} (\Phi(\tau_u^*)- \Phi(\tau_l^*)).
\end{equation*}
Now as the standard normal density function, $\phi$ is continuous, positive, increasing on $(-\infty,0)$ and decreasing on $(0, \infty)$ and $\Phi^{-1}$ is increasing we have 
\begin{equation*}
    \frac{\partial J_{\btheta}}{\partial p} \le \frac{1}{\phi(\max\{|\tau_l^*|, |\tau_u^*|\})} (\Phi(\tau_u^*)- \Phi(\tau_l^*)) \equiv M(\btheta).
\end{equation*}
Now for $p_1, p_2 \in [0,1],$ by the Mean Value Theorem, there is a $p^*\in(0,1)$ such that
\begin{equation*}
    |J_{\btheta}(p_1) - J_{\btheta}(p_1)| = \frac{\partial J_{\btheta}}{\partial p} \Bigg\vert_{p=p^*} |p_1 - p_2| \le M(\btheta) |p_1 - p_2|.
\end{equation*}
Thus, $J_{\btheta}$ is Lipschitz. Moreover, if $\Theta$ is compact there is a common Lipschitz constant $M = {\sup}_{\btheta \in \Theta}M(\btheta)$
\end{proof}

\begin{property} \label{prop:interval}
The weights in \eqref{eqn:obsest} satisfy $\xi(n, h_l(n,\btheta), h_u(n,\btheta),k) = J_{\btheta}(p)$ for some $p \in [\frac{i-1}{n}, \frac{i}{n}].$ 
\end{property}

\begin{proof}
Recall that $\xi(n, h_l(n,\btheta), h_u(n,\btheta),k)$ is the median of the $\text{Beta}(h_l(n,\btheta) + k, h_u(n,\btheta),k) + n + 1 - k) \overset{d}{=} \text{Beta}(nc_l + k, nc_u + n + 1 - k)$ distributions for $k=1,\dots,n$, where $c_l$ and $c_u$ are defined as in \eqref{eqn:expectations1} and \eqref{eqn:expectations2}. Denote the mode and mean of the $k^{th}$ beta distribution as $m(k)$ and $\mu(k)$, respectively. Now 
$$\frac{c_l + \frac{k-1}{n}}{c_l + c_u + 1 + \frac{1}{n}} < \frac{c_l + \frac{k-1}{n}}{c_l + c_u + 1} < \frac{c_l + \frac{k-1}{n}}{c_l + c_u + 1 - \frac{1}{n}}$$
which implies
\begin{equation} \label{eqn:weight_res1}
    \mu(k-1) < \frac{c_l + \frac{k-1}{n}}{c_l + c_u + 1} < M(k)
\end{equation}
As $c_u + c_l + 1 -\frac{1}{n} > 0$. Similarly, as $c_u>0,$
\begin{equation*}
    \left(c_l + \frac{k-1}{n}\right)(c_l + c_u + 1) + \frac{1}{n}\left(c_l + \frac{k-1}{n}\right) < \left(c_l + \frac{k-1}{n}\right)(c_l + c_u + 1)+  \frac{1}{n}\left(c_l + c_u + 1\right),
\end{equation*}
yielding
\begin{equation*}
    \frac{c_l + \frac{k-1}{n}}{c_l + c_u + 1} < \frac{c_l + \frac{k}{n}}{c_l + c_u + \frac{1}{n}}
\end{equation*}
which gives
\begin{equation}\label{eqn:weight_res2}
    \frac{c_l + \frac{k-1}{n}}{c_l + c_u + 1} < \mu(k).
\end{equation}
Finally, consider a third inequality
\begin{equation*}
    \left(c_l + \frac{k}{n}\right)(c_l + c_u + 1) - \frac{1}{n}\left(c_l + c_u + 1\right) < \left(c_l + \frac{k}{n}\right)(c_l + c_u + 1) -  \frac{1}{n}\left(c_l + \frac{k}{n}\right),
\end{equation*}
again as $c_l + c_u + 1 - \frac{1}{n}$, it follows that
\begin{equation*}
    \frac{c_l + \frac{k-1}{n}}{c_l + c_u + 1 - \frac{1}{n}} < \frac{c_l + \frac{k}{n}}{c_l + c_u + 1}
\end{equation*}
giving
\begin{equation}\label{eqn:weight_res3}
    M(k) < \frac{c_l + \frac{k}{n}}{c_l + c_u + 1}.
\end{equation}
Now, combining \eqref{eqn:weight_res1} , \eqref{eqn:weight_res2} and \eqref{eqn:weight_res3} it has been shown that
$$\frac{c_l + \frac{k-1}{n}}{c_l + c_u + 1} < \mu(k) < \frac{c_l + \frac{k}{n}}{c_l + c_u + 1}$$
and 
$$\frac{c_l + \frac{k-1}{n}}{c_l + c_u + 1} < M(k) < \frac{c_l + \frac{k}{n}}{c_l + c_u + 1}.$$
Now, for all values of $k,$ both parameters of the Beta distributions are greater than one, so the distributions will exhibit a mode-median-mean inequality or a mean-median-mode inequality \cite{Groeneveld-1977, Arab-convex2021}. Then, after re-substituting $c_l$ and $c_u$ and simplifying we have that
\begin{equation*}
    J_{\btheta}\left(\frac{k}{n}\right) = \Phi^{-1}\left( \frac{c_l + \frac{k}{n}}{c_l + c_u + 1} \right).
\end{equation*}
Thus,
\begin{equation*} \label{eqn:weightfunbounds}
    J_{\btheta}\left(\frac{k-1}{n}\right) < \xi(n, h_l(n,\btheta), h_u(n,\btheta),k) < J_{\btheta}\left(\frac{k}{n}\right).
\end{equation*}
Moreover, by the Intermediate Function Theorem, $\xi(n, h_l(n,\btheta), h_u(n,\btheta),k) = J_{\btheta}(p_{n:k}(\btheta))$ for some $p_{n:k}(\btheta)\in[\frac{k-1}{n}, \frac{k}{n}].$
\end{proof}

Note that property \ref{prop:interval} also gives that the sequence of weights is an increasing sequence in $k$. Along with Property \ref{prop:Lipschitz} and compact $\Theta$, it follows that
$$\left| \frac{1}{n}\sum_{i=1}^n \xi(n, h_l(n,\btheta), h_u(n,\btheta),k) - \int_0^1 J_{\btheta}(p)dp \right| = o(n^{-1/2}), \text{ as } n\rightarrow \infty$$
as 
\begin{align*}
   & \left| \frac{1}{n}\sum_{i=1}^n \xi(n, h_l(n,\btheta), h_u(n,\btheta),i) - \int_0^1 J_{\btheta}(p)dp \right| \\
   & =  \left| \sum_{i=1}^n \int_{\frac{i-1}{n}}^\frac{i}{n} (J_{\btheta}(p_{n:i}(\btheta)) -  J_{\btheta}(p))dp \right| \\
   & \le \sum_{i=1}^n \int_{\frac{i-1}{n}}^\frac{i}{n} |J_{\btheta}(p_{n:i}(\btheta)) -  J_{\btheta}(p)| dp \\
   & \le \frac{M}{n},
\end{align*}
and similarly,
$$\left|\frac{1}{n}\sum_{i=1}^n \xi(n, h_l(n,\btheta), h_u(n,\btheta),k)^2 - \int_0^1 J_{\btheta}(p)^2dp \right|=o(n^{-1/2}), \text{ as } n\rightarrow \infty.$$
as
\begin{align*}
   &\left|\frac{1}{n}\sum_{i=1}^n \xi(n, h_l(n,\btheta), h_u(n,\btheta),k)^2 - \int_0^1 J_{\btheta}(p)^2dp \right|  \\
    &\le \sum_{i=1}^n  \int_{\frac{i-1}{n}}^\frac{i}{n} |J_{\btheta}(p_{n:i}(\btheta) - J_{\btheta}(p)||J_{\btheta}(p_{n:i}(\btheta) + J_{\btheta}(p)|dp \\
    & \le \frac{2MB}{n}
\end{align*}
where $B = \sup_\Theta\max\{|\tau_l^*|,|\tau_u^*|\}$.

\subsection{Moments}
The moments of the standard truncated normal distributions are also of interest. Now, the moment generating function of the standard truncated normal distribution is 
\begin{equation} \label{eqn:mgf}
    \varphi(t) = \exp\left\{\frac{t^2}{2}\right\} \frac{\Phi(\tau_u^* - t) - \Phi(\tau_l^*-t)}{\Phi(\tau_u^*) - \Phi(\tau_l^*)},
\end{equation}
where $s(t) = \frac{\Phi(\tau_u^* -t)-\Phi(\tau_l^* -t)}{\Phi(\tau_u^* )-\Phi(\tau_l^*)}.$

Now, the first four moments of the standard truncated normal distribution will be of interest. Consider the first four derivatives of $\varphi$. 

\begin{alignat*}{2}
&\varphi'(t) &&= t\varphi(t) - \exp\{t^2/2\}s'(t) \\
&\varphi''(t) &&= (1-t^2)\varphi(t) + 2t\varphi'(t) + \exp\{t^2/2\}s''(t) \\
&\varphi^{(3)}(t) &&= (t^3 - 3t)\varphi(t) + (3-3t^2)\varphi'(t) + 3t\varphi''(t) +  \exp\{t^2/2\}s^{(3)}(t) \quad \text{and}  \\
&\varphi^{(4)}(t) &&= (6t - t^4 - 3)\varphi(t) + (4t^3 -12t)\varphi'(t) + (6 - 3t - 3t^2)\varphi''(t) \\ 
& && \hspace{0.25in} + 4t\varphi^{(3)}(t) + \exp\{t^2/2\}s^{(4)}(t) .
\end{alignat*}
Now, the derivative of $s$ evaluated at $0$ are
\begin{alignat*}{2}
    &s'(0) &&= -\frac{\phi(\tau_u^*) - \phi(\tau_l^*)}{\Phi(\tau_u^* )-\Phi(\tau_l^*)} \\
    &s''(0) &&= \frac{\phi'(\tau_u^*) - \phi'(\tau_l^*)}{\Phi(\tau_u^* )-\Phi(\tau_l^*)} \\
    &s^{(3)}(0) &&= -\frac{\phi''(\tau_u^*) - \phi''(\tau_l^*)}{\Phi(\tau_u^* )-\Phi(\tau_l^*)} \quad \text{and}\\
    &s^{(4)}(0) &&= \frac{\phi^{(3)}(\tau_u^*) - \phi^{(3)}(\tau_l^*)}{\Phi(\tau_u^* )-\Phi(\tau_l^*)},
\end{alignat*}
which leads to the first four moments of the standard truncated normal distributions
\begin{align*}
    &\alpha_1 \equiv \varphi'(0) = -\frac{\phi(\tau_u^*) - \phi(\tau_l^*)}{\Phi(\tau_u^* )-\Phi(\tau_l^*)} \\
    &\alpha_2 \equiv \varphi''(0) = 1 + \frac{\phi'(\tau_u^*) - \phi'(\tau_l^*)}{\Phi(\tau_u^* )-\Phi(\tau_l^*)} \\
    &\alpha_3 \equiv \varphi^{(3)}(0) = -3\frac{\phi(\tau_u^*) - \phi(\tau_l^*)}{\Phi(\tau_u^* )-\Phi(\tau_l^*)} - \frac{\phi''(\tau_u^*) - \phi''(\tau_l^*)}{\Phi(\tau_u^* )-\Phi(\tau_l^*)} \quad \text{and}  \\
    &\alpha_4 \equiv \varphi^{(4)}(0) = 3 + 6\frac{\phi'(\tau_u^*) - \phi'(\tau_l^*)}{\Phi(\tau_u^* )-\Phi(\tau_l^*)} + \frac{\phi^{(3)}(\tau_u^*) - \phi^{(3)}(\tau_l^*)}{\Phi(\tau_u^* )-\Phi(\tau_l^*)}.
\end{align*}
This leads to the following property.

\begin{property} \label{prop:pos_definite}
The matrix $A =\begin{bmatrix}
    \alpha_2 - \alpha_1^2 & \alpha_3-\alpha_1\alpha_2 \\
    \alpha_3-\alpha_1\alpha_2 & \alpha_4 - \alpha_2^2
\end{bmatrix}$ is positive definite. 
\end{property}

\begin{proof}
Consider $Z\sim\text{TN}(0,1,\tau_l^*, \tau_u^*).$ The covariance matrix of $(Z,Z^2)^T$ is A, and therefore $A$ is positive definite.
\end{proof}

Now, we refer to these distributions as the standard truncated normal distributions due to the following.

\begin{property} \label{prop:location_scale}
Let $Y\sim\text{TN}(\mu, \sigma, \tau_l, \tau_u)$ and $Z\sim\text{TN}(0,1,\tau_l^*, \tau_u^*)$, where $\tau_l^*=\frac{\tau_l-\mu}{\sigma}$ and  $\tau_u^*=\frac{\tau_u-\mu}{\sigma}$. Then $Y\overset{d}{=}\mu+\sigma Z$. 
\end{property}

\begin{proof}
Consider the CDF of $Z$, $F_Z(t) = \frac{\Phi(t) - \Phi(\tau_l^*)}{\Phi(\tau_u^*) - \Phi(\tau_l^*)}.$ Then the CDF of $W=\mu + \sigma Z$  is 

\begin{equation}
    F_W(t) = F_Z\left(\frac{t-\mu}{\sigma}\right) = \frac{\Phi\left(\frac{t-\mu}{\sigma}\right) - \Phi(\tau_l^*)}{\Phi(\tau_u^*) - \Phi(\tau_l^*)}
\end{equation}
which is the CDF of the $\text{TN}(\mu, \sigma, \tau_l, \tau_u)$ distribution. Thus, $Y\overset{d}{=}\mu+\sigma Z.$
\end{proof}

\section{Jacobian Derivation and Properties}
\label{app:derivative}
Multiple lemmas and theorems utilized the limiting form of the sequence of the random functions, $\Psi$, which had the form,
$$\Psi(\btheta)=\begin{pmatrix}
        \mu_0 + \sigma_0\alpha_{01} - \mu -\sigma\alpha_1 \\
        \mu(J_{\btheta},F_{\btheta_0}) - \mu\alpha_1 - \sigma\alpha_2 \\
        \tau_{l0} - \tau_l \\
        \tau_{u0} - \tau_u
    \end{pmatrix}.$$
Recall that the parameter of interest is $\btheta = (\mu, \sigma, \tau_l, \tau_u)^t,$ reparameterizing to $\boldeta = (\mu, \sigma, \tau_l^*, \tau_u^*)^t$, where $\tau_l^* = \frac{\tau_l-\mu}{\sigma}$ and $\tau_u^* = \frac{\tau_u-\mu}{\sigma}$, it can be seen that
\begin{align}
    \Psi(\btheta)&=\begin{pmatrix}
        \mu_0 + \sigma_0\alpha_{01} - \mu -\sigma\alpha_1 \\
        \mu(J_{\btheta},F_{\btheta_0}) - \mu\alpha_1 - \sigma\alpha_2 \\
        \tau_{l0} - \tau_l \\
        \tau_{u0} - \tau_u
    \end{pmatrix} \\
    &=\begin{pmatrix}
        \mu_0 + \sigma_0\alpha_{01} - \mu -\sigma\alpha_1 \\
        \mu(J_{\btheta},F_{\btheta_0}) - \mu\alpha_1 - \sigma\alpha_2 \\
        \mu_0 + \sigma_0\tau_{l0}^* - \mu - \sigma\tau_l^* \\
        \mu_0 + \sigma_0\tau_{u0}^* - \mu - \sigma\tau_u^*
    \end{pmatrix} \equiv \Psi^*(\boldsymbol{\eta}(\btheta)).
\end{align}

Now the Jacobian matrix takes the form $$\frac{\partial\Psi^*}{\partial \boldsymbol{\eta}}(\boldsymbol{\eta}) = \begin{bmatrix}
    -1 & -\alpha_1 &-\sigma\frac{\partial\alpha_1}{\partial\tau_u^*} & \sigma\frac{\partial\alpha_1}{\partial\tau_l^*} \\
    -\alpha_1 & -\alpha_2 & \frac{\partial\mu(J_{\btheta},F_{\btheta_0})}{\partial \tau_u^*} - \mu \frac{\partial\alpha_1}{\partial\tau_u^*} - \sigma \frac{\partial\alpha_2}{\partial\tau_u^*} & \frac{\partial\mu(J_{\btheta},F_{\btheta_0})}{\partial \tau_l^*} - \mu \frac{\partial\alpha_1}{\partial\tau_l^*} - \sigma \frac{\partial\alpha_2}{\partial\tau_l^*} \\
    -1 & -\tau_l^* & -\sigma & 0 \\
    -1 & -\tau_u^* & 0 & -\sigma
\end{bmatrix},$$
as $\alpha_1$ and $\alpha_2$ only depend on $\tau_l^*$ and $\tau_u^*$. Now, by the Leibniz rule, the partial derivatives of $\mu(J_{\btheta}, F_{\btheta_0})$ are $$\frac{\partial\mu(J_{\btheta}, F_{\btheta_0})}{\partial \tau_u^*} = \int_0^1\frac{\phi(\tau_u^*)tF^{-1}_{\btheta_0}(t)}{\phi(J_{\btheta}(t))}dt,$$ and $$\frac{\partial\mu(J_{\btheta}, F_{\btheta_0})}{\partial \tau_l^*} = \int_0^1\frac{\phi(\tau_l^*)(1-t)F^{-1}_{\btheta_0}(t)}{\phi(J_{\btheta}(t))}dt.$$
The partial derivatives of $\alpha_1$ and $\alpha_2$ are
\begin{align*}
    & \frac{\partial \alpha_1}{\partial \tau_u^*} = \frac{-\phi'(\tau_u^*)(\Phi(\tau_u^*) - \Phi(\tau_l^*)) + \phi(\tau_u^*)(\phi(\tau_u^*) - \phi(\tau_l^*))}{(\Phi(\tau_u^*) - \Phi(\tau_l^*))^2}, \\
    & \frac{\partial \alpha_1}{\partial \tau_l^*} = \frac{\phi'(\tau_l^*)(\Phi(\tau_u^*) - \Phi(\tau_l^*)) - \phi(\tau_u^*)(\phi(\tau_u^*) - \phi(\tau_l^*))}{(\Phi(\tau_u^*) - \Phi(\tau_l^*))^2}, \\
    & \frac{\partial \alpha_2}{\partial \tau_u^*} = \frac{\phi''(\tau_u^*)(\Phi(\tau_u^*) - \Phi(\tau_l^*)) - \phi(\tau_u^*)(\phi'(\tau_u^*) - \phi'(\tau_l^*))}{(\Phi(\tau_u^*) - \Phi(\tau_l^*))^2}, \quad \text{and}\\
    & \frac{\partial \alpha_2}{\partial \tau_l^*} = \frac{-\phi''(\tau_l^*)(\Phi(\tau_u^*) - \Phi(\tau_l^*)) + \phi(\tau_l^*)(\phi'(\tau_u^*) - \phi'(\tau_l^*))}{(\Phi(\tau_u^*) - \Phi(\tau_l^*))^2}.
\end{align*}

Now, looking at $\frac{\partial\Psi^*}{\partial \boldsymbol{\eta}}(\boldsymbol{\eta}_0)$ we have 

$$\frac{\partial\Psi^*}{\partial \boldsymbol{\eta}}(\boldsymbol{\eta}_0) = \begin{bmatrix}
     \boldsymbol{A}_{11} & \boldsymbol{A}_{12} \\
     \boldsymbol{A}_{21} & \boldsymbol{A}_{22}
\end{bmatrix}.$$

where 

$\boldsymbol{A}_{11} = \begin{bmatrix}
    -1 & -\alpha_{1_0} \\
    -\alpha_{1_0} & -\alpha_{2_0}
\end{bmatrix}$, $\boldsymbol{A}_{12} = \begin{bmatrix}
    a_{12,11} & a_{12,12} \\
     a_{12,21} & a_{12,22} \\
\end{bmatrix},$ $\boldsymbol{A}_{21} = \begin{bmatrix}
    -1 & -\tau_{l_0}^* \\
    -1 & -\tau_{u_0}^* 
\end{bmatrix}$, and $\boldsymbol{A}_{22} = \begin{bmatrix}
    -\sigma_0 & 0 \\
    0 & -\sigma_0
\end{bmatrix},$
where
\begin{align*}
    a_{12,11} &= -\sigma_0\frac{-\phi'(\tau_{u_0}^*)(\Phi(\tau_{u_0}^*) - \Phi(\tau_{l_0}^*)) + \phi(\tau_{u_0}^*)(\phi(\tau_{u_0}^*) - \phi(\tau_{l_0}^*))}{(\Phi(\tau_{u_0}^*) - \Phi(\tau_{l_0}^*))^2} \\
    a_{12,12} &= -\sigma_0\frac{\phi'(\tau_{l_0}^*)(\Phi(\tau_{u_0}^*) - \Phi(\tau_{l_0}^*)) - \phi(\tau_{l_0}^*)(\phi(\tau_{u_0}^*) - \phi(\tau_{l_0}^*))}{(\Phi(\tau_{u_0}^*) - \Phi(\tau_{l_0}^*))^2} \\
    a_{12,21} &= \frac{-\sigma_0}{2} \frac{\phi''(\tau_{u_0}^*)}{\Phi(\tau_{u_0}^*) - \Phi(\tau_{l_0}^*)} + \frac{\sigma_0}{2}\frac{\phi(\tau_{u_0}^*)(\phi'(\tau_{u_0}^*) - \phi'(\tau_{l_0}^*))}{(\Phi(\tau_{u_0}^*) - \Phi(\tau_{l_0}^*))^2} \\
    a_{12,22} &= \frac{\sigma_0}{2} \frac{\phi''(\tau_{l_0}^*)}{\Phi(\tau_{u_0}^*) - \Phi(\tau_{l_0}^*)} - \frac{\sigma_0}{2}\frac{\phi(\tau_{l_0}^*)(\phi'(\tau_{u_0}^*) - \phi'(\tau_{l_0}^*))}{(\Phi(\tau_{u_0}^*) - \Phi(\tau_{l_0}^*))^2}.
\end{align*}
Now, $\frac{\partial\boldeta}{\partial \boldsymbol{\btheta}}(\btheta_0)$ takes the form 
$$\frac{\partial\boldeta}{\partial \boldsymbol{\btheta}}(\btheta_0) = \begin{bmatrix}
   I & 0 \\
   B & \frac{1}{\sigma_0} I
\end{bmatrix},$$
where
$B = \begin{bmatrix}
    \frac{-1}{\sigma_0} & \frac{-1}{\sigma_0}\tau_{l0}^* \\
    \frac{-1}{\sigma_0} & \frac{-1}{\sigma_0}\tau_{u0}^*
\end{bmatrix}.$ Thus, by the chain rule, 

$\frac{\partial\Psi}{\partial \btheta}(\btheta_0) = \frac{\partial\Psi^*}{\partial \boldsymbol{\eta}}(\boldsymbol{\eta}(\btheta_0))\frac{\partial\boldeta}{\partial \boldsymbol{\btheta}}(\btheta_0) = \begin{bmatrix}
    A_{11} + A_{12}B & \frac{1}{\sigma_0}A_{12} \\
    A_{21} + A_{22}B & \frac{1}{\sigma_0}A_{22}
\end{bmatrix}.$

Looking at these matrices, it can be seen that
$A_{21} + A_{22}B = 0,$ and $\frac{1}{\sigma_0}A_{22}=-I.$ The last matrix of importance is $$A_{11} + A_{12}B = \begin{bmatrix}
    -(\alpha_{2_0} - \alpha_{1_0}^2) & -(\alpha_{3_0} - \alpha_{2_0}\alpha_{1_0}) \\
    -\frac{1}{2}(\alpha_{3_0} - \alpha_{2_0}\alpha_{1_0}) &  -\frac{1}{2}(\alpha_{4_0} - \alpha_{2_0}^2)
\end{bmatrix}.$$

Now, using properties of block matrices, it follows that 
\begin{align*}
    \det\left(\frac{\partial\Psi}{\partial \btheta}(\btheta_0)\right) = \det(A_{11} + A_{12}B) \det (-I) = \frac{1}{2} \det \left(\begin{bmatrix}
    (\alpha_{2_0} - \alpha_{1_0}^2) & (\alpha_{3_0} - \alpha_{2_0}\alpha_{1_0}) \\
    (\alpha_{3_0} - \alpha_{2_0}\alpha_{1_0}) &  (\alpha_{4_0} - \alpha_{2_0}^2)
\end{bmatrix} \right) > 0
\end{align*}
by Property \ref{prop:pos_definite}. Also, by properties of block matrices, it can be seen that the upper left block of $\frac{\partial\Psi}{\partial \btheta}(\btheta_0)^{-1}$ is $(A_{11} + A_{12}B)^{-1}.$

\section{Covariance Derivation}
\label{app:covariance}

In this section, different integrals are calculated for different covariance terms.

\begin{property} \label{prop:cov_mat}
    Derivation of the covariance matrix $$\bSigma =
    \begin{pmatrix}
    \mathbb{V}(X_i) & \mathbb{E}\left(\frac{\sigma_0}{2}\left(\frac{X_i - \mu_0}{\sigma_0}\right)^2 \right) \mathbb{E}(X_i) \\
    \mathbb{E}\left(\frac{\sigma_0}{2}\left(\frac{X_i - \mu_0}{\sigma_0}\right)^2 \right) \mathbb{E}(X_i) & \mathbb{V}\left(\frac{\sigma_0}{2}\left(\frac{X_i - \mu_0}{\sigma_0}\right)^2 \right)
\end{pmatrix}$$
\end{property}

\begin{proof}
In Lemma \ref{lemma:L_stat_mvn_dist}, the form of the covariance matrix $\Sigma = \begin{pmatrix}
        \sigma_{11}(\btheta) & \sigma_{12}(\btheta) \\  
        \sigma_{12}(\btheta) & \sigma_{22}(\btheta)
    \end{pmatrix}$
for 
$$\sigma_{11}(\btheta)=\int_{\tau_{l0}}^{\tau_{u0}}\int_{\tau_{l0}}^{\tau_{u0}} [ F_{\btheta_0}(\min(x,y)) -F_{\btheta_0}(x)F_{\btheta_0}(y) ] dx dy $$
$$ \sigma_{12}(\btheta)=\int_{\tau_{l0}}^{\tau_{u0}}\int_{\tau_{l0}}^{\tau_{u0}} J_{\btheta}(F_{\btheta_0}(y)) [ F_{\btheta_0}(\min(x,y)) -F_{\btheta_0}(x)F_{\btheta_0}(y) ] dx dy \quad \text{and}$$
$$ \sigma_{22}(\btheta)=\int_{\tau_{l0}}^{\tau_{u0}}\int_{\tau_{l0}}^{\tau_{u0}} J_{\btheta}(F_{\btheta_0}(x))J_{\btheta}(F_{\btheta_0}(y)) [F_{\btheta_0}(\min(x,y)) -F_{\btheta_0}(x)F_{\btheta_0}(y) ] dx dy$$
was given to be $\Sigma = \begin{pmatrix}
        \sigma_{11}(\btheta) & \sigma_{12}(\btheta) \\  
        \sigma_{12}(\btheta) & \sigma_{22}(\btheta)
    \end{pmatrix}.$
To begin the derivation, consider the  $\sigma_{11}(\btheta)$. Working with $\min(x,y)$ it follows that
\begin{align*}
\sigma_{11}(\btheta_0) &= \int_{\tau_{l0}}^{\tau_{u0}}\int_{\tau_{l0}}^{\tau_{u0}} [ F_{\btheta_0}(\min(x,y)) -F_{\btheta_0}(x)F_{\btheta_0}(y) ]dxdy  \\
&= \int_{\tau_{l0}}^{\tau_{u0}}\int_{\tau_{l0}}^{y} F_{\btheta_0}(x)(1-F_{\btheta_0}(y))dx + \int_{\tau_{l0}}^{\tau_{u0}}\int_{\tau_{l0}}^{y} (1-F_{\btheta_0}(x))F_{\btheta_0}(y)dx
\end{align*}
Using integration by parts, this can be written as
\begin{align*}
&=  \int_{\tau_{l0}}^{\tau_{u0}}(1-F_{\btheta_0}(y))[xF(x)\vert_{\tau_{l0}}^y-\int_{\tau_{l0}}^{y} xf(x)dx] \\
& \hspace{0.5in}+ F_{\btheta_0}(y)[x(1-F_{\btheta_0}(x))|_y^{\tau_{u0}}-\int_y^{\tau_{u0}} xf(x)dx]dy
\end{align*}
As $F_{\btheta_0}(\tau_{l0})=0$ and $F_{\btheta_0}(\tau_{u0})=1$ this simplifies to
\begin{align*}
    &= \int_{\tau_{l0}}^{\tau_{u0}} F_{\btheta_0}(y) \int_{\tau_{0}}^{\tau_{0}}xf(x)dxdy - \int_{\tau_{l0}}^{\tau_{u0}} \int_{\tau_{0}}^yxf(x)dxdy
\end{align*}
Using integration by parts and the First Fundamental Theorem of Calculus, this yields
\begin{align*}
   &=\mathbb{E}(X)y(F_{\btheta_0}(y)\vert_{\tau_{l0}}^{\tau_{u0}} -  \mathbb{E}(X) \int_{\tau_{l0}}^{\tau_{u0}}yf_{\btheta_0}(y)dy \\
   & \hspace{0.5in}-y\int_{\tau_{l0}}^yxf_{\btheta_0}(x)dx\vert_{\tau_{l0}}^{\tau_{u0}} + \int_{\tau_{l0}}^{\tau_{u0}}y^2f_{\btheta_0}(y)dy
\end{align*}
Finally simplifying gives
\begin{align*}
&= \mathbb{E}(X^2)-\mathbb{E}(X)^2 = \mathbb{V}(X).
\end{align*}

Next, consider $\sigma_{22}(\btheta)$.
\begin{align*}
\sigma_{22}(\btheta_0) &= \int_{\tau_{l0}}^{\tau_{u0}}\int_{\tau_{l0}}^{\tau_{u0}} J_{\btheta_0}(F_{\btheta_0}(x))J_{\btheta_0}(F_{\btheta_0}(y)) [ F_{\btheta_0}(\min(x,y)) -F_{\btheta_0}(x)F_{\btheta_0}(y) ] dx dy \\
\end{align*}
Looking at the inside integral, using $\min(x,y)$ and the relation of $J_{\btheta}$ relating to the quantile function as given in Appendix \ref{app:standard_truncated}, it follows that
\begin{align*}
    &\int_{\tau_{l0}}^{\tau_{u0}} J_{\btheta_0}(F_{\btheta_0}(x))J_{\btheta_0}(F_{\btheta_0}(y)) [ F_{\btheta_0}(\min(x,y)) -F_{\btheta_0}(x)F_{\btheta_0}(y) ] dx\\
    &= \int_{\tau_{l0}}^{y} \frac{x-\mu_0}{\sigma_0} \frac{y-\mu_0}{\sigma_0} F_{\btheta_0}(x) (1-F_{\btheta_0}(y))dx \\ 
    & \hspace{0.5in} +  \int_{y}^{\tau_{u0}} \frac{x-\mu_0}{\sigma_0} \frac{y-\mu_0}{\sigma_0} (1-F_{\btheta_0}(x)) (F_{\btheta_0}(y))dx \\
\end{align*}
Using integration by parts, and again using the fact that $F_{\btheta_0}(\tau_{l0})=0$ and $F_{\btheta_0}(\tau_{u0})=1.$ This simplifies to 
\begin{align*}
    &= \frac{y-\mu_0}{\sigma_0} F_{\btheta_0}(y) \frac{\sigma_0}{2} \mathbb{E}\left(\left(\frac{X-\mu_0}{\sigma_0} \right)^2 \right) - \frac{y-\mu_0}{\sigma_0}\frac{\sigma_0}{2}\int_{\tau_{l0}}^y \left(\frac{x-\mu_0}{\sigma_0}\right)^2 f_{\btheta_0}(x)dx
\end{align*}
Putting back into the double integral, it can be seen that
\begin{align*}
     \sigma_{22}(\btheta_0) &= \frac{\sigma_0}{2}\mathbb{E}\left(\left(\frac{X-\mu_0}{\sigma_0} \right)^2 \right) \int_{\tau_{l0}}^{\tau_{u0}} \frac{y-\mu_0}{\sigma_0}F_{\btheta_0}(y)dy 
     \\ - & \int_{\tau_{l0}}^{\tau_{u0}} \frac{y-\mu_0}{\sigma_0} \frac{\sigma}{2} \int_{\tau_{l0}}^y  \left(\frac{x-\mu_0}{\sigma_0}\right)^2 f_{\btheta_0}(x)dxdy
\end{align*}
Using integration by parts and the First Fundamental Theorem of Calculus, 
\begin{align*}
    &= \left(\frac{\sigma_0}{2}\right)^2 \mathbb{E}\left(\left(\frac{X-\mu_0}{\sigma_0} \right)^4 \right) -  \left(\frac{\sigma_0}{2}\right)^2 \mathbb{E}\left(\left(\frac{X-\mu_0}{\sigma_0} \right)^2 \right)^2 = \mathbb{V}\left(\frac{\sigma}{2}\left(\frac{X-\mu_0}{\sigma_0} \right)^2 \right) 
\end{align*}.

Finally, looking at $\sigma_{12}(\btheta_0)$, the derivation follows closely to the previous two. Using the definition of $J_{\btheta_0}$ and $\min(x,y)$,
\begin{align*}
\sigma_{12}&=\int_{\tau_{l0}}^{\tau_{u0}}\int_{\tau_{l0}}^{\tau_{u0}}J_{\btheta_0}(F_{\btheta_0}(y)) [ F_{\btheta_0}(\min(x,y)) -F_{\btheta_0}(x)F_{\btheta_0}(y) ] dy dx \\
    &= \int_{\tau_{l0}}^{\tau_{u0}}\int_{\tau_{l0}}^{\tau_{u0}} \frac{y-\mu_0}{\sigma_0} [ F_{\btheta_0}(\min(x,y)) -F_{\btheta_0}(x)F_{\btheta_0}(y) ] dy dx
\end{align*}
Using integration by parts and simplifying
\begin{align*}
    &= \mathbb{E}(X)\int_{\tau_{l0}}^{\tau_{u0}}\frac{y-\mu_0}{\sigma_0}F_{\btheta_0}(y)dy -  \int_{\tau_{l0}}^{\tau_{u0}}\frac{y-\mu_0}{\sigma_0}\int_{\tau_{l0}}^y x f_{\btheta_0}(x)dxdy
\end{align*}
Again using integration by parts and the First Fundamental Theorem of Calculus,
\begin{align*}
   &=  \mathbb{E}\left( \frac{\sigma_0}{2}X\left(\frac{X-\mu_0}{\sigma_0}\right)^2\right) - \mathbb{E} \left(\frac{\sigma_0}{2}\left(\frac{X-\mu_0}{\sigma_0}\right)^2\right) \mathbb{E}(X).
\end{align*}
\end{proof}

\end{appendix}

\bibliographystyle{unsrt}
\bibliography{refs}     

\end{document}